%% file: main.tex
\documentclass[11pt]{article}
\usepackage{amsmath,amsfonts,amsthm,amssymb,xcolor}
\usepackage{thm-restate}
\usepackage{fancybox}
\usepackage{mathtools}
\usepackage{pdfsync}
\usepackage{hyperref}
\usepackage{nicefrac}
\usepackage{fullpage}
\usepackage{xspace}
\usepackage{appendix}
\usepackage{tikz}
\usepackage[procnumbered,ruled,boxed,linesnumbered]{algorithm2e}
\DontPrintSemicolon
\SetKw{KwAnd}{and}
\SetProcNameSty{textsc}
\SetFuncSty{textsc}
\SetKwRepeat{Do}{do}{while}

\usepackage{tikz}
\usepackage{appendix}
\usepackage{enumerate}
\usepackage{bm}
\usepackage{subeqnarray}
\usepackage{cases}
\usepackage{pifont}
\usepackage{graphicx}
\usepackage{epstopdf}
\usepackage{makecell}
\usepackage{blkarray}

\usepackage{footnote}
\makesavenoteenv{tabular}
\makesavenoteenv{table}
\usepackage{threeparttable}

\newtheorem{theorem}{Theorem}[section]
\newtheorem{corollary}[theorem]{Corollary}
\newtheorem{lemma}[theorem]{Lemma}

\newtheorem{fact}[theorem]{Fact}

\theoremstyle{definition}
\newtheorem{definition}[theorem]{Definition}
\newtheorem{remark}[theorem]{Remark}

\newenvironment{fminipage}%
  {\begin{Sbox}\begin{minipage}}%
  {\end{minipage}\end{Sbox}\fbox{\TheSbox}}

\let\originalleft\left
\let\originalright\right
\renewcommand{\left}{\mathopen{}\mathclose\bgroup\originalleft}
  \renewcommand{\right}{\aftergroup\egroup\originalright}

\def\eq#1{\begin{equation*}\begin{split}#1\end{split}\end{equation*}}
\def\eql#1#2{\begin{equation}{#1}\begin{split}#2\end{split}\end{equation}}

\def\al#1{\begin{align}
             #1
            \end{align}}

\def\comeq#1{\stackrel{\mathrm{#1}}{=}}

\def\pleq{\preccurlyeq}
\def\pgeq{\succcurlyeq}

\def\defeq{\stackrel{\mathrm{def}}{=}}

\def\pr#1{\left( #1 \right ) }
\def\br#1{\left[ #1 \right ] }
\def\dr#1{\left\{#1\right\}}

\def\abs#1{\left|#1  \right|}

\def\calG{\mathcal{G}}

\def\calL{\mathcal{L}}

\def\calT{\mathcal{T}}

\def\calN{\mathcal{N}}

\newcommand\PPi{\boldsymbol{\Pi}}

\def\aa{\pmb{\mathit{a}}}
\newcommand\bb{\boldsymbol{\mathit{b}}}

\newcommand\dd{\boldsymbol{\mathit{d}}}
\newcommand\ee{\boldsymbol{\mathit{e}}}

\newcommand\vv{\boldsymbol{\mathit{v}}}
\newcommand\ww{\boldsymbol{\mathit{w}}}
\newcommand\yy{\boldsymbol{\mathit{y}}}
\newcommand\zz{\boldsymbol{\mathit{z}}}
\newcommand\xx{\boldsymbol{\mathit{x}}}

\renewcommand\AA{\boldsymbol{\mathit{A}}}
\newcommand\BB{\boldsymbol{\mathit{B}}}
\newcommand\CC{\boldsymbol{\mathit{C}}}
\newcommand\DD{\boldsymbol{\mathit{D}}}
\newcommand\EE{\boldsymbol{\mathit{E}}}

\newcommand\GG{\boldsymbol{\mathit{G}}}
\newcommand\II{\boldsymbol{\mathit{I}}}

\newcommand\NN{\boldsymbol{\mathit{N}}}
\newcommand\MM{\boldsymbol{\mathit{M}}}
\newcommand\LL{\boldsymbol{\mathit{L}}}
\newcommand\PP{\boldsymbol{\mathit{P}}}
\newcommand\QQ{\boldsymbol{\mathit{Q}}}
\newcommand\RR{\boldsymbol{\mathit{R}}}
\renewcommand\SS{\boldsymbol{\mathit{S}}}

\newcommand\UU{\boldsymbol{\mathit{U}}}
\newcommand\WW{\boldsymbol{\mathit{W}}}

\newcommand\XX{\boldsymbol{\mathit{X}}}
\newcommand\YY{\boldsymbol{\mathit{Y}}}
\newcommand\ZZ{\boldsymbol{\mathit{Z}}}

\def\Otil#1{\tilde{O}\left(#1\right)}

\def\nm#1{\left\| #1 \right\|}
\def\no#1{\left\|#1\right\|_1}
\def\nt#1{\left\| #1 \right\|_2}
\def\ni#1{\left\|#1\right\|_{\infty}}

\def\la{\lambda}

\def\tp{^\top}

\def \nnz#1{\textbf{nnz}\left(#1\right)}

\def \Lt#1{\LL^{(#1)}}

\def \Gt#1{{\GG^{(#1)}}}

\def \xt#1{\xx^{(#1)}}

\def \Zt#1{\ZZ^{(#1)}}

\def \Ltt#1{{\LL_*^{(#1)}}}

\def \bt#1{\bb^{(#1)}}

\def \Ut#1{\UU^{(#1)}}

\newcommand \arrlf{\leftarrow}
\newcommand \arr{\rightarrow}

\newcommand \eps{\epsilon}

\def \Diag#1{\textbf{Diag}\left(#1\right)}

\newcommand{\zero}{\mathbf{0}}
\newcommand{\one}{\mathbf{1}}

\newcommand \inv{^{-1}}

\def\MatSize#1#2{\mathbb{R}^{#1\times#2}}
\def\MS#1#2{\mathbb{R}^{#1\times#2}}

\newcommand \Real{\mathbb{R}}

\def \poly#1{\rm poly\left(#1\right)}

\def \calLt#1{\calL^{(#1)}}

\def\At#1{\AA^{\pr{#1}}}

\def\Ut#1{\UU^{\pr{#1}}}
\def\Wt#1{\WW^{\pr{#1}}}

\newcommand\sleq{\subseteq}
\newcommand\sgeq{\supseteq}
\newcommand\arl{\arrlf}

\input{newcommand}

\begin{document}

\title{
Sparsified Block Elimination for Directed Laplacians
}

\author{
    Richard Peng\\
    University of Waterloo\\
    \texttt{y5peng@uwaterloo.ca}
    \and
    Zhuoqing Song  \\
    Fudan University \\
    \texttt{zqsong19@fudan.edu.cn}
}

\date{}

\maketitle
\thispagestyle{empty}

\begin{abstract}
We show that the sparsified block elimination algorithm for solving undirected Laplacian linear systems from [Kyng-Lee-Peng-Sachdeva-Spielman STOC'16] directly works for directed Laplacians. Given access to a sparsification algorithm that, on graphs with $n$ vertices and $m$ edges, takes time $\mathcal{T}_{\rm S}(m)$ to output a sparsifier with $\mathcal{N}_{\rm S}(n)$ edges, our algorithm solves a directed Eulerian system on $n$ vertices and $m$ edges to $\epsilon$ relative accuracy in time
$$ O(\mathcal{T}_{\rm S}(m) + {\mathcal{N}_{\rm S}(n)\log {n}\log(n/\epsilon)}) + \tilde{O}(\mathcal{T}_{\rm S}(\mathcal{N}_{\rm S}(n)) \log n), $$
where the $\tilde{O}(\cdot)$ notation hides $\log\log(n)$ factors. By previous results, this implies improved runtimes for linear systems in strongly connected directed graphs, PageRank matrices, and asymmetric M-matrices. When combined with slower constructions of smaller Eulerian sparsifiers based on short cycle decompositions, it also gives a solver that runs in $O(n \log^{5}n \log(n / \epsilon))$ time after $O(n^2 \log^{O(1)} n)$ pre-processing. At the core of our analyses are constructions of augmented matrices whose Schur complements encode error matrices.
\end{abstract}


\section{Introduction}

The design of efficient solvers for systems of linear equations
in graph Laplacian matrices and their extensions has
been a highly fruitful topic in algorithms.
Laplacian matrices directly correspond to undirected graphs:
off-diagonal entries are negations of edge weights,
while the diagonal entries contain weighted degrees.
Solvers for Laplacian matrices
led to breakthroughs in fundamental problems
in combinatorial optimization.
Tools developed during such studies have in turn
influenced data structures,
randomized numerical linear algebra,
scientific computing,
and network science~\cite{S10,T10}.

An important direction in this Laplacian paradigm of
designing graph algorithms is extending tools developed
for undirected Laplacian matrices to directed graphs.
Here a perspective from random walks and Markov chains
leads to directed Laplacian matrices~\cite{CKPPSV16}.
Such matrices have directed edge weights in off-diagonal
entries, and weighted out-degrees on diagonals.
In contrast to solving linear systems in undirected Laplacians,
solving linear systems in directed Laplacians is
significantly less well-understood.
Almost-linear time~\cite{cohen2017almost}
and nearly-linear time solvers~\cite{cohen2018solving}
were developed very recently,
and involve many more moving pieces.

In particular, the nearly-linear time algorithm
from~\cite{cohen2018solving} combined
block Gaussian elimination with single variables/vertex
elimination, analyzed using matrix Martingales.
In contrast, for undirected Laplacians,
both block elimination~\cite{kyng2016sparsified}
or matrix Martingales~\cite{KS16} can give different
nearly-linear time solver algorithms, and there also
exists more combinatorial approaches~\cite{KOSZ13}.
In this paper, we simplify this picture for directed
Laplacian solvers by providing an analog of
the sparsified Cholesky/multi-grid solver from~\cite{kyng2016sparsified}.
This algorithm's running time is close to
the limit of sparsification based algorithms: the running
time of invoking a sparsification routine on its own output.
Formally, we show:
\newcommand\StrcEulerianLapSolveBlk{
    Given a \strc\ Eulerian Laplacian $\LL\in\MS{n}{n}$ and an error parameter $\eps\in (0, 1)$, we can process it in time  $O\pr{\TSE\pr{m, n, 1} } + \Otil{\TSE\pr{\NSE\pr{n, 1}, n, 1}\log n  }$ so that, with high probability,
    given any vector $\bb\in \Real^n$ with $\bb\perp \one$,
    we can  compute  a vector $\xx\in \Real^n$ in time
    $O\pr{\NSE\pr{n, 1}\log n \log\pr{n/\eps}}$  such that
  \eq{
    \nA{\U{\LL}}{\xx - \LL\dg\bb} \leq \eps \nA{\U{\LL}}{\LL\dg \bb},
  }
  where $\U{\LL} = \pr{\LL + \LL\tp}/2$.
}
\begin{theorem}\label{thm:TSENSEsolver1}
  \StrcEulerianLapSolveBlk

\end{theorem}

This result improves the at least $\Omega(\log^{5}n)$
factor overhead upon sparsification of the previous
nearly-linear time directed Laplacian solver~\cite{cohen2018solving},
and is analogous to the current best overheads for
sparsification based solvers for
undirected Laplacians~\cite{kyng2016sparsified}.
From the existence of sparsifiers of size $O(n\log^4 n\eps^{-2})$~\cite{CGPSSW18},
we also obtain the existence of $O(n \log^{5}n\log(n/\eps))$ time solver
routines that require quadratic time preprocessing to compute.
As with other improved solvers for directed Laplacians
our improvements directly applies to applications
of such solvers,
including random walk related quantities~\cite{CKPPSV16},
as well as PageRank / Perron-Frobenius vectors~\cite{AJSS19}.

Our result complements recent developments
of better sparsifiers of Eulerian Laplacians~\cite{CGPSSW18,LSY19,PY19}.
By analyzing a pseudocode that's entirely analogous
to the undirected block-elimination algorithm from~\cite{kyng2016sparsified},
we narrow the gap between Laplacian solvers for
directed and undirected graphs.
Our result also emphasizes the need for better directed
sparsification routines.
While there is a rich literature on
undirected sparsification~\cite{BSST13},
the current best directed sparsification algorithms rely on expander
decompositions, so have rather large logarithmic factor overheads.
We discuss such bounds in detail in Appendix~\ref{sec:sparsify}.

Finally, our analysis of this more direct algorithm
require better understanding the accumulation errors
in Eulerian Laplacians and their partially eliminated
states, known as Schur Complements.
It was observed in~\cite{cohen2018solving} that these
objects are significantly less robust than their
undirected analogs.
Our analysis of these objects rely on augmentations
of matrices: constructing larger matrices
whose Schur complements correspond to the final objects
we wish to approximate, and bounding errors on these larger matrices
instead.
This approach has roots in symbolic computation,
and can be viewed as a generalization of
low-rank perturbation formulas such as Sherman-Morrison-Woodbury~\cite{B21}.
We believe both this algebraic technique,
and the additional robustness properties of
directed Schur Complements we show,
may be of independent interest.

\subsection{Related Works}

Directed Laplacian matrices arise in problems related
to directed random walks / non-reversible Markov chains,
such as computations of stationary distributions,
hitting times and escape probabilities.
A formal treatment of applying an Eulerian solver to
these problems can be found in~\cite{CKPPSV16}
and~\cite{AJSS19}.
Adaptations of Eulerian Laplacian solvers have also led
to improved bounded-space algorithms for estimating
random walk probabilities~\cite{AKMPS20}.

Our algorithm is most closely related to the
previous nearly-linear time directed Laplacian solver~\cite{cohen2018solving}.
That algorithm is motivated by single variable elimination
and a matrix Martingale based analysis.
However, it invokes both components of block elimination
algorithms: finding strongly diagonally dominant subsets,
and invoking sparsification as black-boxes.
The runtime overhead of this routine over sparsification
is at least $\log^{5}n$:
in~\cite{cohen2018solving}\footnote{arXiv version
1~\url{https://arxiv.org/pdf/1811.10722v1.pdf}},
Lemma 5.1 gives that each phase (for a constant
factor reduction) invokes sparsification
$O(\log^{2}n)$ times, and each call
is ran with error at most $\frac{1}{O(\log^{3}n)}$
(divided by $\log^{2}n$ in Line 2 of Algorithm 2,
and also by $\log{n}$ in Line 5 of Algorithm 3).

While our algorithms are directed analogs of the
undirected block elimination routines
from in~\cite{kyng2016sparsified},
our analyses rely on many structures
developed in~\cite{cohen2018solving}.
Specifically, our cumulative error during elimination steps
is bounded via the matrix that's the sum of undiretifications
of the intermediate directed matrices.
On the other hand, we believe our algorithm is more natural:
our sampling no longer needs to be locally unbiased,
the per-step errors do not need to be decreased by polylog factors,
and the algorithm is no longer divided into inner/outer phases.
This more streamlined algorithm leads to our runtime improvements.

Our Schur Complement sparisifcation algorithm is based
on the partial block elimination routine from~\cite{kyng2016sparsified},
which is in turn based on a two-term decomposition
formula for (pseudo-)inverses from~\cite{peng2014efficient}.
We remark that there is a good sparsification routine in the low space setting~\cite{AKMPS20}.
There is a subsequent algorithm that replaces this decomposition
with directly powering via random walks~\cite{CCLPT15}
that's also applicable for sparsifying undirected Schur Complements.
However, that algorithm relies on sparsifying $3$-step
random walk polynomials, which to our knowledge,
is a subroutine that has not been studied in directed settings.
As a result, we are unable to utilize this later
development directly.

The existence of $O(n \log^{4}n)$ sized sparsifiers
in~\cite{CGPSSW18} relies on decomposing unit weighted
graphs into short cycles and $O(n)$ extra edges.
While this decomposition has a simple $O(m^2)$ time
algorithm (peel off all vertices with degree $<$ 3,
then return the lowest cross-edge in the BFS tree),
the current fastest construction of it takes
$m^{1 + o(1)}$ time~\cite{CGPSSW18,PY19,LSY19}.
As a result, we need to instead invoke the more expensive,
graph decomposition based, algorithms from~\cite{cohen2017almost}
for sparsification.
Also, we can only use the naive $O(m^2)$ construction of
$O(\log{n})$-lengthed cycle decompositions
(after an initial sparsification call to make
$m = O(n\log^{O(1)}n)$)
because the almost-linear time algorithm
in~\cite{LSY19} produces $O(\log^{2}{n})$-lengthed cycles.

\section{Preliminary }\label{sec:notation}

\subsection{Notations}
\begin{flushleft}

\textbf{General Notations:}
The notation $\Otil{\cdot}$ suppresses the
$polyloglog(n)$ factors in this paper.

We let $[n] = \dr{1, 2, \cdots, n}$.
For matrix $\AA$, $\nnz{\AA}$ denotes its number of nonzero entries.

For matrix $\AA\in \MS{n}{n}$ and subsets $T_1, T_2\sleq [n]$, $\AA_{T_1 T_2}\in \MS{\abs{T_1}}{\abs{T_2 }}$ is the submatrix containing the entries with row indexes and column indexes in $(T_1, T_2)$; and $\AA_{-T_1, - T_2}$ is the submatrix of $\AA$ by removing the rows indexed by $T_1$ and columns indexed by $T_2$.
For vector $\vv\in \Real^n$ and subset $C\sleq [n]$, $\vv_{C}$ is the subvector of $\vv$ containing the entries indexed by $C$.
\\~\\

\textbf{Matrix:}
We use $\II_{a}$, $\zero_{b\times c}$ to denote the identity matrix of size $a$ and the $b$-by-$c$ zero matrix,  and we sometimes omit the subscripts when their sizes can be determined from the context.
For any matrix $\XX\in \MatSize{a}{b}$ and set $T_1, T_2 \sleq [n]$ with $\abs{T_1} = a$, $\abs{T_2} = b$,
$\putmat{\XX, T_1, T_2, n} $ denotes an $n$-by-$n $ matrix whose submatrix indexed by $\pr{T_1, T_2}$ equals $\XX$ and all the other entries equal $0$.
In other words, $\putmat{\XX, T_1, T_2, n}$ can be regarded as replacing the submatrix indexed by $\pr{T_1, T_2}$ with $\XX$ in the zero matrix $\zero_{{n}\times {n}}$.

For symmetric matrix $\AA\in\MS{n}{n}$, we use $\la_i\pr{\AA}$ to denotes its $i$-th smallest eigenvalue.
For symmetric matrices $\AA, \BB\in \MS{n}{n}$, we use $\AA \pgeq \BB$ $(\AA \succ \BB)$  to indicate that for any $\xx\in \Real^n$, $\xx\tp\AA\xx \geq \xx\tp\BB\xx$ $(\xx\tp\AA\xx > \xx\tp\BB\xx)$.
We define $\pleq$, $\prec$ analogously.
A square matrix $\AA\in\MS{n}{n}$ is positive semidefinite (PSD) iff $\AA$ is symmetric and $\AA\pgeq \zero$;
 $\AA\in \MS{n}{n}$ is positive definite (PD) iff $\AA$ is symmetric and $\AA \succ \zero$.

For PSD matrix $\AA$, $\AA^{1/2}$ is its square root; $\AA\dg$ denotes its Moore-Penrose pseudoinverse; $\AA^{\dagger/2}$ is the square root of its Moore-Penrose pseudoinverse.
\\~\\

\textbf{Vector:}
$\one_a$, $\zero_b$ denote the $a$-dimensional all-ones vector and $b$-dimensional all-zeros vector; when their sizes can be determined from the context, we sometimes omit the subscripts.

For matrix $\AA\in \MS{n}{n}$, $\Diag{\AA}$ is an $n$-by-$n$ diagonal matrix with the same diagonal entries as $\AA$.
For vector $\xx\in \Real^n$, $\Diag{\xx}$  denotes an $n$-by-$n$ diagonal matrix with its $i$-th diagonal entry equalling $\xx_i$.

For any positive semidefinite matrix $\AA$, we define the vector norm $\nA{\AA}{\xx} = \sqrt{\xx\tp\AA\xx}$.
And $\nm{\cdot}_{p}$ denotes the $\ell^p$ norm.
\\~\\

\textbf{Matrix norm:}
$\nm{\cdot}_p$ denotes the $\ell^p$ norm.
For instance, for $\AA\in \MS{n}{n}$, $\nt{\AA} = \sqrt{\la_n\pr{\AA\tp\AA}}$; $\ni{\AA} = \max_{i\in [n]}\sum_{j=1}^{n}\abs{\AA_{ij}}$.
For matrix $\BB\in \MS{n}{n}$ and PSD matrix $\AA\in \MS{n}{n}$, we denote $\narr{\AA}{\BB} = \sup_{\nA{\AA}{\xx}\neq 0}\frac{\nA{\AA}{\BB\xx}}{\nA{\AA}{\xx}}$.
\\~\\

\textbf{Schur complement:}
For $\AA\in \MS{n}{n}$ and $F, C$ a partition of $[n]$ such that $\AA_{FF}$ is nonsingular, the Schur complement of $F$ in $\AA$ is defined as $\sc{\AA, F} = \AA_{CC} - \AA_{CF}\AA_{FF}\inv\AA_{FC}$.

When we need to emphasize the support set of the
entries that remain,
we also denote $\sc{\AA, - C} = \sc{\AA, F}$.

\subsection{(Directed) Laplacians, Symmetrizations}

A matrix $\LL\in \MS{n}{n}$ is called a directed Laplacian iff $\one\tp\LL = \zero\tp$ and all off-diagonal entries of $\LL$ are non-positive, i.e., $\LL_{ii} = - \sum_{j: j\neq i}\LL_{ji} $ for all $i\in [n]$ and $\LL_{ij} \leq 0$ for all $i\neq j$.
A (directed) Laplacian $\LL$ can be associated with a (directed) graph $\calG[{\LL}]$ whose adjacency matrix is $\Ahat = \Diag{\LL} - \LL\tp$.
The in-degrees/out-degrees of $\LL$ are defined as the in-degrees/out-degrees of $\calG[\LL]$.
For directed Laplacians, its out-degrees equal
its diagonal entries.
If $\calG[{\LL}]$ is strongly connected, we say the (directed) Laplacian $\LL$ is strongly connected.

In addition, if $\LL\one = \zero$, we call $\LL$ an Eulerian Laplacian.
These Laplacians have the property that in-degrees of vertices
equal to out-degrees.
The undirected Laplacian is a special case where $\LL = \LL\tp$.
We often refer to these as symmetric Laplacians,
or just Laplacians.
\\~\\

\textbf{Symmetrization:} For square matrix $\AA\in \MS{n}{n}$, we define its matrix symmetrization as $\U{\AA} = \frac{\AA + \AA\tp}{2}$.
For a directed Laplacian $\LL\in \MS{n}{n}$, we define its undirectification as $\UG{\LL} = \frac{1}{2}\pr{\LL + \LL\tp - \Diag{\pr{\LL + \LL\tp}\one}}. $
$\UG{\LL}$ is called the undirectification because it is a symmetric Laplacian whose adjacency matrix is $\U{\Ahat}$, where $\Ahat = \Diag{\LL} - \LL\tp$ is the adjacency matrix of $\calG[{\LL}]$.
For an Eulerian Laplacian $\LL$, its matrix symmetrization coincides with its undirectification, i.e., $\U{\LL} = \UG{\LL}$.
Eulerian Laplacians are critically important in solvers
for directed Laplacians because they are the only setting
in which the undirectification is positive semidefinite.
\\~\\

\textbf{Row Column Diagonal Dominant (RCDD):} A square matrix $\AA\in \MS{n}{n}$ is $\alp$-RCDD iff $\sum_{j\in [n]\dele \dr{i} }\abs{\AA_{ij}} \leq \frac{1}{1 + \alp}\AA_{ii}$ and $\sum_{j\in [n]\dele \dr{i} }\abs{\AA_{ji}} \leq \frac{1}{1 + \alp}\AA_{ii}$ for any $i\in [n]$. We also say $\AA$ is RCDD if $\AA$ is $0$-RCDD.

\end{flushleft}

\subsection{Sparsification}\label{sec:sparseblkb}

All almost-linear time or faster solvers for directed Laplacians
to date are built around sparsification: the approximation of
graphs by ones with fewer edges.
As it's difficult to even approximate reachability of directed
graphs, \cite{cohen2017almost} introduced the key idea of measuring
approximations w.r.t. a symmetric PSD matrix.
Such approximations are at the core of all subsequent
algorithms, including ours.

\begin{definition}\label{def:aleq1}
    (Asymmetrically bounded)
    Given a matrix $\AA\in \MS{n}{n}$ and a PSD matrix $\UU\in\MS{n}{n}$, $\AA$ is \asymb\ by $\UU$ iff $\ker\pr{\UU} \sleq \ker\pr{\AA\tp}\cap \ker\pr{\AA}$ and $\ndd{\UU}{\AA} \leq 1$.
    We denote it by $\AA \aleq \UU$.

\end{definition}
By our definition, $\AA \aleq \UU$ is equivalent to $-\AA \aleq \UU$.
The following lemma is changed slightly from Lemma~B.2 of~\cite{cohen2017almost}.
\newcommand\lemne{
    For any matrix $\AA\in \MS{n}{n}$ and PSD matrix $\UU\in \MS{n}{n}$, the following statements are equivalent:
    \begin{itemize}
      \item $\AA \aleq \UU.  $ 

      \item $2\xx\tp\AA\yy \leq \xx\tp\UU\xx + \yy\tp\UU\yy,\ \forall \xx, \yy \in \Real^n. $

    \end{itemize}
}
\begin{fact}\label{lem:ne}
    \lemne

\end{fact}

\begin{definition}
\label{def:UApprox}
(Approximation of directed Laplacians via undirectification)
Given matrix $\AA\in \MS{n}{n}$ and directed Laplacian $\BB\in \MS{n}{n}$, $\AA$ is an \Lapap{\eps} of $\BB$ iff
$\AA - \BB$ is \asymb\ by $\eps \cdot \UG{\BB}$.
\end{definition}

In particular, for \strc\ Eulerian Laplacians $\AA$ and $\BB$, $\AA$ is an \Lapap{\eps} of $\BB$ iff $\nd{\BB}{\pr{\AA - \BB}} \leq \eps.  $

We will utilize sparsifiers for Eulerian Laplacians~\cite{cohen2017almost,CGPSSW18},
as well as implicit sparsifiers for products of directed
adjacency matrices as black boxes throughout our presentations.
The formal statements of these black boxes are below.

\begin{theorem}\label{thm:SparEoracle1}
    (Directed Laplacian sparsification oracle)
     Given a directed Laplacian $\LL\in \MS{n}{n}$ with $\nnz{\LL} = m$ and error parameter $\dlt \in (0, 1)$, there is an oracle $\SparE$ which  runs in at most $\TSE\pr{m, n, \dlt}$ time,
    where $\TSE\pr{m, n, \dlt} = O((m\log^{O(1)}n + n\log^{O(1)}n) \dlt^{-O(1)})$,
    to return with high probability a directed Laplacian $\Lap$ satisfying:
     \begin{enumerate}[(i)]
       \item $\nnz{\Lap} \leq \NSE\pr{n, \dlt}$
       where $\NSE\pr{n, \dlt} = O(n\log^{O(1)}n \dlt^{-O(1)})$;

       \item $\Diag{\Lap} = \Diag{\LL}$; \label{enum:Diagnece}

       \item 
        $\Lap - \LL \aleq \dlt\cdot \UG{\LL} $. \label{enum:LapLLdltUdef1}

%

     \end{enumerate}

\end{theorem}
\begin{remark}
Conditions~\eqref{enum:Diagnece},~\eqref{enum:LapLLdltUdef1}
in Theorem~\ref{thm:SparEoracle1} above are equivalent to
$\Lap$ and $\LL$ having the same in- and out-degrees,
and $\nG{\LL}{\pr{\Lap - \LL}} \leq \dlt  $ respectively.
By having the same in-degrees and out-degrees, we mean $\Diag{\Lap} = \Diag{\LL}  $ and $\Lap\one = \LL\one$.
\end{remark}

\begin{lemma}\label{lem:SparP}
    (Lemma~3.18 of~\cite{cohen2017almost})
    Let $\xx, \yy\in \Real^n$ be nonnegative vectors with $\nnz{\xx} + \nnz{y} = m$ and let $\eps, p\in (0, 1)$.
    And we denote $\GG = \pr{\one\tp\xx}\Diag{\yy} - \xx\yy\tp$.
    Then, there is a routine $\SparP$ which computes
    with probability at least $1 -  p$
    a nonnegative matrix $\AA$ in $O\pr{m\eps^{-2}\log \frac{m}{p}}$ time such that $\nnz{\AA } = O\pr{m\eps^{-2}\log \frac{m}{p} }$, $\AA - \xx\yy\tp \aleq \eps\cdot \UG{\GG}$.

\end{lemma}

Given an Eulerian Laplacian $\LL\in \MS{n}{n} $ and a partition $F, C$ of $[n]$, by invoking \\ $\SparE$ on subgraphs with edges inside $\pr{F, F}$, $(F, C)$, $(C, F)$, $(C, C)$ respectively, we can get a Laplacian sparsification procedure
$\SE$ so that the sparsified Eulerian Laplacians returned by $\SE$ not only satisfy all the properties mentioned in Theorem~\ref{thm:SparEoracle1}, but also keep the in-degrees and out-degrees of the subgraph supported by $\pr{F, F}$.
Analogously, a routine $\SP$ can be constructed by applying $\SparP$ four time.
For explicit definitions of $\SE$ and $\SP$,
see Lemma~\ref{lem:SE}.

\subsection{Sufficiency of Solving Eulerian
Systems to Constant Error}

Previous works on solvers for directed Laplacians and their
generalizations (to RCDD and M-matrices) established
that it's sufficient to solve Eulerian systems to constant
relative accuracy in their undirectification.
\begin{itemize}
    \item The iterative refinement procedure shown in~\cite{cohen2017almost} shows that a constant accuracy
    solver can be amplified to one with $\epsilon$ relative
    accuracy in $O(\log(1 / \epsilon))$ iterations.
    \item The stationary computation procedure in~\cite{CKPPSV16}
    showed that arbitrary strongly connected Laplacians with
    mixing time $T_{mix}$ can be solved to $2$-norm error $\epsilon$
    by solving $O(\log(T_{mix} / \epsilon))$ systems in Eulerian
    Laplacians.
    This was subsequently simplified and extended to M-matrices
    and row-column-diagonally-dominant matrices in~\cite{AJSS19}
    (with an extra $\log{n}$ factor in running time).
    A purely random walk (instead of matrix perturbation) based
    outer loop is also given in the thesis of Peebles~\cite{Peebles19:thesis}.
\end{itemize}

\input{overview}

\input{schur}

\input{solver}

\bibliography{ref}

\input{Appendix}

\end{document}

%% file: newcommand.tex
\newcommand\SparP{\textsc{SparseProduct}}
\newcommand\SparE{\textsc{OraSparseLaplacian}}  \def\Lapap#1{$#1$-\emph{asymmetric approximation}}

\newcommand\FindRCDD{\textsc{FindRCDDBlock}}
\newcommand\PRI{\textsc{PreRichardson}}

\newcommand\SparseSchur{\textsc{SparseSchur}}

\newcommand\SE{\textsc{SparseEulerianFC}} \newcommand\TSE{\calT_{\rm S}} \newcommand\NSE{\calN_{\rm S}}
     
\newcommand\SP{\textsc{SparseProductFC}}
\newcommand\SCC{\textsc{SchurChain}}
\newcommand\PreC{\textsc{Precondition}}

\def\mx#1{\begin{bmatrix}
           #1
        \end{bmatrix}}
\def\vc#1{\begin{pmatrix}
        #1
        \end{pmatrix}}
\def\sc#1{\textsc{Sc}\pr{#1}}

\newcommand\alp{\alpha}
\newcommand\dele{\backslash}
\def\U#1{\UU\br{#1}}
\def\nd#1#2{\nt{\U{#1}^{\dagger/2} {#2} \U{#1}^{\dagger/2}}  }
\def\ndd#1#2{\nt{{#1}^{\dagger/2} {#2} {#1}^{\dagger/2}}  }
\def\nG#1#2{\nt{\UG{#1}^{\dagger/2} {#2} \UG{#1}^{\dagger/2}}  }
\def\nA#1#2{\nm{#2}_{#1}}

\def\Mt#1{\bm{\mathcal{M}}^{\pr{#1}}}

\def\At#1{\AA^{\pr{#1}}}
\def\Dt#1{\DD^{\pr{#1}}}
\def\iv#1{\pr{#1}\inv}

\newcommand\xtil{\widehat{\xx}}

\newcommand\Con{\PPi}

\newcommand\dg{^\dagger}

\newcommand\dlt{\delta}
\newcommand\Ltil{\widetilde{\LL}}

\def\Ytt#1{\widetilde{\YY}^{\pr{#1}}}
\newcommand\bet{\beta}
\newcommand\gam{\gamma}

\def\zerom#1#2{\zero_{\abs{#1}\times \abs{#2}}}

\def\St#1{\SS^{\pr{#1}}}
\def\Stt#1{\widetilde{\SS}^{\pr{#1}}}
\def\iv#1{\pr{#1}\inv}
\def\tpp#1{\pr{#1}\tp}
\def\Ltt#1{\widetilde{\LL}^{\pr{#1}}}
\def\Att#1{\widetilde{\AA}^{\pr{#1}}}
\newcommand\Sap{\widehat{\SS}}
\def\Xtt#1{\widetilde{\XX}^{\pr{#1}}}
\newcommand\Xap{\widetilde{\XX}}
\def\EY#1{\widehat{\EE}^{\pr{#1}}}
\def\EXX#1{\EE_{\XX}^{\pr{#1}}}
\newcommand\EX{\EE_{\XX}}
\newcommand\ex{\xx}
\newcommand\ey{\yy}
\def\Gt#1{\GG^{\pr{#1}}}
\def\Wt#1{\WW^{\pr{#1}}}

\def\Uvc#1{\mathcal{V}^{\rm G}\pr{#1}}
\def\zerov#1{\zero_{\abs{#1}}}
\newcommand\epsz{\eps_0}
\def\putmat#1{\mathcal{P}\pr{#1}}
\def\uni#1{\mathcal{R}\pr{#1}}

\def\psit#1{\psi^{\pr{#1}}}

\newcommand\xn{\xx_{\rm N}}
\newcommand\xu{\xx_{\rm U}}
\newcommand\ym{\yy_{\rm M}}

\newcommand\yu{\yy_{\rm U}}
\def\Mtt#1{\bm{\widetilde{\mathcal{M}}}^{\pr{#1}}}
\def\Ntt#1{\widehat{\NN}^{\pr{#1}}}
\def\Ett#1{\EE^{\pr{#1}}}
\def\unierr#1{\uni{#1}}
\newcommand\Errap{\bm{\mathcal{E}}}

\def\Errt#1{\Errap^{\pr{#1}}}
\newcommand\repmat{\emph{``repetition matrix"}}
\def\rep#1{{\rm Rep}\pr{#1}}
\def\repp#1{{\rm Rep^{+0}}\pr{#1}}
\def\repFC#1{{\rm Rep}\pr{#1}}

\def\Ers#1{\bm{\mathcal{E}}^{\pr{1:#1, #1}}}

\newcommand\Rap{\widehat{\RR}}
\def\poly#1{{\rm poly}\pr{#1}}
\newcommand\Dap{\widehat{\DD}}
\def\Ltilt#1{\widetilde{L}^{\pr{#1}}}
\def\Utilt#1{\widetilde{U}^{\pr{#1}}}
\def\Mpt#1{\MM^{\pr{#1}}}
\def\DS#1{\widetilde{D}^{\pr{#1}}}
\newcommand\bigcdot{}

\newcommand\PA{\PP_{\AA}}
\def\narr#1#2{\left\| #2 \right\|_{#1 \arr #1}}
\newcommand\Lap{\widetilde{\LL}}
\newcommand\strc{strongly connected}
\newcommand\Bap{\widetilde{\BB}}

\def\puts#1{2^{#1}\Fn + \Cn}

\def\spanrm#1{{\rm span}\pr{#1}}
\def\UG#1{\mathcal{U}^{\rm G}\br{#1}}
\newcommand\Fn{\abs{F}}
\newcommand\Cn{\abs{C}}
\newcommand\Contil{\widetilde{\Con}}
\def\XL#1{\mathcal{\QQ}^{\pr{#1}}}

\newcommand\Zhat{\widehat{\ZZ}}

\def\xhatt#1{\widehat{\xx}^{\pr{#1}}}
\newcommand\Ahat{\widehat{\AA}}
\def\Mzoo#1{\MM^{\pr{#1}}}
\def\Nzoo#1{\NN^{\pr{#1}}}
\def\Uzoo#1{\UU^{\pr{#1}}}
\def\betzoo#1{\bet^{\pr{#1}}}
\def\rhozoo#1{\rho^{\pr{#1}}}
\def\muzoo#1{\mu^{\pr{#1}}}
\def\bzoo#1{b^{\pr{#1}}}
\def\azoo#1{a^{\pr{#1}}}

\newcommand\aleq{\stackrel{\mathrm{asym}}{\pleq}}

\newcommand\asymb{\emph{asymmetrically bounded}}

\def\Ly#1{\boldsymbol{\Phi}\pr{#1|\DD_{FF}, F}}
\def\Upm#1{\U{\begin{pmatrix}#1 \end{pmatrix} }}

%% file: overview.tex
\section{Overview}
\label{sec:overview}

Our algorithm is based on sparse Gaussian elimination.
Before we discuss the block version, it is useful to
first describe how the single variable version works
on Eulerian Laplacians.

Recall that Eulerian Laplacians store the (weighted) degrees
on the diagonal, and the negation of the out edge weights
from $j$ in column $j$.

Suppose we eliminate vertex $j$.
Then we need to add a rescaled version of row $j$ to
each row $i$ where $\LL_{j,:}$ is non-zero.
Accounting for $\LL_{j, j} = \dd_{j}$,
this weight for row $i$ is given by
$\frac{\ww_{j \rightarrow i}}{\dd_{j}}$,
and the corresponding decrease in entry $j, k$ is then
\[
\frac{\ww_{j \rightarrow i} \ww_{k \rightarrow j}}{\dd_{j}}.
\]
In other words, when eliminating vertex $j$,
we add an edge $k \rightarrow i$ for each triple of vertices
$k \rightarrow j \rightarrow i$, with weight given by above.

The effect of this elimination on the vector $b$
can also be described through this `distribution' of
row $j$ onto its out-neighbors.
However, to start with, it's useful to focus on what
happens to the matrix.
The key observation for elimination based Laplacian solvers
is that this new matrix remains a graph.
In fact, it can be checked that this process exactly
preserves the in and out degrees of all neighbors of $i$,
so the graph also remains Eulerian.

However, without additional assumptions on the non-zero
structures such as separators,
directly performing the above process takes $O(n^{3})$ time:
the newly added entries quickly increases the density
of the matrix until each row has $\Theta(n)$ entries.
So the starting point of elimination based Laplacian
solvers is to address the following two problems:
\begin{enumerate}
    \item Keeping the results of elimination sparse.
    \item Find vertices that are easy to eliminate.
\end{enumerate}

\subsection{Block Cholesky}
\newcommand\PrecPosition{
\vspace{0.2cm}
\begin{algorithm}[!htb]
\caption{$\PreC\pr{\dr{\dr{\Stt{i}}_{i=1}^d, \dr{F_i}_{i=1}^d}, \xx, N}$}
\label{alg:precondition}

\KwIn{ $\pr{\alp, \bet, \dr{\dlt_i}_{i=1}^d }$-Schur complement chain $\dr{\dr{\Stt{i}}_{i=1}^d, \dr{F_i}_{i=1}^d}$;
    vector $\xx\in \Real^n $; error parameter $\eps\in (0, 1)$  }

\KwOut{ vector $\xx\in \Real^n$ }


\For{$i = 1, \cdots, d - 1$}{
    $\xx_{F_i} \arl \PRI\pr{\Stt{i}_{F_i F_i}, \xx_{F_i}, \Diag{\Stt{i}}_{F_i F_i}\inv, \frac{1}{2}, N } $  \;
    $\xx_{C_i} \arl \xx_{C_i} - \Stt{i}_{C_i F_i}\xx_{F_i}    $
}

$\xx_{F_d} \arl \pr{\Stt{d}}\dg\xx_{F_d}  $  \;

\For{$i = d - 1, \cdots, 1 $}{
    $\xx_{F_i} \arl \xx_{F_i} - \PRI\pr{\Stt{i}_{F_i F_i}, \Stt{i}_{F_i C_i}\xx_{C_i}, \Diag{\Stt{i}}_{F_i F_i}\inv, \frac{1}{2}, N } $
}

Let  $\xx \arl \xx - \frac{\one\tp\xx}{n}\cdot \one $  \;

Return $\xx$

\end{algorithm}
\vspace{0.2cm}
}

One possible solution to the issue above is to directly
sample the edges formed after eliminating each vertex.
It leads to sparsified/incomplete Cholesky factorization
based algorithms~\cite{KS16,cohen2018solving},
including the first nearly-linear time solver for directed Laplacians.

Our algorithm is based on eliminating blocks of vertices,
and is closest to the algorithm from~\cite{kyng2016sparsified}.
It aims to eliminate a block of $\Omega(n)$ vertices
simultaneously.
This subset, which we denote using $F$, is chosen to
be almost independent.
That is, $F$ is picked so that each vertex in $F$
has at least a constant portion of its out-degree
going to $V \setminus F$, which we denote as $C$.

This property means that any random walk on $F$ exits
it with constant probability.
This intuition, when viewed from iterative methods perspective,
implies that power method allows rapid simulation of
elimination onto $C = V \setminus F$.
From a matrix perspective, it means these
matrices are well-approximated by their diagonal.
So subproblem $\LL_{FF}\inv\bb$ can be solved to high
accuracy in $O\pr{\log n}$ iterations via power method.
We formalize the guarantees of such procedures,
\textsc{PreRichardson} in Lemma~\ref{lem:PRIconverge1}
in Section~\ref{sec:preconditioner}.

\PrecPosition

Compared to single-vertex elimination schemes,
block elimination has the advantage of having
less error accumulation.
Single elimination can be viewed as eliminating $1$
vertex per step, while we will show that
the block method eliminates $\Omega(n)$
vertices in $O(\log\log{n})$ steps.
This smaller number of dependencies in turn provides
us the ability to bound errors more directly.

Formally, given the partition $F, C\sleq [n]$ with the permutation matrix $\PP$ such that $\PP\LL\PP\tp = \mx{\LL_{FF}&\LL_{FC}\\ \LL_{CF}&\LL_{CC}}$,
the block Cholesky factorization of $\LL\in \MS{n}{n}$ is given as
\eq{
    \LL =
    \PP\tp\mx{
        \II_{\abs{F}} & \zerom{F}{C}  \\
        \LL_{CF}\LL_{FF}\inv & \II_{\abs{C}}
    }
    \cdot
    \mx{
        \LL_{FF} & \zerom{F}{C}  \\
        \zerom{C}{F} & \sc{\LL, F}
    }
    \cdot
    \mx{
        \II_{\Fn} & \LL_{FF}\inv\LL_{FC}  \\
        \zerom{C}{F} & \II_{\abs{C}}
    }\PP.
}

Using the above factorization iteratively (with sparsification) generates
a Schur complement chain $\dr{\dr{\Stt{i}}_{i=1}^d, \dr{F_i}_{i=1}^d }$ (Definition~\ref{def:SCC1}),
the solver algorithm loops through these and solves the subproblems
via the projection / prologation maps defined via the
random walk on $F_i$ respectively, using the power-method
based elimination procedure described above.
Its pseudocode is given in Algorithm~\ref{alg:precondition}
for completeness.

We remark that in practice, the iteration numbers of the preconditioned Richardson iterations in Algorithm~\ref{alg:precondition} can differ with each other.
Here, we set a uniform $N$ merely for simplicity.
If we have unlimited (or quadratic) precomputation power, the method described above suffices to give us a fast solver.
However, due to the exact Schur Complements being dense,
the major remaining difficulty is to efficiently
compute an approximate Schur complement.

\subsection{Schur Complement Sparsification via Partial Block Elimination }

\newcommand\SparSchurPseuC{
\vspace{0.2cm}
\begin{algorithm}[!htb]
\caption{$\SparseSchur\pr{\LL, F, \dlt}$}
\label{alg:SparSchur}

\KwIn{
    strongly connected Eulerian Laplacian $\LL\in \MatSize{n}{n}$;
    partition $F, C$ of $[n]$;
    error parameter $\dlt \in (0, 1)$
}

\KwOut{
    Sparse approximate Schur complement $\SS$
}

If $\nnz{\LL} \geq O\pr{\NSE\pr{n, \dlt}}$, call $\SparE$ to sparsify $\LL$ with error parameter $O\pr{\dlt}$  \;
Find a permutation matrix $\PP$ such that $\PP\LL\PP\tp = \mx{\LL_{FF} & \LL_{FC}  \\ \LL_{CF} & \LL_{CC} } $  \label{algline:defpermu1}
  \;
Set $K \arl O\pr{\log\log \frac{n}{\dlt}}$, $\eps \arl O\pr{\frac{\dlt }{K}}$,
$\Ltt{0} \arl \LL $,
$\DD \arl \Diag{\Ltt{0}}$, $\Att{0} \arl \DD - \Ltt{0}$  \;

\For{$k = 1, \cdots, K$}{
    \For{$i \in F$}{
         $\Ytt{k, i} \arl \SP\pr{\Att{k-1}_{:, i}, \tpp{\Att{k-1}_{i,:}}, \eps, F}$
    }
        Let $\Ytt{k} \arl \sum_{i\in F} \frac{1}{\DD_{ii}}\Ytt{k,i}$,
        $\Ltt{k,0} \arrlf
        \PP\tp
        \mx{
             \DD_{FF} & - \Att{k-1}_{FC}  \\
             - \Att{k-1}_{CF} & 2\Ltt{k-1}_{CC}
        }\PP  - \Ytt{k}
        $ \label{line:Lttk01} \;
        $\Ltt{k} \arl \SE\pr{\Ltt{k,0},  \eps, F}$
        and
        $\Att{k} \arl \PP\tp\mx{\DD_{FF} & \\ & \Diag{\Ltt{k}_{CC}} }\PP - \Ltt{k} $ \label{line:Attk1} \;
}
\For{$i  \in F$}{
    $\Xtt{i} \arl \SparP\pr{\Att{K}_{C,i}, \tpp{\Att{K}_{i, C}}, \eps } $
}
Let $\Xap \arl \sum_{i\in F} \frac{1}{\DD_{ii}} \Xtt{i}$ and $\Stt{0} \arrlf \frac{1}{2^K}\pr{\Ltt{K}_{CC} - \Xap}  $ \;
Compute a patching matrix $\RR \in \MatSize{\Cn}{\Cn} $ with $\RR_{2:\Cn,1} = - \Stt{0}_{2:\Cn,:}\one $, $\RR_{1,2:\Cn} = - \one\tp\Stt{0}_{:, 2:\Cn}$, $\RR_{1, 1} = - \RR_{1,2:\Cn}\one - \one\tp\RR_{2:\Cn,1} - \one\tp\Stt{0}\one $, and $\RR_{ij} = 0$ for $i\neq 1$ and $j\neq 1$   \;
Set $\Sap = \Stt{0} + \RR$  \;
Return $\SS = \SparE\pr{\Sap, \dlt/8}$

\end{algorithm}
\begin{remark}
    This permutation matrix $\PP$ defined on line~\ref{algline:defpermu1} of Algorithm~\ref{alg:SparSchur} is only used to simplify the pseudocodes in Lines~\ref{line:Lttk01},~\ref{line:Attk1}. We don't need to construct it in practice.
    We use the same $\DD_{FF}$ in each iteration to simplify our analysis.
    It is possible to replace $\DD_{FF}$ by $\Diag{\Ltt{k-1}}_{FF} $ in iterations $k$ and achieve similar running time.

\end{remark}
}

Thus, the main bottleneck toward an efficient algorithm
is the fast construction of approximate Schur complements.
We will give such an algorithm whose running time is close
to those of sparsification primitives via a partial block
elimination process.

In simple terms, a step of this process
squares the $(F, F)$ block.
Repeating this gives quadratic convergence.
With $\alpha$ about $O(1)$, $O(\log\log{n})$ iterations
suffice, so the resulting error is easier to control than
martingales.

\begin{lemma}
    [\cite{peng2014efficient}]
    For any diagonal matrix $\DD\in\MatSize{n}{n}$ and a matrix $\AA\in \MatSize{n}{n}$ with $\DD - \AA$ nonsingular, we have
    \eql{\label{eq:D-A}}{
        \pr{\DD - \AA}\inv = \frac{1}{2}\pr{\DD\inv + \pr{\II + \DD\inv\AA}\pr{\DD - \AA\DD\inv\AA}\inv\pr{\II + \AA\DD\inv} }.
    }
\end{lemma}

The main identity~\eqref{eq:D-A} gives rise to our definition for partial-block-eliminated Laplacian of $\LL$.
Let $\DD = \Diag{\LL}$ and $\AA = \DD - \LL$, let $\PP$ be the permutation matrix such that $\PP\LL\PP\tp = \mx{\LL_{FF}& \LL_{FC} \\ \LL_{CF} & \LL_{CC}}$.
The partial-block-eliminated operator $\Ly{\LL} $ is defined as
\eq{
    \Ly{\LL} = \PP\tp\mx{\LL_{FF} & - \AA_{FC} \\ - \AA_{CF} & 2\LL_{CC}}\PP - \AA_{:, F}\DD_{FF}\inv\AA_{F, :}.
}

We define the first exact partial-block-elimination of $\LL$ by $\Lt{1} = \Ly{\LL}$.
Then, $\frac{1}{2}\Lt{1}$ is an Eulerian Laplacian which has the same Schur complement of $F$ in $\LL$, i.e.,
\eq{
    \sc{\LL, F} = \frac{1}{2}\sc{\Lt{1}, F }.
}
The $2$-nd to the $K$-th partially-block-eliminated Laplacians are defining iteratively as $\Lt{2} = \Ly{\Lt{1}}$, $\cdots$, $\Lt{K} = \Ly{\Lt{K-1}}$.
$\Lt{k}$ can also be regarded as a partially powered matrix of $\LL$, which uses the powering to obtain better spectral properties.
Specifically, when we focus on the $(F, F)$ block of $\Lt{k}$, it is easy to see that $\ni{\DD_{FF}\inv\At{k}_{FF}}$ converges at a quadratic rate, where $\At{k}_{FF} = \Lt{k}_{FF} - \DD_{FF}$.
Formal construction of the $k$-th partially-block-eliminated Laplacians and their properties are deferred to Appendix~\ref{sec:exactPBE}.

To encounter the increasing density of $\Lt{k}$, sparsification blackboxes in Section~\ref{sec:sparseblkb} are naturally accompanied with the partial block elimination to yield a Schur complement sparsification method (Algorithm~\ref{alg:SparSchur}).
Our algorithm is essentially a directed variant of the one
from~\cite{kyng2016sparsified}.
A slight difference is that in the last step,
we need to fix the degree discrepancies caused by approximating
the strongly RCDD matrix by its diagonal.

The running time of Algorithm~\ref{alg:SparSchur} is shown in Theorem~\ref{thm:SparSchur}.
The $k$-th iterand $\Ltt{k}$ of Algorithm~\ref{alg:SparSchur} is termed the approximate $k$-th partially-block-eliminated Laplacian, while its exact version is just the (exact) $k$-th partially-block-eliminated Laplacian $\Lt{k}$ defined above.
To guarantee the performance of Algorithm~\ref{alg:SparSchur}, the most important thing is to provide relatively tight bounds for the difference $\Ltt{k} - \Lt{k}$.

\SparSchurPseuC

\newcommand\thmSparSchur{
    (Schur complement sparsification)
    For a \strc\ Eulerian Laplacian $\LL\in\MS{n}{n}$,
    let $F, C$ be a partition of $[n]$ such that $\LL_{FF}$ is $\alp$-RCDD $(\alp = O(1))$ and  let $\dlt \in (0, 1)$ be an error parameter,
    the subroutine $\SparseSchur$ (Algorithm~\ref{alg:SparSchur}) runs in time
    \eq{ O\pr{\TSE\pr{m, n, \dlt}} + \Otil{\TSE\pr{\NSE\pr{n, \dlt}\dlt^{-2}  , n, \dlt}\log n  }   }
    to return with high probability a \strc\ Eulerian Laplacian $\SS$ satisfying
    $\nnz{\SS} = O\pr{\NSE\pr{\abs{C}, \dlt}} $
    and
    \eql{\label{eq:sttgood}}{
        \SS - \sc{\LL, F} \aleq \dlt \cdot \U{\sc{\LL, F}}.
    }
}
\begin{theorem}\label{thm:SparSchur}
    \thmSparSchur

\end{theorem}

Compared to the undirected analog from~\cite{kyng2016sparsified},
powered directed matrices exhibit significantly more complicated
spectral structures.
To analyze them, we develop new interpretations of directed
Schur complements based on matrix extensions.

\subsection{Bounding Error Accumulations in Partially-Eliminated Laplacians by Augmented Matrices }\label{sec:oveaugm1}

When considering the approximate partial block elimination,
in one update step, not only new sparsification errors are added into $\Ltt{k}$, the errors accumulated from previous steps will multiply with each other and get possibly amplified.
In addition, error accumulations in Schur complements of
directed Laplacians are not as straightforward as their
undirected counterparts.
It's not the case that for two directed Eulerian Laplacians
with the same undirectifiation, the undirectification of their
Schur complements are the same.
For instance, consider the undirected vs. directed cycle,eliminated till only two originally diametrically opposite
vertices remain.
The former has a Schur complement that has weight $2/n$,
while the latter has a Schur complement that has weight $1$.

By the definition of \Lapap{\eps},
we need to essentially show the following inequality
in order to obtain the approximations needed for
a nearly-linear time algorithm:
\eq{
    &\frac{1}{2^k}\U{\Lt{k}} = \frac{1}{2^k}\U{\Ly{\Lt{k-1}}} \\
    =&\frac{1}{2^k}
    \U{\vc{
        2\Lt{k-1}_{CC} - \At{k-1}_{CF}\DD_{FF}\inv\At{k-1}_{FC} & - \At{k-1}_{CF}\pr{\II + \DD_{FF}\inv\At{k-1}_{FF}}  \\
        - \pr{\II + \At{k-1}_{FF}\DD_{FF}\inv}\At{k-1}_{FC} & \DD_{FF} - \At{k-1}_{FF} \DD_{FF}\inv\At{k-1}_{FF }
    }} \\
    \pleq&  O\pr{1}\cdot \U{\LL},\ \forall 1\leq k\leq K.
}
Here significant difficulties arise due to the already
complicated formula of $\Lt{k}$.
So we instead express the exact and approximate partial block elimination as Schur
complements of large augmented matrices introduced below.

In the rest of Section~\ref{sec:overview}, we assume $C = \dr{1, 2, \cdots, \abs{C}}$ for simplicity.

\subsubsection{A Reformulation for Partial Block Elimination }
For the exact and approximate $k$-th partially-block-eliminated matrices $\Lt{k}$, $\Ltt{k}$,
we define augmented matrices $\Mt{0, k}$, $\Mtt{0, k}$ of size $2^{k} |F| + |C|$.
We start with the construction of a desirable $\Mt{0, k}$. 
To this end, we define a sequence of augmented matrices $\dr{\Mt{i, k}}_{i=0}^k$, where $\Mt{k,k} = \Lt{k}$ and each $\Mt{i, k}$ is a Schur complement of $\Mt{i-1, k}$.
Here we only give an informal explanation
of how we construct $\Mt{i,k}$.
The formal definitions of these augmented matrices are given in Section~\ref{sec:reformPBEvAM1}.
To begin with, for some fixed $k\in [K]$, we define
$
    \Mt{k, k} \defeq \Lt{k}.
$
And we write $F_1 = F$ in the remainder of Section~\ref{sec:overview} and the entire Section~\ref{sec:reformPBEvAM1} and Section~\ref{sec:Schurcplstable2}.

Next, we take $\Mt{k-1, k}$ and $\Mt{k-2, k}$ as examples to show how  we  define such a sequence of matrices $\Mt{k-1, k}, \cdots, \Mt{0, k}$.

Define $\Mt{k-1, k}$ as follows
\eq{
    \Mt{k-1, k} \defeq
    \begin{blockarray}{cccl}
        C & F_1 & F_2 & \arl \text{column indexes } \\
        \begin{block}{[ccc]l}
            2\Lt{k-1}_{CC} & - \At{k-1}_{CF} & -\At{k-1}_{CF} & \arl \text{rows indexed by } C  \\
            - \At{k-1}_{FC} & \DD_{FF}  & -\At{k-1}_{FF} & \arl \text{rows indexed by } F_1 \\
            - \At{k-1}_{FC} & - \At{k-1}_{FF} & \DD_{FF}  &  \arl \text{rows indexed by } F_2  \\
        \end{block}
    \end{blockarray}
}

Then, it follows by direct calculations that
\eq{
    \sc{\Mt{k-1, k}, F_2} =
    \mx{
        2\Lt{k-1}_{CC} - \At{k-1}_{CF}\DD_{FF}\inv\At{k-1}_{FC} & - \At{k-1}_{CF}\pr{\II + \DD_{FF}\inv\At{k-1}_{FF}} \\
        - \pr{\II + \At{k-1}_{FF}\DD_{FF}\inv}\At{k-1}_{FC}  & \DD_{FF} - \At{k-1}_{FF}\DD_{FF}\inv\At{k-1}_{FF}
    }
    \comeq{\eqref{line:densebiclique}} \Lt{k}. 
}

Next, we define $\Mt{k-2, k}$ as follows
\eq{
    \Mt{k-2, k} \defeq
    \begin{blockarray}{cccccl}
        C & F_1 & F_2 & F_3 & F_4 & \arl \text{column indexes } \\
        \begin{block}{[ccccc]l}
            4\Lt{k-2}_{CC} & - \At{k-2}_{CF}  & - \At{k-2}_{CF} & - \At{k-2}_{CF}  & -\At{k-2}_{CF} & \arl \text{rows indexed by } C  \\
            - \At{k-2}_{FC} & \DD_{FF} &   & -\At{k-2}_{FF} &  & \arl \text{rows indexed by } F_1  \\
            - \At{k-2}_{FC} &   & \DD_{FF} & & -\At{k-2}_{FF} & \arl \text{rows indexed by } F_2  \\
            - \At{k-2}_{FC} &   & -\At{k-2}_{FF} &  \DD_{FF} &  & \arl \text{rows indexed by } F_3 \\
            - \At{k-2}_{FC} & - \At{k-2}_{FF} &  &  & \DD_{FF} & \arl \text{rows indexed by } F_4  \\
        \end{block}
    \end{blockarray}.
}
It follows by direct calculations that
\eq{
    &\sc{\Mt{k-2, k}, F_3\cup F_4} \\
    =&
    \mx{
        4\Lt{k-2}_{CC} - 2\At{k-2}_{CF}\DD_{FF}\inv\At{k-2}_{FC} & - \At{k-2}_{CF}\pr{\II + \DD_{FF}\inv\At{k-2}_{FF}} & - \At{k-2}_{CF}\pr{\II + \DD_{FF}\inv\At{k-2}_{FF}}   \\
        - \pr{\II + \At{k-2}_{FF}\DD_{FF}\inv}\At{k-2}_{FC} & \DD_{FF}  & - \At{k-2}_{FF}\DD_{FF}\inv\At{k-2}_{FF}   \\
        - \pr{\II + \At{k-2}_{FF}\DD_{FF}\inv}\At{k-2}_{FC} & - \At{k-2}_{FF}\DD_{FF}\inv\At{k-2}_{FF} & \DD_{FF}
    }  \\
    =&
    \mx{
        2\Lt{k-1}_{CC} & - \At{k-1}_{CF} & -\At{k-1}_{CF} \\
        - \At{k-1}_{FC} & \DD_{FF}  & -\At{k-1}_{FF} \\
        - \At{k-1}_{FC} & - \At{k-1}_{FF} & \DD_{FF}
    }  \\
    =& \Mt{k-1, k}.
}

We will show $\frac{1}{2^k}\U{\Lt{k}} \pleq O\pr{1} \cdot \U{\LL}$ later by analyzing the properties of $\Mt{0, k}$.

We believe this representation  may be of independent interest.
We also remark that these augmented matrices only arise during
analysis, and are not used in the algorithms.

\subsubsection{Bounding Error Accumulation of Schur complement sparsification }
Next, we propose Lemma~\ref{lem:schurdUL1} to bound the errors after taking Schur complements.
However, in our analysis, iteratively applying Lemma~\ref{lem:schurdUL1} to bound $\frac{1}{2^k}\pr{\Ltt{k} - \Lt{k}} $ will lead to more $\log n$ factors in the running time.

To derive a tighter bound, we introduce another group of augmented matrices $\dr{\Mtt{0, k}}$ which are defined by
 attaching sparsification errors to $\Mt{0, k}$.
$\Mtt{0, k}$ can help us disentangle the sparsification errors generated from different iterations and see how these errors accumulate as we do partial block eliminations more clearly.

We use another group of augmented matrices $\dr{\XL{k}}$ to bound the difference between $\Lt{k}$ and $\Ltt{k} $.
The augmented matrix $\XL{k}$ is defined as the weighted sum of a group of
``reptition matrix" (Section~\ref{sec:erroraccusml1}).
$\XL{k}$ adopts many properties similar to $\Mt{0, k}$, so it is easy to analyze.
Then, we can give tighter bound for $\Mtt{K} - \Mt{K}$ using the robustness of Schur complements in this case (Section~\ref{sec:Schurcplstable2}).

Here, we use some examples to illustrate how we construct the augmented matrices $\dr{\XL{k}}$, $\dr{\Mtt{0, k}}$.
Formal definitions are in Section~\ref{sec:reformPBEvAM1}.

Firstly, we define the sparsification error $\Ett{k} = \Ltt{k} - \Ly{\Ltt{k-1}}  $.
Take $\Mtt{0, 2}$ as an example.
We define
\eq{
      \Mtt{1, 2} =
      \mx{
            2\Ltt{1}_{CC} & - \Att{1}_{CF} & -\Att{1}_{CF}   \\
            - \Att{1}_{FC} & \DD_{FF}  & -\Att{1}_{FF}  \\
            - \Att{1}_{FC} & - \Att{1}_{FF} & \DD_{FF}
      }
      +
      \mx{
        \Ett{2}_{CC} & \Ett{2}_{CF} &  \\
        \Ett{2}_{FC} & \Ett{2}_{FF} &   \\
         &  & \zerom{F}{F}
      }
}
and
\eq{
    \Mtt{0, 2} =&
    \mx{
        4\LL_{CC} & - \AA_{CF}  & - \AA_{CF} & - \AA_{CF}  & - \AA_{CF}   \\
            - \AA_{FC} & \DD_{FF} &   & - \AA_{FF} &    \\
            - \AA_{FC} &   & \DD_{FF} & & -\AA_{FF}   \\
            - \AA_{FC} &   & -\AA_{FF} &  \DD_{FF} &   \\
            - \AA_{FC} & - \AA_{FF} &  &  & \DD_{FF}
    }  \\
    &+
    \mx{
        2\Ett{1}_{CC} & \Ett{1}_{CF} & \Ett{1}_{CF} &  \\
        \Ett{1}_{FC} & \zero & \Ett{1}_{FF} & \\
        \Ett{1}_{FC} & \Ett{1}_{FF} & \zero & \\
        & & & \zerom{2F}{2F  }
    }
    +
    \mx{
        \Ett{2}_{CC} & \Ett{2}_{CF} &  \\
        \Ett{2}_{FC} & \Ett{2}_{FF} &   \\
         &  & \zerom{3F}{3F  }
    }.
}
Then, it follows by direct calculations that
\eq{
    \sc{\Mtt{0, 2}, F_3\cup F_4} = \Mtt{1, 2},\quad
    \sc{\Mtt{1, 2}, F_2} = \Ltt{2}.
}
Analogously, for larger $k$, we can also define a sequence of augmented matrices $\dr{\Mtt{i, k}}$ such that each $\Mtt{i+1, k}$ is a Schur complement of $\Mtt{i, k}$ and $\Mtt{k,k} = \Ltt{k} $.
Then, $\sc{\Mtt{0, k}, -[n]} = \Ltt{k}$.

Next, we take $\XL{1}$, $\XL{2}$ as examples to show how we construct $\XL{k}$ iteratively.
We remark that the $\XL{1}$, $\XL{2}$ defined below are a little different with those in Lemma~\ref{lem:Ersepsz1}.
In addition, they are not the best choice up to constants.
However, what roles they play in the proofs can be seen easily from the following examples.

We naturally start with $\XL{1} = \U{\Mt{0,1}}$ and
it is easy to show that $\Mt{0, 1} - \Mtt{0, 1} \aleq O(\eps) \cdot \XL{1}$ and $\Ltt{1} - \Lt{1} \aleq O\pr{\eps}\cdot \sc{\U{\XL{1}}, F}$ by Theorem~\ref{thm:SparEoracle1} and Fact~\ref{lem:scrobust}.
Using~\eqref{eq:xEy127} in the proof of Lemma~\ref{lem:Ersepsz1}, we have
\eq{
  &\mx{
    2\Ett{1}_{CC} & \Ett{1}_{CF} & \Ett{1}_{CF} &   \\
        \Ett{1}_{FC} & \zero & \Ett{1}_{FF} &  \\
        \Ett{1}_{FC} & \Ett{1}_{FF} & \zero &  \\
        & & & \zerom{2F}{2F}
  } \aleq
  O\pr{\eps} \cdot
  \Upm{
    2\Ltt{1}_{CC} & \Ltt{1}_{CF} & \Ltt{1}_{CF} &  \\
    \Ltt{1}_{FC} & \Ltt{1}_{FF} & \zero &  \\
    \Ltt{1}_{FC} & \zero & \Ltt{1}_{FF}  & \\
    & & & \zerom{2F}{2F}
  }  \\
  =&
  O\pr{\eps} \cdot
  \Upm{
    2\Ltt{1}_{CC} - 2\Lt{1}_{CC} & \Ltt{1}_{CF} - \Lt{1}_{CF} & \Ltt{1}_{CF} - \Lt{1}_{CF} &  \\
    \Ltt{1}_{FC} - \Lt{1}_{FC} & \Ltt{1}_{FF} - \Lt{1}_{FF} & \zero &  \\
    \Ltt{1}_{FC} - \Lt{1}_{FC} & \zero & \Ltt{1}_{FF} - \Lt{1}_{FF} &   \\
    & & & \zerom{2F}{2F}
  } \\
  & +
  O\pr{\eps} \cdot
  \Upm{
    2\Lt{1}_{CC} & \Lt{1}_{CF} & \Lt{1}_{CF} &  \\
    \Lt{1}_{FC} & \Lt{1}_{FF} & \zero &  \\
    \Lt{1}_{FC} & \zero & \Lt{1}_{FF}  & \\
    & & & \zerom{2F}{2F}
  }.
}
Notice that
\eql{\label{eq:LttLtdiff21}}{
  &\Upm{
    2\Ltt{1}_{CC} - 2\Lt{1}_{CC} & \Ltt{1}_{CF} - \Lt{1}_{CF} & \Ltt{1}_{CF} - \Lt{1}_{CF} &  \\
    \Ltt{1}_{FC} - \Lt{1}_{FC} & \Ltt{1}_{FF} - \Lt{1}_{FF} & \zero &  \\
    \Ltt{1}_{FC} - \Lt{1}_{FC} & \zero & \Ltt{1}_{FF} - \Lt{1}_{FF} &  \\
    & & & \zerom{2F}{2F}
  }  \\
  \pleq&
  O\pr{\eps} \cdot
  \Upm{
    2\XL{1}_{CC} & \XL{1}_{CF} & \XL{1}_{CF} & \XL{1}_{C, F_2} & \XL{1}_{C, F_2}  \\
    \XL{1}_{FC} & \XL{1}_{FF} &  & &  \\
    \XL{1}_{FC} &  & \XL{1}_{FF} & &   \\
    \XL{1}_{F_2, C} &  &  & \XL{1}_{F_2, F_2} &  \\  
    \XL{1}_{F_2, C} &  & &  & \XL{1}_{F_2, F_2}
  }
}
and
\eql{\label{eq:LttLtdiff22}}{
  \Upm{
    2\Lt{1}_{CC} & \Lt{1}_{CF} & \Lt{1}_{CF} &  \\
    \Lt{1}_{FC} & \Lt{1}_{FF} & \zero &  \\
    \Lt{1}_{FC} & \zero & \Lt{1}_{FF}  & \\
    & & & \zerom{2F}{2F}
  }
  \pleq
  \Upm{
    2\Mt{1, 2}_{CC} & \Mt{1, 2}_{CF} & \Mt{1, 2}_{CF} & \Mt{1, 2}_{C, F_2} & \Mt{1, 2}_{C, F_2}  \\
    \Mt{1, 2}_{FC} & \Mt{1, 2}_{FF} &  &  & \\
    \Mt{1, 2}_{FC} &  & \Mt{1, 2}_{FF} &  &  \\
    \Mt{1, 2}_{F_2, C} & & & \Mt{1, 2}_{F_2, F_2} & \\
    \Mt{1, 2}_{F_2, C} &  &  &  & \Mt{1, 2}_{F_2, F_2}
  }.
}
In addition,
\eql{\label{eq:LttLtdiff23}}{
    \mx{
        \Ett{2}_{CC} & \Ett{2}_{CF} &  \\
        \Ett{2}_{FC} & \Ett{2}_{FF} &   \\
         &  & \zerom{3F}{3F  }
    }
    \aleq
    \Upm{\Mt{0, 2}}.
}
Define $\XL{2}$ as a weighted sum of the RHS of~\eqref{eq:LttLtdiff21},~\eqref{eq:LttLtdiff22},~\eqref{eq:LttLtdiff21}.
Then, we will have $\Mtt{0, 2} - \Mt{0, 2} \aleq O\pr{\eps} \cdot \XL{2}$.
Using the robustness of Schur complements, we can derive that $\Ltt{2} - \Lt{2} \aleq O\pr{\eps} \cdot \sc{\XL{2}, -[n]} $.

\begin{figure}
  \centering
  \includegraphics[width=0.9\textwidth]{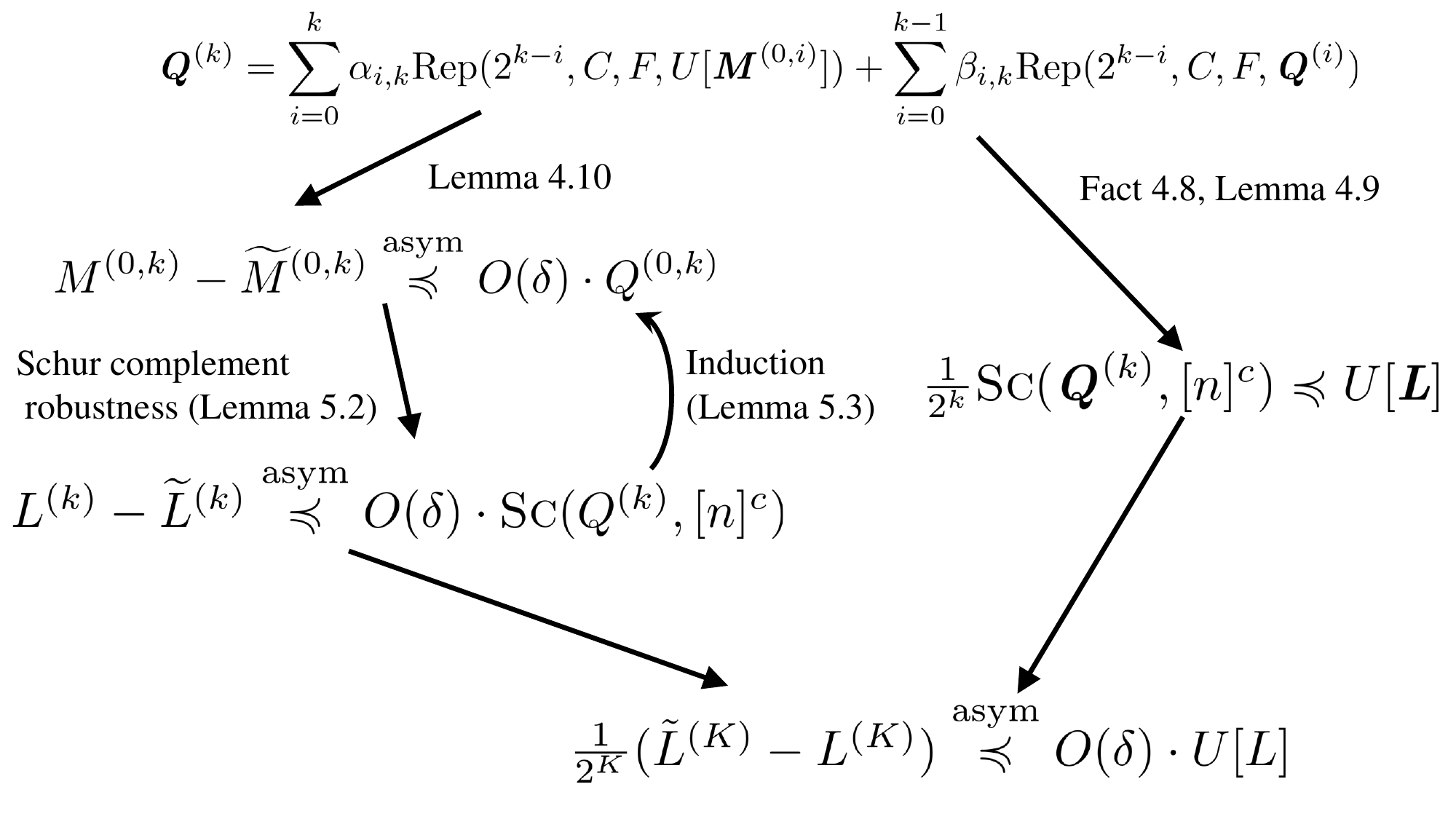}
  \caption{An illustration on our approach to bound the difference $\Lt{k} - \Ltt{k}$}\label{fig:lemmasdiffLtLttk}
\end{figure}

Now, our approach to bound $\Lt{k} - \Ltt{k}$ is summarized in Figure~\ref{fig:lemmasdiffLtLttk}.  
 Essentially, we use augmented matrices to separate the sparsification errors generated from different iterations and see their relations more clearly.
The properties of the augmented matrices $\dr{\Mt{0, k}}$, $\dr{\Mtt{0, k}}$ mentioned in this paper are essentially used to guarantee that we can use
the Schur complement robustness safely to bound the difference $\frac{1}{2^k}\pr{\Ltt{k} - \Lt{k}} $.
Then, we induct on these terms carefully so that the errors $\gam_k = \ndd{\sc{\XL{k}, -[n]}}{\pr{\Ltt{k} - \Lt{k}}} $ increases linearly rather than exponentially in $k$.
In this way, we obtain a sparsified Schur complement with only $\Otil{1}$ calls to the sparsification blackboxes rather than $O\pr{\log^{O\pr{1}} n }$ times.

%% file: schur.tex

\section{Partial Block Elimination via Augmented Matrices}
\label{sec:reformPBEvAM1}
In this section, we introduce our augmented matrices based
view of partial block elimination.
As we will show later, after $O\pr{\log\log n}$ steps of partial elimination,
the $(F, F)$ block of the approximate partially-block-eliminated Laplacian $\Ltt{k}$ can be approximated by its diagonal ``safely".
So, what remains is to bound the error accumulations in the difference $\frac{1}{2^k}\pr{\sc{\Ltt{k}, F} - \sc{\Lt{k}, F}}$,
which we do by bounding differences in $\frac{1}{2^k}\pr{\Ltt{k} - \Lt{k}}$.

 In Section~\ref{sec:ULtk},
we represent the exact $k$-th partially-block-eliminated
Laplacian $\Lt{k}$ by a Schur complement of an augmented matrix
$\Mt{0, k}\in \MS{\pr{2^k\abs{F} + \Cn}}{\pr{2^k\abs{F} + \Cn}}$.

In Section~\ref{sec:erroraccusml1}, we define $\Mtt{0,k}$ as an
inexact version of $\Mt{0,k}$ where sparsification errors accumulate.
Then, we introduce a special type of augmented matrices,
which we term \emph{reptition matrices},
and bound the difference $\Mtt{0, k} - \Mt{0, k}$
in norms based on these \emph{reptition matrices}.

\subsection{A Reformulation of the Exact Partial Block Elimination }\label{sec:ULtk}

To bound the difference $\frac{1}{2^k}\pr{\Ltt{k} - \Lt{k}}$,
a direct way is to bound the difference
$\frac{1}{2^k}\pr{\Ltt{k} - \Lt{k}}$ by $\frac{1}{2^{k-1}}\pr{\Ltt{k-1} - \Lt{k-1}}$ recursively.
However, our attempts of doing so lead to errors that grow
exponentially in $k$, the number of partial block elimination steps.

In this section, we provide a reformulation of the exact version of partial block elimination which is more friendly to error analysis.
To be specific,  our strategy is to construct a large matrix $\Mt{0, k} \in \MS{\pr{\puts{k}}}{\pr{\puts{k}}}$ such that $\Lt{k}$ is a Schur complement of the large matrix $\Mt{0, k}$.
And there is a partition of $\Mt{0, k}$ such that each block is a zero matrix or equals some submatrix of $\LL$.
To construct $\Mt{0, k}$, we will construct a sequence of augmented matrices $\dr{\Mt{i, k}}_{i=0}^{k}$ satisfying Lemma~\ref{lem:Mti}. 
Later, by analyzing the large matrix $\Mt{0, k}$, we can derive tighter bounds for quantities related to $\Lt{k}$.

In Section~\ref{sec:reformPBEvAM1} and Section~\ref{sec:Schurcplstable2},
unless otherwise specified,
we assume $F = \dr{n - \Fn +1, \cdots, n-1, n}$ and $C = \dr{1, 2, \cdots, \Cn} = [n]\dele F$ by default.

Now, we give a rigorous way to construct such a sequence of matrices $\dr{\Mt{i, k}}_{i=0}^k \ (0\leq k\leq K)$.
To begin with, we define some sets to be used later.
We define \eq{F_a = \dr{b\in \mathbb{Z}:\  \Cn + \pr{a - 1}\Fn + 1 \leq b \leq \Cn + a\Fn},\ \forall 1\leq a\leq 2^K. }
Note that in our notation, $F = F_1$.

Then,
we construct a sequence of bijections $\dr{\psit{i}\pr{\cdot}}_{i=0}^K$ which indicate the ``positions" of the blocks equalling $- \At{i}_{FF}$  in the large augmented matrix $\Mt{k - i, k}$.

We start with $\psit{0}\pr{\cdot}$ and will define these $\psit{i}$ iteratively.
The mapping $\psit{0}\pr{\cdot}$ is defined as a trivial mapping from $\dr{1}$ to $\dr{1}$ with
\eq{
    \psit{0}\pr{1} = 1.
}
Then, assume we have defined $\psit{i-1}\pr{\cdot}$.
Now, we define $\psit{i}$ as follows:
\eq{
    \psit{i}\pr{a} =
    \left\{
        \begin{split}
            &a + 2^{i-1},\ a\in [2^{i-1}]  \\
            &\psit{i-1}\pr{a - 2^{i-1}},\ 2^{i-1}+1 \leq a\leq 2^{i}
        \end{split}
    \right.
}
If $\psit{i-1}\pr{\cdot}$ is a bijection from $[2^{i-1}]$ to $[2^{i-1}]$, then $\psit{i}\pr{\cdot} \Big|_{\dr{2^{i-1}+1, \cdots, 2^i}}$ is a bijection from $\dr{2^{i-1} + 1, \cdots, 2^i}$ to $[2^{i-1}]$.
And by the definition, $\psit{i}\pr{\cdot}\Big|_{[2^{i-1}]}$ is a bijection from $[2^{i-1}]$ to $\dr{2^{i-1}+1, \cdots, 2^i}$.
Then,  $\psit{i}$ is a bijection from $[2^i]$ to $[2^i]$.

It follows by induction that for any $k\in [K]$, $\psit{k}\pr{\cdot}$ is a bijection from $[2^k]$ to $[2^k]$.
And by the fact that $2^{i-1}+1 \leq  \psit{i}\pr{a} \leq 2^{i}$ for $a\in [2^{i-1}]$ and $\psit{i}\pr{a} \in [2^{i-1}]$ for $2^{i-1}+1 \leq a \leq 2^{i} $, we have the following relation
\eql{\label{eq:Atnotdiag}}{
    \psit{j}\pr{a} \neq a,\ \forall 1\leq j\leq K,\  a\in [2^{j}].
}

With the notations defined above,
we define the matrix $\Mt{i, k}$ as
\eql{\label{eq:Mtik1}}{
    \Mt{i, k} =&
    2^{k-i}\putmat{\Lt{i}_{CC}, C, C, \puts{k-i}}
    + \sum_{a = 1}^{2^{k-i}}\Big(\putmat{\DD_{FF}, F_a, F_a, \puts{k-i} } \\
    & + \putmat{- \At{i}_{FF}, F_a, F_{\psit{k-i}\pr{a}}, \puts{k-i}} + \putmat{- \At{i}_{FC}, F_a, C, \puts{k-i}} \\
    & + \putmat{- \At{i}_{CF}, C, F_{a}, \puts{k-i}}\Big),\ \forall 0\leq k\leq K,\ 0\leq i\leq k.
}
where the notation $\putmat{\XX, A, B, n}$ has been defined in Section~\ref{sec:notation}, which means putting matrix $\XX$ in the submatrix indexed by $\pr{A, B}$ in a zero matrix $\zero_{n\times n}$;
$\Lt{k}$ is the exact $k$-th partially-block-eliminated Laplacian and formal definitions of $\Lt{k}, \At{k}$ are in Appendix~\ref{sec:exactPBE}.

We have the following properties of $\dr{\Mt{i, k}}$.
\begin{lemma}\label{lem:Mti}
    For any $0\leq k\leq K$, $0\leq i\leq k$, $\Mt{i, k}$ is an Eulerian Laplacian;
    $\Mt{i, k}_{ - [n], - [n]}$, $\Mt{i, k}_{- C, - C}$ are $\alp$-RCDD;
    the Schur complement satisfies
    \eql{\label{eq:scMFi+}}{
        \sc{\Mt{i, k}, \cup_{a = 2^{k-i-1}+1}^{2^{k-i}} F_a } = \Mt{i+1, k}.
    }
    Further,
    \eql{\label{eq:scMFtili+}}{
        \sc{\Mt{i, k}, - [n]} = \Lt{k}.
    }
    In addition, for any $\xx\in \Real^n$, let
    \eql{\label{eq:xtil2pk}}{
        \xtil = \big( \xx_C\tp  \underbrace{\xx_{F}\tp \ \cdots \ \xx_F\tp}_{\text{$2^k$ repetitions of $\xx_F\tp$}}  \big)\tp,
    }
    then,
    \eql{\label{eq:morning1}}{
        \xtil\tp\Mt{0, k}\xtil = 2^k \xx\tp \LL  \xx.
    }

\end{lemma}
\begin{proof}
  We only need to prove the cases of  $i < k$.
  Denote $F_i^+ = \cup_{a = 2^{i-1}+1}^{2^i } F_a$ in this proof.
  Since $\psit{k-i}\pr{\cdot}$ is a bijection from $[2^{k-i}]$ to $[2^{k-i}]$,
  then, for any $a\in [2^{k-i}]$,
  $\Mt{i, k}_{F_a, :} $ contains
   3 nonzero blocks equaling $\DD_{FF}, -\At{i}_{FF}, -\At{i}_{FC}$, respectively;
   and the other blocks are all zero matrices.
  Analogously, for any $a\in [2^{k-i}]$,
  $\Mt{i, k}_{:, F_a}$ contains 3 nonzero blocks equaling $\DD_{FF}, -\At{i}_{FF}, -\At{i}_{CF}$, respectively, while the other blocks are all zero matrices.
  Since $\Lt{i}$ is Eulerian (Lemma~\ref{lem:LtkE}),
  we have $\Mt{i,k}_{F_a, :}\one = \DD_{FF}\one - \At{i}_{FF}\one - \At{i}_{FC}\one = \Lt{i}_{F, :}\one = \zero $.
  Analogously $\one\tp\Mt{i,k}_{:, F_a} = \zero\tp $.
  Also by the definition~\eqref{eq:Mtik1}, $\Mt{i, k}_{C, :}\one = 2^{k-i}\pr{\Lt{i}_{CC}\one - \At{i}_{CF}\one} = 2^{k-i}\Lt{i}_{C, :}\one = \zero$.
  Analogously,
   $\one\tp\Mt{i, k}_{:, C} = \zero\tp$.
  All off-diagonal entries of $\Mt{i,k}$ are non-positive by the definition~\eqref{eq:Mtik1}.
  Thus, $\Mt{i, k}$ is an Eulerian Laplacian.

Notice that all blocks equaling $-\At{i}_{CF}$ are on the rows indexed by $C$; and all blocks equaling $- \At{i}_{FC}$ are on the columns indexed by $C$, thus, on each block row or block column of $\Mt{i, k}_{- C, - C}$, there is exactly one block equaling $\DD_{FF}$ and exactly one block equaling $- \At{i}_{FF}$, and the other blocks are all-zeros matrices.
  Thus, by~\eqref{eq:DinvAsuperl}, $\Mt{i, k}_{- C, - C}$ is $\alp$-RCDD.
  Then, $\Mt{i, k}_{- [n], - [n]}$ is also $\alp$-RCDD as it is a submatrix of $\Mt{i, k}_{ - C, - C}$.

  Also, as there are only 3 nonzero blocks on each row or column of $\Mt{i, k}$,
  when computing the Schur complement of $F_a \  (2^{k-i-1}+1\leq a \leq 2^{k-i})$, we only need to focus on the submatrix
  \eq{
    \Mt{i, k}_{C\cup F_{a - 2^{k-i-1}} \cup F_a,\ C \cup F_{\psit{k-i}\pr{a}} \cup F_a}
    =
    \mx{
        2^{k-i}\Lt{i}_{CC} & - \At{i}_{CF} & - \At{i}_{CF}  \\
        - \At{i}_{FC} & \zero & - \At{i}_{FF}  \\
        - \At{i}_{FC} & - \At{i}_{FF} & \DD_{FF}
    },
  }
  where the block $\Mt{i, k}_{F_{a - 2^{k-i-1}}, F_{\psit{k-i}\pr{a}}} = \zero$ is by~\eqref{eq:Atnotdiag}.

  Then, by direct calculations,
  \eq{
    &\sc{\Mt{i, k}_{C\cup F_{a - 2^{k-i-1}} \cup F_a,\ C \cup F_{\psit{k-i}\pr{a}} \cup F_a}, F_a}  \\
    =&
    \mx{
         2^{k-i}\Lt{i}_{CC} - \At{i}_{CF}\DD_{FF}\inv\At{i}_{FC} & - \At{i}_{CF}\pr{\II + \DD_{FF}\inv\At{i}_{FF}}  \\
        - \pr{\II + \At{i}_{FF}\DD_{FF}\inv}\At{i}_{FC} & - \At{i}_{FF}\DD_{FF}\inv\At{i}_{FF}
    } \\
    =&
    \mx{
        2^{k-i}\Lt{i}_{CC} - \At{i}_{CF}\DD_{FF}\inv\At{i}_{FC} & - \At{i+1}_{CF}  \\
        - \At{i+1}_{FC} & - \At{i+1}_{FF}
    }.
  }
  When computing the Schur complement of $F_{k-i }^+ $, the term $\At{i}_{CF}\DD_{FF}\inv\At{i}_{FC} $ is subtracted from $\Mt{i, k}_{CC} = 2^{k-i}\Lt{i}_{CC}$ for $2^{k-i-1}$ times.   By combining with the equality $\Lt{i+1}_{CC} = 2\Lt{i}_{CC} - \At{i}_{CF}\DD_{FF}\inv\At{i}_{FC}$ from~\eqref{line:densebiclique}, we have
  $
    \sc{\Mt{i, k}, F_{k-i}^+}_{CC} = \Mt{i+1, k}_{CC}.
  $

  To derive~\eqref{eq:scMFi+},
  what remains is to show that the ``positions" of the blocks equaling $-\At{i+1}_{FF}$ are the same in $\sc{\Mt{i, k}, F^+_{k-i} }$ and $\Mt{i+1, k}$.
  This can be seen easily from the observation that
  for $2^{k-i-1}+1 \leq a\leq 2^{k-i}$,
  $
    \sc{\Mt{i, k}_{C\cup F_{a - 2^{k-i-1}} \cup F_a,\ C \cup F_{\psit{k-i}\pr{a}} \cup F_a}, F_a}
  $
  is supported on the submatrix indexed by
  \eq{
    \pr{C \cup  F_{a - 2^{k-i-1}}, C \cup  F_{\psit{k-i}\pr{a  }}}.
  }
  Thus, the ``positions" of blocks equaling $-\At{i+1}_{FF}$ in $\sc{\Mt{i, k}, F_{k-i}^+}$ are
  \eq{
    &\dr{ \pr{F_{a - 2^{k-i-1}}, F_{\psit{k-i}\pr{a  }}}: 2^{k-i-1}+1 \leq a \leq 2^{k-i}}  \\
    =&\dr{ \pr{F_{a - 2^{k-i-1}}, F_{\psit{k-i-1}\pr{a - 2^{k-i-1}  }}}: 2^{k-i-1}+1 \leq a \leq 2^{k-i}}  \\
    =& \dr{ \pr{F_a, F_{\psit{k-i-1}\pr{a}}}: a\in [2^{k-i-1}]},
  }
  which are the exact ``positions" of the blocks equalling $-\At{i+1}_{FF}$ in $\Mt{i+1, k}$ by~\eqref{eq:Mtik1}.
  Therefore, we have proved~\eqref{eq:scMFi+}.

  The relation~\eqref{eq:scMFtili+} follows by Fact~\ref{fact:sctran}
  and induction:
  \eq{
    \sc{\Mt{i, k}, - [n]} = \sc{\sc{\Mt{i, k}, F_{k-i}^+}, -[n]} = \sc{\Mt{i+1, k}, -[n]} = \cdots = \Mt{k, k} = \Lt{k}.
  }
  The relation~\eqref{eq:morning1} follows by the fact that $\psit{k}\pr{\cdot}$ is a bijection from $[2^k]$ to $[2^k]$ and~\eqref{eq:Mtik1}.

\end{proof}

The following lemma answers a question in Section~\ref{sec:oveaugm1}.
That is, $\frac{1}{2^k}\U{\Lt{k}} \pleq O\pr{1}\U{\LL}$.

\begin{lemma}\label{lem:ULtm}
    For any $0\leq k\leq K$,
    \eql{\label{eq:Ltm}}{
        \frac{1}{2^k} \U{\Lt{k}} \pleq  \pr{3 + \frac{2}{\alp}} \U{\LL}.
    }

\end{lemma}
\begin{proof}
  By Lemma~\ref{lem:Mti},
  $\Mt{0, k}_{-[n], -[n]} $ is $\alp$-RCDD, and $\Mt{0, k}$ is an Eulerian Laplacian.
  Thus, using Fact~\ref{lem:scrobust},
  \eq{ \U{\sc{\Mt{0, k}, -[n]}} \pleq \pr{3 + \frac{2}{\alp}}\sc{\U{\Mt{0, k }}, -[n]}.  }

  By~\eqref{eq:scMFtili+},
  $\U{\Lt{k}} \pleq \pr{3 + \frac{2}{\alp}}\sc{\U{\Mt{0, k}}, -[n] }.  $

  For any $\xx\in \Real^n$, define $\xtil$ as in~\eqref{eq:xtil2pk}.
  By Fact~\ref{fact:alpRCDDPSDpPD1}, $\U{\Mt{0, k}}_{-[n], -[n]}$ is PD.
  Then, by Fact~\ref{fact:Schurxusmall} and~\eqref{eq:morning1}, we have
  \eq{
    \xx\tp\sc{\U{\Mt{0, k}}, -[n]}\xx \leq \xtil\tp\U{\Mt{0, k}}\xtil 
    = 2^k \xx\tp \U{\LL} \xx,
  }
  i.e., $\sc{\U{\Mt{0, k}}, -[n] } \pleq 2^k \U{\LL} $.

  Combining the above equations yields that
  \eq{
    \U{\Lt{k}} \pleq \pr{3 + \frac{2}{\alp}}\sc{\U{\Mt{0, k}}, -[n]} \pleq 2^k\pr{3 + \frac{2}{\alp}} \U{\LL}.
  }

\end{proof}

\subsection{Bounding Error Accumulation Using Repetition Matrices }
\label{sec:erroraccusml1}
Before we define $\Mtt{0,k}$, we introduce our notations for the errors induced by sparsification.
\eq{
    &\EY{k,i} = \Att{k-1}_{:,i}\Att{k-1}_{i,:} - \Ytt{k,i},\
    \EY{k} = \sum_{i\in F }\frac{1}{\DD_{ii}}\EY{k,i} =  \Att{k-1}_{:,F}\DD_{FF}\inv\Att{k-1}_{F,:} - \Ytt{k}, \\
    &\Ett{k,0} = \Ltt{k} - \Ltt{k,0},\ \Ett{k} = \EY{k} + \Ett{k,0} = \Ltt{k} - \Ly{\Ltt{k-1}}
}
and
\eq{
    \EXX{i} =  \Att{K}_{C,i}\Att{K}_{i,C} - \Xtt{i},\
    \EX = \sum_{i\in F }\frac{1}{\DD_{ii}}\EXX{i} =  \Att{K}_{CF}\DD_{FF}\inv\Att{K}_{FC} - \Xap.
}

We also denote in the rest of this paper,
\eql{\label{eq:defRap}}{
    \Rap = \RR + \frac{1}{2^K}\pr{\Att{K}_{CF}\pr{\DD_{FF} - \Att{K}_{FF}}\inv - \Att{K}_{CF}\DD_{FF}\inv\Att{K}_{FC}}.
}

Some elementary facts of the results of Algorithm~\ref{alg:SparSchur} is given by Lemma~\ref{lem:LapetcpropSparSchurCpmt1} in Appendix~\ref{sec:someprfs}.

By Lemma~\ref{lem:SparP} and Lemma~\ref{lem:SE}, 
 we can provide bounds for the one-step errors in the next lemma.
 Its proof is deferred to Appendix~\ref{sec:someprfs}.
\begin{lemma}\label{lem:EYEXEtt}
    The error matrices satisfies
    \al{
        &\Ett{k} \aleq \epsz \U{\Ltt{k-1}}, \label{eq:EYLtt}  \\
        &\EX \aleq \eps \U{\sc{\Ltt{K}, F}}, \label{eq:EXscLtt}
    }
    where $\epsz = 2\pr{3 + \frac{2}{\alp}}\pr{2\eps + \eps^2}.  $

\end{lemma}
In the remainder of this paper, we write $\epsz = 2\pr{3 + \frac{2}{\alp}}\pr{2\eps + \eps^2}$.

Recall that we define an augmented matrix $\Mt{0, k}$
such that $\Lt{k}$ is its Schur complement
in Section~\ref{sec:ULtk}.
Now, we define $\Mtt{0, k}$ which is an inexact version of $\Mt{0, k} $ to analyze the properties of $\Ltt{k}$.
We first define
\eq{
    \unierr{k, a, \Ett{i}} =& \putmat{\Ett{i}_{FF}, F_a, F_{\psit{k-i}\pr{a}}, \puts{k-i}} + \putmat{\Ett{i}_{FC}, F_a, C, \puts{k-i}} \\
    & + \putmat{\Ett{i}_{CF}, C, F_{\psit{k-i}\pr{a}}, \puts{k-i}} + \putmat{\Ett{i}_{CC}, C, C, \puts{k-i} }.
}
Then, we define the error matrices
\eql{\label{eq:defErrtik1}}{
    \Errt{i, k} =  \sum_{a=1}^{2^{k-i}} \unierr{k, a,  \Ett{i}}
}
and
\eq{
    \Ers{k} = \sum_{i=1}^{k} \Errt{i, k}.
}

The matrix $\Mtt{0,k}$ is defined as follows
\eq{
    \Mtt{0, k} = \Mt{0, k}
    + \Ers{k}.
}

\begin{lemma}\label{lem:Mtti}
    The Schur complement of $[\puts{k}]\dele [n]$ in $\Mtt{0, k}$ satisfies:
    \eql{\label{eq:Mttsc1}}{
        \sc{\Mtt{0,k}, -[n]} = \Ltt{k}.
    }

\end{lemma}
\begin{proof}

  We define the following auxiliary matrices in this proof
  \eq{
    \Ntt{i, k} =&
    2^{k-i}\putmat{\Ltt{i}_{CC}, C, C, \puts{k-i}}
    + \sum_{a = 1}^{2^{k-i}}\Big(\putmat{\DD_{FF}, F_a, F_a, \puts{k-i} } \\
    & + \putmat{- \Att{i}_{FF}, F_a, F_{\psit{k-i}\pr{a}}, \puts{k-i}} + \putmat{- \Att{i}_{FC}, F_a, C, \puts{k-i}} \\
    & + \putmat{- \Att{i}_{CF}, C, F_{a}, \puts{k-i}}\Big).
  }
  Since $\Ltt{0} = \LL$, we have $\Ntt{0, k} = \Mt{0, k}$.
  Also, from definition, $\Ntt{k, k} = \Ltt{k}$.
  We also denote $F_i^+ = \cup_{a = 2^{i-1}+1}^{2^{i} } F_a $ in this proof.

  Define
  \eq{
    \Mtt{i, k} = \Ntt{i,k} + \sum_{j=i+1}^k \Errt{j,k}.
  }
  Similar with the arguments in Lemma~\ref{lem:Mti}, when computing the Schur complement of $F_a$ in $\Ntt{i,k} + \Errt{i+1,k}$ where $2^{k-i-1}+1 \leq a \leq 2^{k-i} $, we only need to focus on the following submatrix
  \eq{
    &\pr{\Ntt{i, k} + \Errt{i+1, k}}_{C \cup  F_{a - 2^{k-i-1}} \cup F_a,\ C\cup F_{\psit{k-i}\pr{a }}\cup F_a}  \\
    =&
    \mx{
        2^{k-i}\Ltt{i}_{CC} & - \Att{i}_{CF}  & - \Att{i}_{CF}   \\
        - \Att{i}_{FC}  & \zero & - \Att{i}_{FF}  \\
        - \Att{i}_{FC}  & - \Att{i}_{FF}  & \DD_{FF}
    }
    +
    \mx{
        2^{k-i-1}\Ett{i+1}_{CC} & \Ett{i+1}_{CF} & \zero \\
        \Ett{i+1}_{FC} & \Ett{i+1}_{FF} & \zero \\
        \zero & \zero & \zero
    },
  }
  where the block $\Ntt{i,k}_{F_{a - 2^{k-i-1}}, F_{\psit{k-i}\pr{a}}} = \zero$ is also by~\eqref{eq:Atnotdiag}.

  By the definition of $\Ett{i+1 } $, we have
  \eq{
    &\sc{\pr{\Ntt{i, k} + \Errt{i+1,k}}_{C\cup F_{a - 2^{k-i-1}} \cup F_a,\ C\cup F_{\psit{k-i }\pr{a  }}\cup F_a}, F_a}  \\
    =&
    \mx{
         2^{k-i}\Ltt{i}_{CC} - \Att{i}_{CF}\DD_{FF}\inv\Att{i}_{FC} + 2^{k-i-1}\Ett{i+1}_{CC} & - \Att{i}_{CF}\pr{\II + \DD_{FF}\inv\Att{i}_{FF}} + \Ett{i+1}_{CF} \\
        - \pr{\II + \Att{i}_{FF}\DD_{FF}\inv}\Att{i}_{FC} + \Ett{i+1}_{FC}  & - \Att{i}_{FF}\DD_{FF}\inv\Att{i}_{FF} + \Ett{i+1}_{FF}
    } \\
    =&
    \mx{
        \pr{2^{k-i} - 2}\Ltt{i}_{CC} + \Ltt{i+1}_{CC} & - \Att{i+1}_{CF}  \\
        - \Att{i+1}_{FC} & - \Att{i+1}_{FF}
    }.
  }
  By similar arguments with Lemma~\ref{lem:Mti}, the ``positions" of the blocks equalling $- \Att{i+1}_{FF}$ in \\ $\sc{\Ntt{i, k} + \Errt{i+1,k}, F_{k-i}^+}$ are exactly the same as those in $\Ntt{i+1, k}$.

  Since $F_{k-i}^+$ contains $2^{k-i-1}$ sets of size $\abs{F}$, we have
  \eq{
    \sc{\Ntt{i, k} + \Errt{i+1,k}, F_{k-i}^+}_{CC} = 2^{k-i-1}\pr{2\Ltt{i}_{CC} - \Att{i}_{CF}\DD_{FF}\inv\Att{i}_{FC} + \Ett{i+1}_{CC}}  = 2^{k-i-1}\Ltt{i+1}_{CC}.
  }
  Thus, we have shown
  \eq{
     \sc{\Ntt{i, k} + \Errt{i+1,k}, F_{k-i}^+} = \Ntt{i+1, k}.
  }

  Since the support set of $\sum_{j=i+2}^{k}\Errt{j, k} $ is $[\puts{k-i-2}]$ which is disjoint with $F_{k-i}^+ $, we have
  \eq{
    \sc{\Mtt{i,k}, F_{k-i}^+} = \sc{\Ntt{i, k} + \sum_{j=i+1}^{k}\Errt{j,k}, F_{k-i}^+} = \Ntt{i+1, k} + \sum_{j=i+2}^{k}\Errt{j,k} = \Mtt{i+1, k}.
  }
  By induction,
  \eq{
    \sc{\Mtt{0, k}, -[n]} = \Mtt{k,k} = \Ntt{k,k} = \Ltt{k}.
  }

\end{proof}

To help bound $\Ers{k}$,
we define some special kinds of matrices termed \emph{``repetition matrices"}. We will construct the matrices $\dr{\XL{k}}$ as linear combinations of \emph{``repetition matrices"}.
\begin{definition}
    (``Repetition matrices")
    We will use the following 3 kinds of \emph{``repetition matrices"}:
    consider a matrix $\AA\in \MS{m}{m}$ and subset $C \sleq [m]$ and $E = [m]\dele C$,
    \begin{enumerate}[(i)]
      \item
      the $k$-\repmat\ of $\AA$ is defined as follows:
    \eq{
        \rep{k, C, \AA} =
        \mx{
            k \AA_{CC} & \AA_{CE} & \AA_{CE} & \cdots & \AA_{CE} \\
            \AA_{EC} & \AA_{EE} & \zero & \cdots & \zero  \\
            \AA_{EC} & \zero & \AA_{EE} & \ddots & \vdots   \\
            \vdots & \vdots & \ddots& \ddots & \zero  \\
            \AA_{EC} & \zero & \cdots & \zero & \AA_{EE}
        } \in \MS{\pr{k\abs{E} + \abs{C}}}{\pr{k\abs{E} + \abs{C}}},
    }
    where the repetition numbers of the blocks $\AA_{CE}, \AA_{EC}, \AA_{EE}$ are $k$;

    \item $\repp{k, C, \AA, N}$ is a larger matrix by appending all-zeros rows and columns to $\rep{k, C, \AA}$:
    \eq{
        \repp{k, C, \AA, N} =
        \mx{
            \rep{k, C, \AA}  & \zero \\
            \zero  & \zero
        }\in \MS{N}{N},
    }
    where $N \geq k\abs{E} + \abs{C}$ is used to indicate the size of $\repp{k, C, \AA, N}$  ;

    \item if $F, F_+$ is a partition of $E$, $\repFC{k, C, F, \AA}$ is defined as a permutation of the $k$-\repmat\ of $\AA$, which has the following form:
        \eq{
            \repFC{k, C, F, \AA} =
            \mx{
                k\AA_{CC} & \AA_{CF} & \cdots & \AA_{CF} & \AA_{C F^+} & \cdots & \AA_{C F^+}  \\
                \AA_{FC} & \AA_{FF} & & & & & \\
                \vdots & & \ddots & & & &  \\
                \AA_{FC} & & & \AA_{FF} & & & \\
                \AA_{F^+ C} & & & & \AA_{F^+ F^+} & & \\
                \vdots & & & & & \ddots & \\
                \AA_{F^+ C} & & & & & & \AA_{F^+ F^+ }
            }.
        }

    \end{enumerate}

\end{definition}

Now, we define the matrices $\dr{\XL{k}}_{0\leq k\leq K}$ which are used to bound $\Mt{0, k} - \Mtt{0, k} $ and then $\Lt{k} - \Ltt{k}  $.
We define $\dr{\XL{k}}_{0\leq k\leq K}$ iteratively together with the error quantities $\dr{\gam_k}_{0\leq k\leq K}$.

We start from $\XL{0} = \U{\LL}$ and $\gam_0 = 0$.
If we have defined $\dr{\XL{i}}_{0\leq i < k}$, $\dr{\gam_i}_{0\leq i < k}$, then
 $\XL{k}\in \MS{\pr{\puts{k}}}{\pr{\puts{k}}}  $ and $\gam_k \in \Real_+ $ are defined as follows:
\eql{\label{eq:newday1}}{
    \XL{k} \defeq & \frac{k}{4k + \frac{2k}{\alp} + \sum_{i=0}^{k-1}\gam_i}
    \U{\Mt{0,k}}  \\
    & + \frac{1}{4k + \frac{2k}{\alp} + \sum_{i=0}^{k-1}\gam_i}
    \sum_{i=0}^{k-1}  \gam_i\rep{2^{k-i},  C, F, \XL{i} }
     \\
    & + \frac{{3 + \frac{2}{\alp}}}{4k + \frac{2k}{\alp} + \sum_{i=0}^{k-1}\gam_i}
     \sum_{i=0}^{k-1} \rep{2^{k-i},  C, F, \U{\Mt{0,i}} },
}
\eql{\label{eq:defgamk}}{
          \gam_k \defeq \sup_{\xx, \yy\notin \ker\pr{\sc{\XL{k }, -[n]}}} \frac{2\xx\tp\pr{\Ltt{k} - \Lt{k}}\yy}{\xx\tp\sc{\XL{k}, -[n]}\xx + \yy\tp\sc{\XL{k}, -[n]}\yy } ,
}
\begin{remark}
    The first term $\U{\Mt{0, k}}$ in~\eqref{eq:newday1} is only used to guarantee that \\ $\XL{k} \pgeq \frac{k }{4k + \frac{2k}{\alp} + \sum_{i=0}^{k-1}\gam_i}\U{\Mt{0,k}}$ (Fact~\ref{fact:frMti1}\eqref{enum:Q2}).
    Without this term,~\eqref{eq:XLb1} still holds with only slight changes in the constants.
\end{remark}

\begin{remark}
We use the more complicated definition of $\gam_k$ as
in~\eqref{eq:defgamk} because the relations between the kernels of
$\sc{\XL{k}, -[n]}$ and $\Ltt{k} - \Lt{k}$ is only shown in
Section~\ref{sec:Schurcplstable2}.
In Lemma~\ref{lem:scLttyes}, we show $\gam_k < +\infty$.
Then, by Fact~\ref{lem:ne},
we obtain the simplification,
    \eq{ \gam_k = \ndd{\sc{\XL{k}, -[n]}}{\pr{\Ltt{k} - \Lt{k}}}.}

\end{remark}

The following elementary properties of $\XL{k}$ follows directly by Lemma~\ref{lem:Mti}, Fact~\ref{fact:DLD} and the fact that the coefficients on the RHS of~\eqref{eq:newday1} equals $1$.
\begin{fact}\label{fact:frMti1}
   $\XL{k}$ is a Laplacian satisfying:
    \begin{enumerate}[(i)]
      \item $\U{\Mt{0,k}} \pleq \pr{4 + \frac{2}{\alp} + \frac{\sum_{i=0}^{k-1}\gam_i}{k}}\XL{k}  $; \label{enum:Q2}
      \item $\Diag{\XL{k}} = \Diag{\Mt{0,k}} $; \label{enum:Q3}
      \item $\XL{k}_{ - C, - C } $, $\XL{k}_{-[n], -[n]}$ are $\alp$-RCDD; \label{enum:Q4}
      \item $\nt{\pr{\XL{k}}^{1/2}\Diag{\Mt{0,k}}^{-1/2}}^2 \leq 2 $; \label{enum:Q5}

    \end{enumerate}
\end{fact}

\begin{lemma}\label{enum:Q6}
     $\sc{\XL{k}, - C } \pleq 2^{k}\U{\sc{\LL, F}} $.
\end{lemma}

The following lemma shows how we bound the sparsification errors attached to $\Mtt{0, k}$ by $\XL{k}$.
\begin{lemma}\label{lem:Ersepsz1}

    \eql{\label{eq:XLb1}}{
        \Ers{k}
        \aleq \epsz \pr{4k + \frac{2k}{\alp} + \sum_{i=0}^{k-1}\gam_i} \XL{k}.
    }

\end{lemma}
\begin{proof}
We only prove the case when all $\gam_k < +\infty$, and the proof for the case when some $\gam_k = +\infty$ follows trivially.\footnote{Actually, we will show in Lemma~\ref{lem:scLttyes} that all $\gam_k < +\infty$.  }

We define $\XL{0} = \U{\LL}$,
and will construct $\XL{k}$ iteratively.
That is, we want to show that given
$\XL{0}, \cdots, \XL{k-1}$ satisfying these conditions,
we can find $\XL{k}$.

  Firstly, we fix some $i\in \dr{0, 1, \cdots, k-1}$.
  By Lemma~\ref{lem:EYEXEtt} and Fact~\ref{lem:ne}, we have
  \eq{
    2\xx\tp \Ett{i+1} \yy \leq \epsz \pr{\xx\tp \U{\Ltt{i}} \xx + \yy\tp \U{\Ltt{i}} \yy},\
\qquad \forall \xx, \yy\in \Real^n.
  }

Then, for any
$\xx, \yy\in \Real^{\puts{k}},  $
by the definition of $\Errt{i+1, k}$ in~\eqref{eq:defErrtik1}, we have
\eq{
    &2\xx\tp \Errt{i+1, k} \yy
    = 2\sum_{a=1}^{2^{k-i-1} } \xx\tp \unierr{k, a, \Ett{i+1}} \yy
    = 2\sum_{a=1}^{2^{k-i-1} } \vc{\xx_C \\ \xx_{F_a}}\tp \Ett{i+1} \vc{\yy_C \\ \yy_{F_{\psit{k - i - 1}\pr{a}}}}  \\
    \leq&  \epsz \sum_{a=1}^{2^{k-i-1} } \pr{ \vc{\xx_C \\ \xx_{F_a}}\tp \U{\Ltt{i}} \vc{\xx_C \\ \xx_{F_{a}}} + \vc{\yy_C \\ \yy_{F_{\psit{k-i-1}\pr{a}}}}\tp \U{\Ltt{i} } \vc{\yy_C  \\ \yy_{F_{\psit{k-i-1}\pr{a}}}} }.
}
By the fact that  $\psit{k-i-1}\pr{\cdot}$ is a bijection from $[2^{k-i-1}]$ to $[2^{k-i-1} ]$ and $\U{\Ltt{i}}$ is PSD (because $\Ltt{i}$ is an Eulerian Laplacian from Lemma~\ref{lem:LapetcpropSparSchurCpmt1}),
we have
  \eql{\label{eq:xEy127}}{
    &2\xx\tp \Errt{i+1, k} \yy  \\
    \leq& \epsz \sum_{a=1}^{2^{k-i-1} } \pr{ \vc{\xx_C \\ \xx_{F_a}}\tp \U{\Ltt{i}} \vc{\xx_C \\ \xx_{F_{a}}} + \vc{\yy_C \\ \yy_{F_{a}}}\tp \U{\Ltt{i} } \vc{\yy_C  \\ \yy_{F_{a  }}} }  \\
    \leq& \epsz \sum_{a=1}^{2^{k-i} } \pr{ \vc{\xx_C \\ \xx_{F_a}}\tp \U{\Ltt{i}} \vc{\xx_C \\ \xx_{F_{a}}} + \vc{\yy_C \\ \yy_{F_{a}}}\tp \U{\Ltt{i} } \vc{\yy_C  \\ \yy_{F_{a  }}} }  \\
    =& \epsz \pr{\xx\tp \repp{2^{k-i}, C, \U{\Ltt{i}}, \puts{k}}\xx  + \yy\tp \repp{2^{k-i}, C, \U{\Ltt{i}}, \puts{k}}\yy  }.
  }

  To bound the repetition matrix $\repp{2^{k-i}, C, \U{\Ltt{i}}, \puts{k}}$, we bound the matrices
  $\repp{2^{k-i}, C, \U{\Ltt{i} - \Lt{i}}, \puts{k}}$ and $\repp{2^{k-i}, C, \U{\Lt{i}}, \puts{k}} $, respectively.

  By the definition of $\gam_i $ in~\eqref{eq:defgamk},
  \eq{
    \U{\Ltt{i} - \Lt{i}} \pleq \gam_i \sc{\XL{i}, -[n] }.
  }
  Denote the set $E_k = C \cup \pr{\cup_{a=1}^{2^{k}}F_a} $.
  It is straightforward that
  \eq{
    \sc{\repFC{2^{k-i}, C, F, \XL{i}}, - E_{k-i} }
    = \rep{2^{k-i}, C, \sc{\XL{i}, -[n]} }.
  }
  Then, by Fact~\ref{fact:repprvpleq} and Fact~\ref{fact:scz}, we have
  \eql{\label{eq:XLLtt-Lt1}}{
    &\repp{2^{k-i}, C, \U{\Ltt{i} - \Lt{i}}, \puts{k}} \\
    \pleq& \gam_i \repp{2^{k-i}, C, \sc{\XL{i}, -[n]}, \puts{k}} \\
    \pleq& \gam_i \rep{2^{k-i}, C, F, \XL{i}}.
  }

  By Lemma~\ref{lem:Mti}, $\Mt{0, i}_{-[n], -[n] }$ is $\alp$-RCDD and $\sc{\Mt{0,i}, -[n]} = \Lt{i} $.
  Then, by Fact~\ref{lem:scrobust}, we have
  \eq{
    \U{\Lt{i}} = \U{\sc{\Mt{0, i}, -[n]}} \pleq \pr{3 + \frac{2}{\alp}}\sc{\U{\Mt{0, i}}, -[n]}.
  }
  Similar to~\eqref{eq:XLLtt-Lt1}, we also have
  \eql{\label{eq:XLLt1}}{
    &\repp{2^{k-i}, C, \U{\Lt{i}}, \puts{k}}
    \pleq \pr{3 + \frac{2}{\alp}} \rep{2^{k-i}, C, F, \U{\Mt{0, i}}}.
  }

Then, by substituting~\eqref{eq:XLLtt-Lt1},~\eqref{eq:XLLt1}
into~\eqref{eq:xEy127}, summing over $i=0, 1, \cdots, k-1$,
and combining with~\eqref{eq:newday1} and Fact~\ref{lem:ne}, we have~\eqref{eq:XLb1}.

\end{proof}

\section{Robustness of Schur Complements and Full Error Analysis}\label{sec:Schurcplstable2}

In this section, we show additional robustness properties
of Schur complements suitable for analyzing errors on the
augmented matrices.
Specifically, we establish conditions on $\AA, \BB, \UU$
where $\AA - \BB \aleq \eps\cdot \UU $,
as well as the set to be eliminated, $F$,
so that $\sc{\AA, F} - \sc{\BB, F} \aleq \dlt\cdot \sc{\UU, F} $.

Using these properties, we bound the norms of errors in
Schur complements of the $\XL{k}$ and $\gam_k $.
Such bounds allow us to complete the proof of
Theorem~\ref{thm:SparSchur}.

\subsection{Schur Complement Robustness }
The next lemma is a perturbed version of Fact~\ref{fact:LDL}. 
It is used to prove Lemma~\ref{lem:schurdUL1} below.
\begin{lemma}\label{lem:pertbLDL}
    Suppose that $\LL\in \MatSize{n}{n}$ is an Eulerian Laplacian, $\DD = \Diag{\LL}$, $\WW$ is PSD, $\nt{\WW^{1/2}\DD^{-1/2}} \leq a $,
    and the matrix $\EE\in \MatSize{n}{n}$ satisfies $\EE \aleq b\WW$
    with $a^2 b < 2$.
    Then the matrix $\MM = \LL + \EE$ satisfies:  
    \eq{
        &\MM\tp\DD\inv\MM \pleq \frac{1}{2 - a^2b}\pr{\pr{4 + 2a^2b}\U{\LL} + 2b \WW}, \\
        &\MM\DD\inv\MM\tp \pleq \frac{1}{2 - a^2b}\pr{\pr{4 + 2a^2b}\U{\LL} + 2b \WW}.
    }

\end{lemma}
\begin{proof}
  For any $\xx \in \Real^n$,
\eq{
\xx\tp\MM\tp\DD\inv\MM\xx = \xx\tp\LL\tp\DD\inv\LL\xx  + \xx\tp\LL\tp\DD\inv\EE\xx + \xx\tp\EE\tp\DD\inv\MM\xx.
}
We bound each of these terms on the RHS separately.
\begin{itemize}
\item For the first term, by Fact~\ref{fact:LDL},
  $
    \xx\tp\LL\tp\DD\inv\LL\xx \leq 2\xx\tp \U{\LL} \xx.
  $
\item For the second term,
  by Fact~\ref{lem:ne},
  the conditions $\EE \aleq b \WW $ and $\nt{\WW^{1/2  }\DD^{-1/2}} \leq a $ and Fact~\ref{fact:LDL},
  \eq{
    &2\xx\tp\LL\tp\DD\inv\EE\xx \leq b\pr{\xx\tp \WW \xx + \xx\tp \LL\tp \DD\inv \WW \DD\inv \LL \xx}  \\
    =& b\pr{\xx\tp \WW \xx + \nt{\WW^{1/2}\DD\inv\LL\xx}^2  }
    \leq  b\pr{\xx\tp \WW \xx + \nt{\WW^{1/2}\DD^{-1/2}}^2 \nt{\DD^{-1/2}\LL\xx}^2  }  \\
    \leq&  b\pr{\xx\tp \WW \xx + a^2 \xx\tp \LL\tp \DD\inv \LL \xx  }
    \leq b\pr{\xx\tp \WW \xx + 2 a^2 \xx\tp \U{\LL} \xx  }.
  }
\item For the third term,
using the conditions $ \EE \aleq b\WW       $ and $\nt{\WW^{\dagger/2}\DD^{-1/2}} \leq a $ again yields that
  \eq{
     &2\xx\tp\EE\tp\DD\inv\MM\xx \leq b\pr{\xx\tp \WW \xx + \xx\tp \MM\tp \DD\inv \WW \DD\inv \MM \xx  } \\
    =& b\pr{\xx\tp \WW \xx + \nt{\WW^{1/2}\DD\inv\MM\xx}^2 } \leq b\pr{\xx\tp \WW \xx + a^2\nt{\DD^{-1/2}\MM\xx}^2 }  \\
    =& b\pr{\xx\tp \WW \xx + a^2 \xx\tp \MM\tp \DD\inv\MM\xx}.
  }
  By combining the above equations, we have
  \eq{
    2\xx\tp\MM\tp\DD\inv\MM\xx \leq  \pr{4 + 2a^2b}\xx\tp\U{\LL} \xx + 2b \xx\tp\WW\xx + a^2 b \xx\tp \MM\tp\DD\inv\MM \xx.
  }
\end{itemize}

  Rearranging the above equations yields that
  \eq{
    \xx\tp\MM\tp\DD\inv\MM\xx \leq  \frac{1}{2 - a^2b}\pr{\pr{4 + 2a^2b}\xx\tp\U{\LL}\xx + 2b \xx\tp\WW\xx}.
  }
  which gives the result for $\MM\tp\DD\inv\MM$.
  The bound for $\MM\DD\inv\MM\tp$ follows analogously.
\end{proof}

The following lemma shows the robustness of the Schur complements.
It is used in the proof of Lemma~\ref{lem:scLttyes}
to bound $\gam_k$.

\begin{lemma}\label{lem:schurdUL1}
    Let  $\NN\in \MatSize{n}{n}$ be an Eulerian Laplacian, let $\MM$ be an $n$-by-$n$ matrix, let $\UU \in\MatSize{n}{n}$ be PSD and $F, C$ a partition of $[n]$.
    Suppose that
    $\UU_{FF}  $ is nonsingular,
    $\UU\one = \zero$,
    $\NN_{FF}$ is $\rho$-RCDD $(\rho > 0)$,
    $\U{\NN}_{FF} \pgeq \frac{1}{\mu}\UU_{FF}$,
    $\U{\NN} \pleq \beta \UU$,
    $\nt{\UU^{1/2}\Diag{\NN}^{-1/2}} \leq a$,
    and the matrix  $\EE = \MM - \NN$  satisfies $\EE\aleq b \cdot \UU $
    with $b < \min\dr{\frac{2}{a^2}, \frac{1}{\mu}}$.

    Then, $\MM_{FF}$, $\NN_{FF}$ are nonsingular and
    \eql{\label{eq:scdiff1}}{
        {\sc{\MM, F} - \sc{\NN, F}} \aleq  b \pr{1 + \frac{1}{\rho}}\frac{\mu\pr{\bet\pr{4 + 2a^2 b} + 2b}}{\pr{1 - \mu b}^2\pr{2 - a^2 b}} \cdot \sc{\UU, F}.
    }

\end{lemma}
\begin{proof}
    Without loss of generality, we assume $F = \dr{1, \cdots, \abs{F}}$ and $C = [n]\backslash F$ in this proof.

    By the condition $\MM - \NN \aleq b\cdot \UU $, we have $2\xx\tp\pr{\MM - \NN}\yy \leq b\pr{\xx\tp\UU\xx + \yy\tp\UU\yy},\ \forall \xx\in \Real^n $.
    Then, $\xx\tp\pr{\U{\MM}_{FF} - \U{\NN}_{FF}}\xx \leq b\xx\tp\UU_{FF}\xx,\ \forall \xx\in \Real^{\abs{F}}. $
    Thus, $\U{\MM}_{FF} \pgeq \U{\NN}_{FF} - b\UU_{FF} $.
    By the condition $\U{\NN}_{FF} \pgeq \frac{1}{\mu}\UU_{FF}$, we have
    \eql{\label{eq:MFFpgeqNFF1}}{
        \U{\MM}_{FF} \pgeq \U{\NN}_{FF} - b\UU_{FF} \pgeq \pr{1 - b\mu}\U{\NN}_{FF}
    }
    and
    \eql{\label{eq:MFpgeqUFF1}}{
        \U{\MM}_{FF} \pgeq \pr{1 - b\mu}\U{\NN}_{FF} \pgeq \pr{\frac{1}{\mu} - b}\UU_{FF} \succ \zero.
    }
    Since $\UU$ is PSD, $\UU_{FF}$ is also PSD. Since $\UU_{FF}$ is also nonsingular, $\UU_{FF} \succ \zero$.
    Then, by the condition $b < \frac{1}{\mu} $ and~\eqref{eq:MFpgeqUFF1}, we have $\U{\MM}_{FF} \succ \zero $.
    Then, by Fact~\ref{fact:kerAsmltkUAaPSD}, $\MM_{FF}$ and $\NN_{FF}$ are nonsingular.
    And by Fact~\ref{fact:LUL}, we have
    \eql{\label{eq:MLUL1}}{\MM_{FF}\inv \U{\MM}_{FF} \tpp{\MM_{FF}\inv} \pleq \U{\MM}_{FF}\inv,  }
    \eql{\label{eq:NUNNbyfactLUL}}{\tpp{\NN_{FF}\inv} \U{\NN}_{FF} \NN_{FF}\inv \pleq \U{\NN_{FF}}\inv. }

  For any $\xx, \yy\in \Real^{\abs{C}}$, define
  \eq{
    \xn =
    \vc{
        - \NN_{FF}\inv\NN_{FC}\xx \\
        \xx
    },
    \quad
    \xu =
    \vc{
        - \UU_{FF}\inv\UU_{FC}\xx  \\
        \xx
    }
  }
  and
  \eq{
    \ym =
    \vc{
        - \tpp{\MM_{FF}\inv}\MM_{CF}\tp \yy  \\
        \yy
    },
    \quad
    \yu =
    \vc{
        - {\UU_{FF}\inv}\UU_{CF}\tp\yy  \\
        \yy
    }.
  }
  Then,
  \eq{
    \NN\xn =
    \vc{
        \zerov{F}  \\
        \sc{\NN, F}\xx
    },\quad
    \ym\tp\MM =
    \vc{
        \zerov{F}\tp & \yy\tp\sc{\MM, F}
    }.
  }
  Thus, we have
  \eq{
    \ym\tp\MM\xn = \vc{\zerov{F}\tp & \yy\tp\sc{\MM, F} } \vc{ - \NN_{FF}\inv\NN_{FC}\xx  \\ \xx } = \yy\tp\sc{\MM, F}\xx.
  }
  Similarly,
  $
    \ym\tp \NN \xn = \yy\tp\sc{\NN, F}\xx.
  $
  Therefore,
    \eq{
        \yy\tp\pr{\sc{\MM, F} - \sc{\NN, F}}\xx = \ym\tp\pr{\MM - \NN}\xn.
    }
    Combining the above equation with $\EE \aleq b \cdot \UU $ and Fact~\ref{lem:ne} yields that
    \eq{
        2\yy\tp\pr{\sc{\MM, F} - \sc{\NN, F}}\xx \leq b \pr{\ym\tp\UU\ym + \xn\tp\UU\xn  }.
    }
    Denote the projection matrix onto the image space of $\NN$ by $\Con $.
   It follows by direct calculations that $\pr{\xn}_F + \UU_{FF}\inv\UU_{FC}\xx = - \NN_{FF}\inv\pr{\NN\xu}_F = - \NN_{FF}\inv\pr{\Con\NN\xu}_F $.

  We denote the matrix
  \eq{
    \PP = \Con
    \mx{
        \tpp{\NN_{FF}\inv} \UU_{FF} \NN_{FF}\inv & \zerom{F}{C}  \\
        \zerom{C}{F} & \zerom{C}{C}
    } \Con
  }
  in this proof.
  Then by Fact~\ref{fact:Schurxusmall}, we have
  \eql{\label{eq:xnUP1}}{
    &\xn\tp\UU\xn = \nA{\sc{\UU, F}}{\xx}^2 +  \nA{\UU_{FF}}{{\NN_{FF}}\inv\pr{\Con\NN\xu}_F}^2 = \nA{\sc{\UU, F}}{\xx}^2 +  \nA{\PP}{\NN\xu}^2.
  }

  By the condition $\U{\NN_{FF}} \pgeq \frac{1}{\mu}\UU_{FF}$, we have
  \eql{\label{eq:NFUNbyNFbU}}{\tpp{\NN_{FF}\inv} \UU_{FF} \NN_{FF}\inv \pleq \mu\tpp{\NN_{FF}\inv} \U{\NN_{FF}} \NN_{FF}\inv. }

  Since $\NN_{FF}$ is $\rho$-RCDD, we have $\U{\NN_{FF}}  \pgeq \frac{\rho}{1 + \rho}\Diag{\NN}_{FF}$, i.e.,
  \eql{\label{eq:DiagUNFF1}}{\U{\NN_{FF}}\inv  \pleq \pr{1 + \frac{1}{\rho}}\Diag{\NN}_{FF}\inv.  }

  Then, combining~\eqref{eq:NFUNbyNFbU},~\eqref{eq:NUNNbyfactLUL},~\eqref{eq:DiagUNFF1} with  Fact~\ref{fact:LDL} yields that
  \eql{\label{eq:Pmorning1}}{
    &\PP
    \pleq
    \pr{1 + \frac{1}{\rho}} \mu \Con
    \mx{
        \Diag{\NN}_{FF}\inv  & \zerom{F}{C}  \\
        \zerom{C}{F} & \zero{C}{C}
    } \Con
    \pleq
    \pr{1 + \frac{1}{\rho}} \mu
    \Con\Diag{\NN}\inv\Con  \\
    =& \pr{1 + \frac{1}{\rho}}\mu \pr{\NN\dg}\tp\NN\tp\Diag{\NN}\inv\NN\NN\dg
    \pleq 2\pr{1 + \frac{1}{\rho}}\mu \pr{\NN\dg}\tp \U{\NN} \NN\dg   \\
    \pleq& 2\pr{1 + \frac{1}{\rho}}\mu \bet \pr{\NN\dg}\tp \UU \NN\dg.
  }

  We also define the matrix
  \eq{
    \QQ = \Contil
    \mx{
        \MM_{FF}\inv \UU_{FF} \tpp{\MM_{FF}\inv} & \zerom{F}{C}  \\
        \zerom{C}{F} & \zerom{C}{C}
    }
    \Contil
  }
  in this proof, where $\Contil$ is the projection matrix onto the image space of $\MM\tp$.
  Similar to~\eqref{eq:xnUP1}, we have
  \eql{\label{eq:ymUQ1}}{
    \ym\tp \UU \ym = \nA{\sc{\UU, F}}{\yu}^2 + \nA{\QQ}{\MM\tp\yu}^2.
  }

Then, combining~\eqref{eq:MFpgeqUFF1},~\eqref{eq:MLUL1},~\eqref{eq:MFFpgeqNFF1},~\eqref{eq:DiagUNFF1} yields that
\eql{\label{eq:QMc1}}{
    &\MM_{FF}\inv \UU_{FF} \tpp{\MM_{FF}\inv}
    \pleq \frac{\mu}{1 - \mu b} \MM_{FF}\inv \U{\MM}_{FF} \tpp{\MM_{FF}\inv}
    \pleq \frac{\mu}{1 - \mu b}\U{\MM}_{FF}\inv  \\
    \pleq& \frac{\mu}{\pr{1 - \mu b}^2}\U{\NN}_{FF}\inv
    \pleq \frac{\mu}{\pr{1 - \mu b}^2}\pr{1 + \frac{1}{\rho}}\Diag{\NN}_{FF}\inv.
}

  By~\eqref{eq:QMc1}, we have
  \eq{
    \QQ
    &\pleq
    \frac{\mu}{\pr{1 - \mu b}^2}\pr{1 + \frac{1}{\rho}} \Contil
    \mx{
        \Diag{\NN}_{FF}\inv  & \zerom{F}{C}  \\
        \zerom{C}{F} & \zerom{C}{C}
    } \Contil
    \pleq
    \frac{\mu}{\pr{1 - \mu b}^2}\pr{1 + \frac{1}{\rho}}
    \Contil\Diag{\NN}\inv\Contil  \\
    &= \frac{\mu}{\pr{1 - \mu b}^2}\pr{1 + \frac{1}{\rho}} \MM\dg\MM\Diag{\NN}\inv\MM\tp\pr{\MM\dg}\tp.
  }

  Then, using Lemma~\ref{lem:pertbLDL} and $\U{\NN} \pleq \beta \UU$, we have
  \eql{\label{eq:QM511}}{
    \QQ &\pleq \frac{\mu}{\pr{1 - \mu b}^2}\pr{1 + \frac{1}{\rho}}\MM\dg\pr{\frac{1}{2 - a^2b}\pr{\pr{4 + 2a^2b}\U{\NN} + 2b \UU}}\pr{\MM\dg}\tp  \\
    &\pleq \pr{1 + \frac{1}{\rho}}\frac{\mu\pr{\bet\pr{4 + 2a^2 b} + 2b}}{\pr{1 - \mu b}^2\pr{2 - a^2 b}} \MM\dg \UU \tpp{\MM\dg}.
  }

  With the above preparations, our proof for~\eqref{eq:scdiff1} has 2 steps.
  Firstly, we prove~\eqref{eq:scdiff1} with an additional condition: $\ker\pr{\NN\tp}\cup\ker\pr{\MM} \sleq \ker\pr{\UU}$.
  Then, we remove this extra condition by taking limits.

  To begin with, we prove~\eqref{eq:scdiff1} under the condition $\ker\pr{\NN\tp}\cup\ker\pr{\MM} \sleq \ker\pr{\UU}$.
  Since $\Con$ is the projection matrix onto the image space of $\NN$ and $\ker\pr{\NN\tp}\sleq \ker{\UU}$, we have
  $\Con\UU\Con = \UU. $
  Then, by~\eqref{eq:Pmorning1}, we have
  $\nA{\PP}{\NN\xu}^2 \leq 2\pr{1 + \frac{1}{\rho}}\mu\bet \xu\tp \NN\tp\tpp{\NN\dg}\UU\NN\dg\NN \xu = 2\pr{1 + \frac{1}{\rho}}\mu\bet \xu\tp\Con\UU\Con\xu = 2\pr{1 + \frac{1}{\rho}}\mu\bet \xu\tp\UU\xu.    $

  Then, combining with~\eqref{eq:xnUP1} yields that
  \eq{
    \xn\tp \UU \xn 
    \leq \nA{\sc{\UU, F}}{\xx}^2 + 2\pr{1 + \frac{1}{\rho}}\mu \bet  \xu\tp \UU \xu = \pr{1 + 2\pr{1 + \frac{1}{\rho}}\mu \bet} \xx\tp \sc{\UU, F} \xx.
  }

  Analogously, we have $\Contil\UU\Contil = \UU$.
  Then, by~\eqref{eq:QM511} and~\eqref{eq:ymUQ1}, we have
  \eq{
    \ym\tp \UU \ym \leq \pr{1 + \pr{1 + \frac{1}{\rho}}\frac{\mu\pr{\bet\pr{4 + 2a^2 b} + 2b}}{\pr{1 - \mu b}^2\pr{2 - a^2 b}}} \yy\tp \U{\sc{\UU, F}} \yy.
  }

Combining the above equations yields that
when $\ker\pr{\NN\tp}\cup\ker\pr{\MM} \sleq \ker\pr{\UU}$,
  \eql{\label{eq:kercupadditional1}}{
    &2\yy\tp \pr{\sc{\MM, F} - \sc{\NN, F}} \xx  \\
    \leq&  b \pr{\pr{1 + 2\pr{1 + \frac{1}{\rho}}\mu{\bet}}\xx\tp\sc{\UU, F}\xx + \pr{1 + \pr{1 + \frac{1}{\rho}}\frac{\mu\pr{\bet\pr{4 + 2a^2 b} + 2b}}{\pr{1 - \mu b}^2\pr{2 - a^2 b}}}\yy\tp\sc{\UU, F}\yy}  \\
    \leq& b \pr{1 + \pr{1 + \frac{1}{\rho}}\frac{\mu\pr{\bet\pr{4 + 2a^2 b} + 2b}}{\pr{1 - \mu b}^2\pr{2 - a^2 b}}}\pr{\xx\tp\sc{\UU, F}\xx + \yy\tp\sc{\UU, F}\yy},\ \forall \xx, \yy\in \Real^{\abs{C}}.
  }

Next, we remove the condition $\ker\pr{\NN\tp}\cup\ker\pr{\MM} \sleq \ker\pr{\UU}$.
The PSD matrix $\UU$ has the spectral decomposition as $\UU = \sum_{i=1}^{n} \la_{i}\pr{\UU} \zz_i\zz_i\tp$.
Since $\UU\one = \zero$, without loss of generality, we may assume $\zz_1 = \one$.  
Denote $d = {\rm rank}\pr{\UU}$.
Then, by adding small symmetric perturbations of the form $\sum_{i=2}^{n-d}\dlt^{\pr{i}}\zz_i\zz_i\tp \ (\dlt^{(i)} > 0)$, we can have a sequence of PSD matrices $\dr{\Uzoo{j}}_{j\geq 0}$ such that for each $j \geq 0$, $\ker\pr{\Uzoo{j}} = \spanrm{\one} $, $\Uzoo{j} \succ \UU$  and $\lim_{j\arr +\infty} \Uzoo{j} = \UU$.

By adding undirected edges with small weights in $\NN$, we can find a sequence of Eulerian Laplacians $\dr{\Nzoo{j}}_{j\geq 0}$ such that each
$\Nzoo{j}$ is strongly connected,
$\Nzoo{j}_{FF}$ is nonsingular and
$\lim_{j\arr +\infty} \Nzoo{j} = \NN.  $
Then, by Fact~\ref{fact:strcrank1},
\eq{ \ker\pr{\Nzoo{j}} = \ker\pr{\tpp{\Nzoo{j}}} = \spanrm{\one} = \ker(\Uzoo{j}).  }
We can let the weights of the new undirected edges in $\NN$ tend to zero faster than $\Uzoo{j} - \UU$.
Then,
\eq{
    \lim_{j\arr +\infty} \ndd{\pr{\Uzoo{j}}}{\pr{\NN - \Nzoo{j}}} = 0.
}

Since $\EE\aleq b\UU$ and $\UU\one = \zero$, we have $\EE\one = \EE\tp\one = \zero$.
Thus, $\MM\one = \MM\tp\one = \zero$.
By adding small symmetric perturbations of the form $\sum_{i=2}^{n } \widehat{\dlt}^{(i)}\zz_i\zz_i\tp$ into $\MM$,
we can find a sequence of matrices $\{\Mzoo{j}\}_{j\geq 0} $
such that $\Mzoo{j}_{FF}$ is nonsingular, $\lim_{j\arr +\infty} \Mzoo{j} = \MM $ and \eq{\ker(\Mzoo{j}) = \ker\pr{\tpp{\Mzoo{j}}} = \spanrm{\one} = \ker(\Uzoo{j}).  }
Also, by setting the perturbations added into $\MM$ to be small enough (with respect to the magnitudes of the perturbations in $\UU $), we can let
\eq{
    \lim_{j\arr +\infty} \ndd{\pr{\Uzoo{j}}}{\pr{\MM - \Mzoo{j}}} = 0.
}
Define $\bzoo{j} = b + \ndd{\pr{\Uzoo{j}}}{\pr{\NN - \Nzoo{j}}} + \ndd{\pr{\Uzoo{j}}}{\pr{\MM - \Mzoo{j}}}$.
Then, combining the above equations and using the relation $\Uzoo{j} \succ \UU $ and $\MM - \NN \aleq b\UU $, we have
$\lim_{j\arr +\infty} \bzoo{j} = b  $ and
\eq{
    \Mzoo{j} - \Nzoo{j} \aleq \bzoo{j} \Uzoo{j}.
}

As the perturbations mentioned above tend to zero,
 we can also define real numbers \\
 $\dr{\azoo{j},  \rhozoo{j}, \muzoo{j}, \betzoo{j}}_{j\geq J}$ easily such that
the matrix sequence $\{\Mzoo{j}, \Nzoo{j}, \Uzoo{j}\}_{j\geq 0}$ satisfy
\begin{itemize}
\item $\Nzoo{j}_{FF}  $ is $\rhozoo{j}$-RCDD $(\rhozoo{j} > 0)$,
\item $\U{\Nzoo{j}}_{FF} \pgeq \frac{1}{\muzoo{j}}\Uzoo{j}_{FF}     $,
\item $\U{\Nzoo{j}} \pleq \betzoo{j} \Uzoo{j} $,
\item $\nt{\pr{\Uzoo{j}}^{1/2}\Diag{\Nzoo{j}}^{-1/2}} \leq \azoo{j}  $,

\end{itemize}
for any $j \geq J$,
and $\azoo{j}$, $\rhozoo{j}$, $\muzoo{j}$, $\betzoo{j}$  tend to $a$,  $\rho$, $\mu$, $\bet$,   respectively.

Since $b < \min\dr{\frac{2}{a^2}, \frac{1}{\mu}}$ and $\bzoo{j}, \azoo{j}, \muzoo{j}$ tend to $b, a, \mu$ respectively, then, there exists a $J' > 0$ such that for any $j \geq J'$, $\bzoo{j} < \min\dr{\frac{2}{\pr{\azoo{j}}^2}, \frac{1}{\muzoo{j}} }.  $

Then by~\eqref{eq:kercupadditional1}, we have for any $ j\geq \max\dr{J, J'}$ and $\xx, \yy\in \Real^{\abs{C}}$,
\eql{\label{eq:perturbedscr1fknl1}}{
    &2\yy\tp \pr{\sc{\Mzoo{j}, F} - \sc{\Nzoo{j}, F}} \xx  \\
    \leq& \bzoo{j} \pr{1 + \pr{1 + \frac{1}{\rhozoo{j}}}\frac{\muzoo{j}\pr{\betzoo{j}\pr{4 + 2(\azoo{j})^2 \bzoo{j}} + 2\bzoo{j}}}{\pr{1 - \muzoo{j} \bzoo{j}}^2\pr{2 - (\azoo{j})^2 \bzoo{j}}}}\pr{\xx\tp\sc{\Uzoo{j}, F}\xx + \yy\tp\sc{\Uzoo{j}, F}\yy}.
}
Since $\MM_{FF} $ is nonsingular, we have $\lim_{j\arr +\infty} \iv{\Mzoo{j}_{FF}} = \MM_{FF}\inv$.
Thus, $\lim_{j\arr +\infty} \sc{\Mzoo{j}, F} = \sc{\MM, F}$.
Analogously, $\lim_{j\arr +\infty} \sc{\Uzoo{j}, F} = \sc{\UU, F}$, $\lim_{j\arr +\infty} \sc{\Nzoo{j}, F} = \sc{\NN, F}$.
Then, taking limits on both sides of~\eqref{eq:perturbedscr1fknl1} and combining with Fact~\ref{lem:ne} lead to~\eqref{eq:scdiff1}.
\end{proof}

\subsection{Inductive Accumulation of Errors}

We can now obtain a relatively tight bound for
$\sc{\Ltt{K}, F} - \sc{\Lt{K}, F}$
by bounding $\gam_k$ iteratively.

\begin{lemma}\label{lem:scLttyes}
    For any $\dlt_0 \in (0, 1)$,
    with a small $\eps = O\pr{\frac{\dlt_0}{K}} $ in Algorithm~\ref{alg:SparSchur},
    the exact and approximate $K$-th partially-block-eliminated Laplacians $\Lt{K}, \Ltt{K}$ satisfies
    \eql{\label{eq:scLttgood}}{
    \frac{1}{2^K}  \pr{\sc{\Ltt{K}, F} - \sc{\Lt{K}, F}} \aleq  O\pr{\dlt_0} \cdot \U{\sc{\LL, F}}.
    }

\end{lemma}
\begin{proof}

  First,
  we will prove $\gam_k \leq O\pr{\dlt_0} \ (\forall 0\leq k\leq K)$ by induction, where $\dr{\gam_k}$ are defined in~\eqref{eq:defgamk}.

  Since $\Ltt{0} = \LL $, we have $\gam_0 = 0$.
  Now, assume $\gam_{i} \leq O\pr{\dlt_0},\ \forall 0\leq i\leq k-1 $, we will show that $\gam_k \leq O\pr{\dlt_0}$.

  By Lemma~\ref{lem:Mti}, $\Mt{0, k}$ is an Eulerian Laplacian. $\Mt{0, k}_{-C, -C}$, $\Mt{0, k}_{-[n], -[n]}$ are $\alp$-RCDD.

  By Fact~\ref{fact:frMti1}, $\XL{k}$ is Laplacian, thus, $\XL{k}$ is PSD and $\XL{k}\one = \zero$.

  By Fact~\ref{fact:frMti1}\eqref{enum:Q4},$\XL{k}_{-C, -C}$, $\XL{k}_{-[n], -[n]}$ are $\alp$-RCDD.
  Thus, $\XL{k}_{-[n], -[n]} \pleq \frac{2 + \alp}{1 + \alp}\Diag{\XL{k}}_{-[n], -[n]} \pleq 2\Diag{\XL{k}}_{-[n], -[n]}. $
  By combining with Fact~\ref{fact:frMti1}\eqref{enum:Q3},  we have
  \eq{
    \U{\Mt{0, k}}_{-[n], -[n]} \pgeq \frac{\alp}{1 + \alp}\Diag{\Mt{0, k}}_{-[n], -[n]} = \frac{\alp}{1 + \alp}\Diag{\XL{k}}_{-[n], -[n]} \pgeq \frac{\alp}{2 + 2\alp}\XL{k}_{-[n], -[n]}.
  }

  By the induction hypothesis, $\sum_{i=0}^{k-1}\gam_i \leq O(k\dlt_0)$.
  Then, by combining with Lemma~\ref{lem:Ersepsz1}, Fact~\ref{fact:frMti1}\eqref{enum:Q2}, Fact~\ref{fact:frMti1}\eqref{enum:Q5} and the condition $\epsz = O\pr{\frac{\dlt_0}{K}}$, we have
  \eq{
    &\U{\Mt{0,k}} \pleq \pr{4 + \frac{2}{\alp} + O(\dlt_0)}\XL{k}  \\
    &\nt{\pr{\XL{k}}^{1/2}\Diag{\Mt{0,k}}^{-1/2}}^2 \leq 2  \\
    &\Ers{k} \aleq \epsz \pr{4k + \frac{2k}{\alp} + k\dlt_0}\cdot \XL{k} \pleq O\pr{\dlt_0}\cdot \XL{k}.
  }

Now, we invoke Lemma~\ref{lem:schurdUL1}
with $\NN:=\Mt{0,k}$, $\MM := \Mtt{0,k}$, $\UU := \XL{k}$, $\EE:= \Ers{k}$, $m:= 2^k \Fn + \Cn$, $F:= [\puts{k}]\dele [n] $, $C:= [n] $, $\rho := \alp$, $\mu := \frac{2 + 2\alp}{\alp}$, $\bet:= 4 + \frac{2}{\alp} + O\pr{\dlt_0} $, $a^2:= 2$, $b:= O(\dlt_0)$.
By  the arguments above,
all the conditions of Lemma~\ref{lem:schurdUL1} are satisfied.
Then,  we have
  \eql{\label{eq:scMFtil0k}}{
    &\sc{\Mtt{0,k}, -[n]} - \sc{\Mt{0,k}, -[n]}  \\
    \aleq& b \pr{1 + \pr{1 + \frac{1}{\rho}}\frac{\mu\pr{\bet\pr{4 + 2a^2 b} + 2b}}{\pr{1 - \mu b}^2\pr{2 - a^2 b}}} \cdot \sc{\XL{k}, -[n]} \\
    \pleq& O\pr{\dlt_0} \cdot \sc{\XL{k}, -[n]}.
  }
  By combining with the fact $\Lt{k} = \sc{\Mt{0, k}, -[n]}$, $\Ltt{k} = \sc{\Mtt{0, k}, -[n]}$, the definition of $\gam_k$ in~\eqref{eq:defgamk} and Fact~\ref{lem:ne}, we have $\gam_k \leq O(\dlt_0)$.
  Then, by induction, we have $\gam_K\leq O\pr{\dlt_0}$.

  By setting $\NN:=\Mt{0,k}$, $\MM := \Mtt{0,k}$, $\UU := \XL{k}$, $\EE:= \Ers{k}$, $m:= 2^k \Fn + \Cn$, $F:= [\puts{k}]\dele C $, $C:= C $, $\rho := \alp$, $\mu := \frac{2 + 2\alp}{\alp}$, $\bet:= 4 + \frac{2}{\alp} + O\pr{\dlt_0} $, $a^2:= 2$, $b:= O(\dlt_0)$ in Lemma~\ref{lem:schurdUL1},
  it is easy to check that all conditions of Lemma~\ref{lem:schurdUL1} are satisfied by similar arguments as above.
  Thus, similar to~\eqref{eq:scMFtil0k}, we have
  \eql{\label{eq:scGs}}{
     & \sc{\Mtt{0,K}, -C} - \sc{\Mt{0,K}, -C}
    \aleq O\pr{\dlt_0}\cdot \sc{\XL{K}, -C}.
  }
  By Fact~\ref{fact:sctran}, Lemma~\ref{lem:Mti}, Lemma~\ref{lem:Mtti},
  \eq{
    &\sc{\Mt{0,K}, -C} = \sc{\sc{\Mt{0,K},-[n]}, F} = \sc{\Lt{K}, F}.  \\
    &\sc{\Mtt{0,K}, -C} = \sc{\sc{\Mtt{0,K}, -[n]}, F} = \sc{\Ltt{K}, F}.
  }
  By Lemma~\ref{enum:Q6}, we have
  $\sc{\XL{K}, - C  } \pleq 2^{K}\U{\sc{\LL, F}} $.

  Substituting the above 3 equations into~\eqref{eq:scGs} and combining with Fact~\ref{fact:aleqpleq1} complete this proof.

\end{proof}
\begin{remark}
    Since $\Stt{0} = \frac{1}{2^K}\pr{\Ltt{K}_{CC} - \Xap} $ (in Algorithm~\ref{alg:SparSchur}), the $\frac{1}{2^K}$ factor on the LHS of~\eqref{eq:scLttgood} doesn't matter.
\end{remark}

Now, we are prepared to prove Theorem~\ref{thm:SparSchur}.
\begin{proof}[Proof of Theorem~\ref{thm:SparSchur}]
  By Lemma~\ref{lem:scLtkequal1}, we have the expansion
  \eq{
    &\SS - \sc{\LL, F}
    = \SS - \Sap + \Stt{0}  + \RR - \frac{1}{2^K}\sc{\Lt{K}, F}  \\
    =& \SS - \Sap + \frac{1}{2^K}\pr{\Ltt{K}_{CC} - \Xap - \sc{\Lt{K}, F}} + \RR  \\
    =& \SS - \Sap + \frac{1}{2^K}\pr{\sc{\Ltt{K}, F} - \sc{\Lt{K}, F}} + \frac{1}{2^K}\pr{\Att{K}_{CF}\DD_{FF}\inv\Att{K}_{FC} - \Xap }  \\
    & + \frac{1}{2^K}\pr{\Att{K}_{CF}\pr{\DD_{FF} - \Att{K}_{FF}}\inv\Att{K}_{FC} - \Att{K}_{CF}\DD_{FF}\inv\Att{K}_{FC}} + \RR    \\
    =& \SS - \Sap + \frac{1}{2^K}\pr{\sc{\Ltt{K}, F} - \sc{\Lt{K}, F}}  + \frac{1}{2^K}\EX  + \Rap.
  }
  By Lemma~\ref{lem:scLttyes} and choosing a small $\eps = O\pr{\frac{\dlt}{K}}$ in Algorithm~\ref{alg:SparSchur}, we can have
  \eq{
    \frac{1}{2^K}\pr{\sc{\Ltt{K}, F} - \sc{\Lt{K}, F}} \aleq \frac{\dlt}{4}\U{\sc{\LL, F}}.
  }
  Then, by Lemma~\ref{lem:scLtkequal1} and Fact~\ref{fact:aleqU},
  \eq{ \U{\sc{\Ltt{K}, F}} \pleq \U{\sc{\Lt{K}, F}} + 2^K\cdot\frac{\dlt}{4}\U{\sc{\LL, F}} = 2^K\pr{1 + \frac{\dlt}{4}}\U{\sc{\LL, F}}.    }

  Combining the above equation with~\eqref{eq:EXscLtt} and Fact~\ref{fact:aleqpleq1} yields that by choosing $\eps \leq \frac{\dlt}{4 + \dlt} $, we have
  \eq{
    \frac{1}{2^K}\EX \aleq \frac{\eps}{2^K}\U{\sc{\Ltt{k}, F}} \pleq \pr{1 + \frac{\dlt}{4}}\eps \U{\sc{\LL, F}} \pleq \frac{\dlt}{4}\U{\sc{\LL, F}}.
  }
  By Fact~\ref{fact:scUpleqUsc}, Fact~\ref{fact:la2scup1} and Cheeger's inequality, we have
  $
    \la_2\pr{\U{\sc{\LL, F}}} \geq \la_2\pr{\sc{\U{\LL}, F}} \geq \la_2\pr{\U{\LL}} \geq \frac{\min_{i\in [n]}\DD_{ii}\pr{\min_{(i,j): \LL_{ij}\neq 0}\abs{\LL_{ij}}}^2}{8\pr{\sum_{i\in [n]}\DD_{ii}}^2} = \Omega\pr{\frac{1}{\poly{n}}}.
  $
  Then, by choosing $K = O\pr{\log\log \frac{n}{\dlt}}$  and using~\eqref{eq:R}, we can let
  $
    \nd{\sc{\LL, F}}{\Rap} \leq \frac{1}{\la_2\pr{\U{\sc{\LL, F}}}} \nt{\Rap} \leq \frac{\dlt}{4}.
  $
  Since $\LL$ is strongly connected, we have $\sc{\LL, F}$ is strongly connected.
  So,  $\ker\pr{\U{\sc{\LL, F}}} = \spanrm{\one} $.
  By combining with the fact $\Rap\one = \Rap\tp\one = \zero $ from Lemma~\ref{lem:LapetcpropSparSchurCpmt1}, we have \eq{
    \Rap \aleq \frac{\dlt}{4} \U{\sc{\LL,  F}}.
  }

  Combining the above equations with Fact~\ref{fact:aleqsum}  yields that
  \eq{
    \Sap - \sc{\LL, F} \aleq \frac{3\dlt}{4}\U{\sc{\LL, F}}.
  }
  Then, $\U{\Sap}\pleq \pr{1 + \frac{3\dlt}{4}}\U{\sc{\LL, F}} \pleq 2\U{\sc{\LL, F}}$.
  Thus,
  \eq{
    \Sap - \SS \aleq \frac{\dlt}{8}\U{\Sap} \pleq \frac{\dlt}{4}\U{\sc{\LL, F}}.
  }
  Then,~\eqref{eq:sttgood} follows.

  The connectivity of $\SS$ can be readily checked as follows.
  If $\SS$ is not \strc, then, there is a vector $\xx\neq \zero$ and $\xx$ not parallel to $\one$ such that $\SS\xx = \zero$.
  Since $\LL$ is a \strc\ Eulerian Laplacian, we have $\sc{\LL, F}$ is \strc, thus, $\xx\tp\U{\sc{\LL, F}}\xx > 0$.
  Then, by~\eqref{eq:sttgood},
  \eq{
    \xx\tp\U{\sc{\LL, F}}\xx = \xx\tp\pr{\U{\sc{\LL, F}} - \U{\SS}}\xx \leq \dlt \xx\tp\U{\sc{\LL, F}}\xx,
  }
  which contradicts the condition $\dlt\in (0, 1)$.
  Thus, $\SS$ is \strc.

  Since we call $\SE$ in each iteration and $\eps = \Otil{\dlt}$, we have $\nnz{\Ltt{k}} = \Otil{\NSE\pr{n, \dlt}}$.
  Since $\eps = \Otil{\dlt}$ and $K = \Otil{1}$,  the total running time of $\SP$ and $\SparP$ is $\Otil{\NSE\pr{n, \dlt}\dlt^{-2}\log n }$ and $\nnz{\Ltt{k, 0}} = \Otil{\NSE\pr{n, \dlt}\dlt^{-2}\log n }$.
  As $K = \Otil{1}$, $\eps = \Otil{\dlt}$, by Theorem~\ref{thm:SparEoracle1} and Lemma~\ref{lem:SE},
  the total running time of $\SE$ is
  $O\pr{\TSE\pr{m,n,\dlt}}+ \Otil{\TSE\pr{\NSE\pr{n, \dlt}\dlt^{-2}\log n  , n, \dlt} } = O\pr{\TSE\pr{m,n,\dlt}}+ \Otil{\TSE\pr{\NSE\pr{n, \dlt}\dlt^{-2}  , n, \dlt}\log n } $
  which gives the overall running time bound for Algorithm~\ref{alg:SparSchur}.
\end{proof}

%% file: solver.tex
\section{A Nearly-linear Time Solver }\label{sec:solver}

In this section,
we complete the Sparsified Schur Complement based algorithm
by invoking the nearly-linear time Schur complement sparsification procedure derived above in Sections~\ref{sec:reformPBEvAM1}
and~\ref{sec:Schurcplstable2}.
We first call this Schur complement sparsification procedure
repeatedly to construct a sparse Schur complement chain,
in Section~\ref{sec:scc}.
Then, in Section~\ref{sec:preconditioner},
we show that this Schur complement chain gives a preconditoner $\PreC$
for the initial Eulerian Laplacian matrix.
The full high accuracy solver then follows from invoking this
preconditioner inside Richardson iteration.
\begin{algorithm}[!htb]
\caption{Block Cholesky solver for directed Laplacians }
\label{alg:DLap}

\KwIn{strongly connected Eulerian Laplacian $\LL\in \MS{n}{n}$; query vectors $\dr{\bt{q}}_{q=1}^Q\sleq \Real^n$ with each $\bt{q} \perp \one$; error parameters $\dr{\eps_q}_{q=1}^Q\sleq (0, 1)$ }

\KwOut{solutions $\dr{\xt{q }}_{q=1}^Q\sleq \Real^n$ }

Call $\SCC\pr{\LL, 0.25, 0.1}$ to compute a $\dr{0.25, 0.05, \dr{\frac{ 0.1}{ i^2}}_{i=1}^{O(\log n)} }$-Schur complement chain $\dr{\dr{\Stt{i}}_{i=1}^d, \dr{F_i}_{i=1}^d}$ (Sections~\ref{sec:reformPBEvAM1},~\ref{sec:Schurcplstable2},~\ref{sec:scc})

Generate the operator $\ZZ\pr{\xx} = \PreC\pr{\dr{\dr{\Stt{i}}_{i=1}^d, \dr{F_i}_{i=1}^d}, \xx, O\pr{\log n } }$  (Section~\ref{sec:preconditioner}) \;

Using the preconditioned Richardson iteration with the preconditioner $\ZZ$ to solve the Laplacian systems: for each query vector $\bt{q}$, 
compute
$\xt{q} \arl \PRI\pr{\LL, \bt{q},  \ZZ\pr{\cdot},  1, O\pr{\log \pr{n/\eps_q}} }  $

\end{algorithm}

\subsection{Schur Complement Chains }\label{sec:scc}
We first define Schur complement chains over directed graphs, which is a variant of the Schur complement chain for undirected graphs in~\cite{kyng2016sparsified}.
\begin{definition}\label{def:SCC1}
    (Schur complement chain)
    Given a \strc\ Eulerian Laplacian $\LL\in\MS{n}{n}$, an $\pr{\alp, \bet, \dr{\dlt_i}_{i=1}^d}$-Schur complement chain of $\LL$ is a sequence of \strc\ Eulerian Laplacians and subsets $\dr{\dr{\Stt{i}}_{i=1}^d, \dr{F_i}_{i=1}^d }$ satisfying
    \begin{enumerate}[(i)]
      \item $\dr{F_i}_{i=1}^d$ is a partition of $[n]$; each $\Stt{i}$ is supported on $\pr{C_{i-1}, C_{i-1}} $, where $C_i \defeq [n]\dele \pr{\cup_{j=1}^{i} F_j} \ (i=0, 1, \cdots, d-1) $; $\abs{C_i} \leq \pr{1 - \bet}^i n$; $\abs{F_d} = \abs{C_{d-1}} = O\pr{1 }. $
      \item For $1\leq i\leq d - 1 $, $\Stt{i}_{F_{i} F_{i}} $ is $\alp$-RCDD.
      \item $\Stt{1} - \LL \aleq \dlt_1 \cdot \U{\LL} $
        and $\Stt{i+1} - \sc{\Stt{i}, F_i} \aleq \dlt_{i+1} \cdot \U{\sc{\Stt{i}, F}},\ 1\leq i\leq d-1.  $

      \item $\U{\Stt{1}} \pgeq \U{\LL}$ and $\U{\Stt{i+1}} \pgeq \U{\sc{\Stt{i}, F_i} },\ 1\leq i\leq d-1$. \label{enum:ErrPSD1}

    \end{enumerate}

\end{definition}
We also denote $F_0 = C_d = \emptyset$, $C_0 = [n]$ for notational simplicity.

\begin{remark}
Compared with the Schur complement chains for undirected graphs
from~\cite{kyng2016sparsified},
the only new condition is Condition~\eqref{enum:ErrPSD1}.
It guarantees the positive semi-definiteness of the symmetrization of the sparsified approximate Eulerian Laplacian $\Lap$ and the error-bounding
matrix $\Bap$ defined in Section~\ref{sec:preconditioner}.
\end{remark}

To construct a Schur complement chain, we first use the following lemma to find an $\alp$-RCDD subset $F_1$, and then apply the Schur complement sparsification method $\SparseSchur$  to compute $\Stt{1}$ which is an approximation for $\sc{\LL, F_1}$.
Then, we repeat this process to get a desirable Schur complement chain.
\begin{lemma}\label{lem:FindRCDD}
    (Theorem~A.1 of~\cite{cohen2018solving})
    Given an Eulerian Laplacian $\LL\in \MS{n}{n}$ with $\nnz{\LL} = m$, the routine $\FindRCDD$ outputs a subset $F\sleq [n]$ such that $\abs{F} \geq \frac{n}{16\pr{1 + \alp}}$ and $\LL_{FF}$ is $\alp$-RCDD in time $O\pr{m\log\frac{1}{p}}$ with probability at least $1 - p$.

\end{lemma}

By Lemma~\ref{lem:FindRCDD}, we can choose for instance $\alp = 0.1$ in practice.
So, we assume $\alp = O(1)$, when analyzing the complexities below.
Our method to construct a Schur complement chain  is illustrated in Algorithm~\ref{alg:SCC}.
It performance is shown in Theorem~\ref{thm:SCC} whose proof is deferred to Appendix~\ref{sec:someprfs}.

\begin{algorithm}[!htb]
\caption{$\SCC\pr{\LL, \alp, \dlt }$ }
\label{alg:SCC}

\KwIn{strongly connected Eulerian Laplacian $\LL\in\MS{n}{n}$; parameters $\alp > 0$, $\dlt\in (0, 1] $  }

\KwOut{$\pr{\alp, \frac{1}{16\pr{1 + \alp}}, \dr{\frac{\dlt}{i^2}}_{i=1}^d}$-Schur complement chain $\dr{\dr{\Stt{i}}_{i=1}^d, \dr{F_i}_{i=1}^d} $ }

Set $\dlt_i' = \frac{\dlt}{3 i^2}$ for $i \geq 1$. \;
Compute $\St{1} \arl \SparE\pr{\LL, \dlt_1'}$  \;
Let $\Stt{1} \arl \St{1} + \frac{\dlt_1'}{1 - \dlt_1' }\U{\St{1}}. $ \;
Set $i \arl 0$, $C_0 = [n]$ \;

\While{$\abs{C_i} > 100$}{
    $i \arl i + 1$  \;
    $F_i \arl \FindRCDD\pr{\Stt{i}, \alp}$ \;
    $C_i \arl C_{i-1}\dele F_i$  \;
    $\St{i+1} \arl \SparseSchur\pr{\Stt{i}, F_i, \dlt_{i+1}} $ \;
    $\Stt{i+1} \arl \St{i} + \frac{\dlt_{i+1}' }{1 - \dlt_{i+1}' }\U{\St{i}}  $

}

Return $\dr{\dr{\Stt{i}}_{i=1}^d, \dr{F_i}_{i=1}^d} $

\end{algorithm}

\begin{theorem}\label{thm:SCC}
    Given a \strc\ Eulerian Laplacian $\LL\in \MS{n}{n}$ and parameters $\alp = O(1)$, $\dlt\in (0, 1] $,
    the routine $\SCC$ runs in time \eq{ O\pr{\TSE\pr{m, n, \dlt}} + \Otil{\TSE\pr{\NSE\pr{n, \dlt}\dlt^{-2}, n, \dlt }\log n  } } with high probability  to return an $\pr{\alp, \frac{1}{16\pr{1 + \alp}}, \dr{\frac{\dlt}{i^2}}_{i=1}^d}$-Schur complement chain,
    where $d = O\pr{\log n}$.
    In addition, $\sum_{i=1}^{d}\nnz{\Stt{i}} = O\pr{\NSE\pr{n, \dlt}}$.

\end{theorem}

\subsection{Construction of the Preconditioner and the Solver }\label{sec:preconditioner}
After constructing a desirable Schur complement chain, we  use the Schur complement chain to construct a preconditioner and solve $\LL\xx = \bb$ via the preconditioned Richardson iteration.

Consider a linear system $\AA\xx = \bb$, where $\bb$ is in the image space of $\AA$.
Given a preconditioner $\ZZ$, the classical preconditioned Richardson iteration updates as follows:
\eq{
    \xt{k+1} \arl \xt{k} + \eta \ZZ\pr{\bb - \AA \xt{k}}.
}
We initialize $\xt{0} = \zero$ for simplicity.
This procedure is denoted by $\xt{N} = \PRI\pr{\AA, \bb, \ZZ, \eta, N}$.

We will use the following fundamental lemma to guarantee the performance of the preconditioned Richardson iteration in our methods.
\begin{lemma}\label{lem:PRIconverge1}
    (Lemma~4.2 of~\cite{cohen2017almost})
    Let $\AA, \ZZ, \UU\in \MS{n}{n}$, where $\UU$ is PSD  and  $\ker\pr{\UU} \sleq \ker\pr{\ZZ} = \ker\pr{\ZZ\tp} = \ker\pr{\AA} = \ker\pr{\AA\tp}$.
    Let $\bb\in \Real^n$ be a vector inside the image space of $\AA$.
    Denote the projection onto the image space of $\AA$ by $\PA$.
    Denote $\xt{N} = \PRI\pr{\AA, \bb, \ZZ, \eta, N}$.
    Then, $\xt{N}$ satisfies
    \eq{
        \nA{\UU}{\xt{N} - \AA\dg\bb} \leq \narr{\UU}{\PP_{\AA} - \eta\ZZ\AA}^N \nA{\UU}{\AA\dg\bb}.
    }
    In addition, the preconditioned Richardson iteration is a linear operator with
    \eql{\label{eq:PrecondiRichardsonitera}}{
        \xt{N} = \eta \sum_{k=0}^{N-1}\pr{\PA - \eta\ZZ \AA  }^{k}\ZZ\bb.
    }

\end{lemma}

Our construction for the preconditioner is illustrated in Algorithm~\ref{alg:precondition}.

To analyze Algorithm~\ref{alg:precondition}, we define the following matrices.

$\Con = \II - \frac{\one\one\tp}{n}$ is the projection matrix onto the image space of $\LL$.

$\Dap$ is an $n$-by-$n$ diagonal matrix with $\Dap_{F_i F_i} = \Diag{\Stt{i}}_{F_i F_i}$ for $i\in [d]$.

$\Mpt{i, N}$ is the linear operator corresponding to the preconditioned Richardson iterations
\eq{
    \Mpt{i, N} = \frac{1}{2}\sum_{k=0}^{N-1}\pr{\II - \frac{1}{2}\Dap_{F_i F_i}\inv\Stt{i}_{F_i F_i}}^k\Dap_{F_i F_i}\inv = \frac{1}{2}\sum_{k=0}^{N-1}\Dap_{F_i F_i}\inv\pr{\II - \frac{1}{2}\Stt{i}_{F_i F_i}\Dap_{F_i F_i}\inv}^k,\ i \in [d - 1].
}

$\Ltilt{i, N}$ and $\Utilt{i, N}$ are block lower  triangular and block upper  triangular matrices of the block Cholesky factorization with
\eq{
    \Ltilt{i, N} =
    \mx{
        \II_{\sum_{j=1}^{i-1}\abs{F_i}} & & \\
        & \II &  \\
        & \Stt{i}_{C_i F_i}\Mpt{i, N} & \II
    },\
    \Utilt{i, N} =
    \mx{
        \II_{\sum_{j=1}^{i-1}\abs{F_i}} & &   \\
        & \II & \Mpt{i, N}\Stt{i}_{F_i C_i}  \\
        &  & \II
    },
}
where $\II_k$ denotes the $k$-by-$k$ identity matrix.
$\DS{N}$ is the block diagonal matrix corresponding to the block Cholesky factorization
\eq{
    \DS{N} =
    \mx{
        \iv{\Mpt{1, N}} & & & \\
        & \ddots & &  \\
        & & \iv{\Mpt{d-1, N}} & \\
        & & & \Stt{d}_{F_d F_d}
    },
}
where the invertibility of $\Mpt{i, N}$ is given by Lemma~\ref{lem:Mpt}.

Note that
\eq{
    \pr{\Ltilt{i}}\inv =
    \mx{
        \II_{\sum_{j=1}^{i-1}\abs{F_i}} & & \\
        & \II &  \\
        & - \Stt{i}_{C_i F_i}\Mpt{i, N} & \II
    },\
    \pr{\Utilt{i}}\inv =
    \mx{
        \II_{\sum_{j=1}^{i-1}\abs{F_i}} & &   \\
        & \II & - \Mpt{i, N}\Stt{i}_{F_i C_i}  \\
        &  & \II
    }.
}

Then, the routine $\PRI$ is a linear operator which is equivalent to multiplying vector $\xx$ with the matrix $\Con\Zhat$, where $\Zhat\in\MS{n}{n}$ is defined as follows:
\eq{
    \Zhat = \pr{\Utilt{1, N}}\inv  \bigcdot \cdots \bigcdot \iv{\Utilt{d-1, N}}  \pr{\DS{N}}\dg \iv{\Ltilt{d-1, N}} \bigcdot  \cdots \bigcdot \iv{\Ltilt{1, N}}.
}

We also define the following matrices which are counterparts of $\dr{\Ltilt{i, N}}$ when $N = +\infty$ in Algorithm~\ref{alg:precondition}:
\eq{
    \Ltilt{i, \infty} =
    \mx{
        \II_{\sum_{j=1}^{i-1}\abs{F_i}} & & \\
        & \II &  \\
        &  \Stt{i}_{C_i F_i}\iv{\Stt{i}_{F_i F_i}} & \II
    }
}
The matrices $\dr{\Utilt{i, \infty}}$, $\dr{\DS{\infty}}$ are defined similarly by replacing $\Mpt{i, N}$  with  $\iv{\Stt{i}_{F_i F_i}}$ in $\dr{\Utilt{i, N}}$, $\dr{\DS{N}}$.

Define $\Lap$ as an approximation for $\LL$ with the errors induced by the Schur complement sparsification procedure
\eql{\label{eq:defLap}}{
    \Lap = \Stt{1}  + \sum_{i=1}^{d-1} \putmat{\Stt{i+1} - \sc{\Stt{i}, F_{i}}, C_{i}, C_{i}, n},
}
where the notation $\putmat{\cdot}$ is defined in Section~\ref{sec:ULtk} which means putting a matrix on the designated position in an all-zeros matrix with designated size.

Then, by direct calculations,
\eql{\label{eq:LapLform1}}{
    \Lap = \Ltilt{1, \infty} \cdots \Ltilt{d-1, \infty} \DS{\infty} \Utilt{d-1, \infty} \cdots \Utilt{1, \infty}.
}

The following matrices $\BB$ and $\Bap$ are playing the role of $\UU$ in Lemma~\ref{lem:PRIconverge1}:
\eq{
    &\BB = \dlt_1\U{\LL} + \sum_{i=2}^{d} \dlt_i\putmat{\U{\Stt{i}}, C_{i-1}, C_{i-1}, n}, \\
    &\Bap = \dlt_1\U{\Lap} + \sum_{i=1}^{d-1} \dlt_{i+1} \putmat{\U{\sc{\Lap, \cup_{j=1}^{i}F_j}}, C_i, C_i, n }.
}

The proofs of the following lemmas are deferred to Appendix~\ref{sec:someprfs}.
\begin{lemma}\label{lem:BBpleqo1Bap}
    If the input $\dr{\alp, \bet, \dr{\dlt_i}_{i=1}^d}$-Schur complement chain satisfies $\sum_{i=1}^{d}\dlt_i \leq 1  $, then
    $
        \BB \pleq \Bap \pleq 2\BB.
    $

\end{lemma}

\begin{lemma}\label{lem:Mpt}
    For $N \geq 1 $, $\Mpt{i, N}$ is nonsingular and
    $\ni{\iv{\Utilt{i, N}} - \iv{\Utilt{i, \infty}}} \leq \frac{\pr{1 + \alp}}{\alp}\pr{\frac{2 + \alp}{2\pr{1 + \alp}}}^N,   $
    $\no{\iv{\Ltilt{i, N}} - \iv{\Ltilt{i, \infty}}} \leq \frac{\pr{1 + \alp}}{\alp}\pr{\frac{2 + \alp}{2\pr{1 + \alp}}}^N,  $
    $
        \ni{\iv{\Stt{i}_{F_i F_i}} - \Mpt{i, N}  } \leq \frac{\pr{1 + \alp}}{\alp}\pr{\frac{2 + \alp}{2\pr{1 + \alp}}}^N \ni{\Dap_{F_i F_i}\inv}.
    $

\end{lemma}

From~\eqref{eq:PrecondiRichardsonitera},
to analyze the quality of the predconditioner $\Con\Zhat$,
we need to provide bounds for $\Con - \Con\Zhat\LL$.

\begin{lemma}\label{lem:precondiquality1}
    Given $\dr{\alp, \bet, \dr{\dlt_i}_{i=1}^d}$-Schur complement chain with $d = O\pr{\log n}$ and $\sum_{i=1}^{d} \dlt_i \leq  \frac{1}{4}   $,
    by setting $N = O\pr{\log n }$ in Algorithm~\ref{alg:precondition}, we can have
    $
        \narr{\Bap}{\Con - \Con\Zhat\LL} \leq \frac{1}{2}.
    $

\end{lemma}
\begin{proof}
    From the fact $\dr{\Stt{i}}$ are all Eulerian Laplacians and Fact~\ref{fact:ESchurE}, we have $\Lap\one = \Lap\tp\one = \zero$.
    By~\eqref{eq:LapLform1} and the strong connectivity of $\Stt{d}$, we have ${\rm rank}\pr{\Lap} = n-1$.
    Then, $\ker\pr{\Lap} = \ker\pr{\Lap\tp } = {\rm span}\pr{\one}$.
    Thus,  $\Lap\Lap\dg = \Lap\dg\Lap = \Con$.

    Now, we expand $\Con - \Con\Zhat\LL$ as follows
    \eq{
        \Con - \Con\Zhat\LL = \Lap\dg\pr{\Lap - \LL} + \pr{\Lap\dg - \Con\Zhat\Con}\LL.
    }
     By the fact that $\pr{\Con - \Con\Zhat\LL}\one = \zero $ and $\Bap \pgeq \U{\LL} $, we have $\ker\pr{\Con - \Con\Zhat\LL} \sgeq \ker\pr{\Bap} $.
    Then, by combining with the definition  of $\narr{\Bap}{\cdot}$, we have
    \eql{\label{eq:percondiexpand1}}{
        \narr{\Bap}{\Con - \Con\Zhat\LL}^2 &= \narr{\Bap}{\Lap\dg\pr{\Lap - \LL} + \pr{\Lap\dg - \Con\Zhat\Con}\LL}^2  \\
        &= \nt{\Bap^{1/2}\pr{\Lap\dg\pr{\Lap - \LL} + \pr{\Lap\dg - \Con\Zhat\Con}\LL}\Bap^{\dagger/2}}^2  \\
        &\leq 2\nt{\Bap^{1/2}\Lap\dg\pr{\Lap - \LL}\Bap^{\dagger/2}}^2 + 2\nt{\Bap^{1/2}\pr{\Lap\dg - \Con\Zhat\Con}\LL\Bap^{\dagger/2}}^2.
    }

    By the definitions of $\Lap, \BB$, we have
    \eq{
        &2\xx\tp \pr{\Lap - \LL} \yy  \\
        \leq& \xx\tp \pr{\dlt_1 \U{\LL} + \sum_{i=2}^{d} \dlt_i \putmat{\U{\Stt{i}}, C_{i-1}, C_{i-1}, n}} \xx  \\
         & + \yy\tp \pr{\dlt_1 \U{\LL} + \sum_{i=2}^{d} \dlt_i \putmat{\U{\Stt{i}}, C_{i-1}, C_{i-1}, n}} \yy  \\
        =& \xx\tp\BB\xx + \yy\tp\BB\yy.
    }
    By combining with Lemma~\ref{lem:BBpleqo1Bap}, we have
    \eql{\label{eq:Lap-LBap1}}{
        \ndd{\Bap}{\pr{\Lap - \LL}} \leq \ndd{\BB}{\pr{\Lap - \LL}} \leq 1.
    }

    Next, we bound $\pr{\Lap\dg - \Con\Zhat\Con}\LL $.
    From the definition of $\dr{\alp, \bet, \dr{\dlt_i}_{i=1}^d}$-Schur complement chain, we have
    $\U{\Stt{i+1}} \pgeq  \U{\sc{\Stt{i}, F_i}}. $
    Combining with Fact~\ref{fact:scUpleqUsc}, we have $\U{\Stt{i+1}} \pgeq  \sc{\U{\Stt{i}}, F_i}$.
    Then, by Fact~\ref{fact:la2scup1}, $\la_2\pr{\U{\Stt{i+1}}} \geq  \la_2\pr{\sc{\U{\Stt{i}}, F_i}} \geq  \la_2\pr{\U{\Stt{i}}}. $
    By induction, $\la_2\pr{\U{\Stt{i}}} \geq \la_2\pr{\U{\LL}}.  $
    By Cheeger's inequality, $\la_2\pr{\U{\LL}} = \Omega\pr{\frac{1}{\poly{n}}}$.
    Then, for any $i\in [d]$, $\la_2\pr{\U{\Stt{i}}} = \Omega\pr{\frac{1}{\poly{n}}}$.
    By Fact~\ref{fact:la2scup1}, $\ni{\Dap_{F_i F_i}\inv } \leq \frac{2}{\la_2\pr{\U{\Stt{i}}} } = O\pr{\poly{n}}. $
    Also, $\la_2\pr{\Bap} \geq  \la_2\pr{\BB} = \Omega\pr{\frac{1}{\poly{n}}}$.
    It follows by induction easily that $\nt{\Bap} = O\pr{\poly{n}} $.

    By Fact~\ref{fact:MpinvABC} and~\eqref{eq:LapLform1}, we have
    \eq{
        \Lap\dg = \Con\iv{\Utilt{1, \infty}} \cdots \iv{\Utilt{d-1, \infty}} \pr{\DS{\infty}}\dg \iv{\Ltilt{d-1, \infty}} \cdots \iv{\Ltilt{1, \infty}}\Con.
    }
    Since $\ni{\iv{\Stt{i}_{F_i F_i}}\Stt{i}_{F_i C_i}} = \frac{1}{2}\ni{\sum_{k=0}^{+\infty} \pr{\II - \frac{1}{2}\Dap_{F_i F_i}\inv\Stt{i}_{F_i F_i}}\Dap_{F_i F_i}\inv \Stt{i}_{F_i C_i} } \leq \frac{1 + \alp}{\alp},    $
    we have
    $\no{\iv{\Ltilt{i, \infty}}} \leq \frac{1 + 2\alp}{\alp} $.
    Analogously, $\ni{\iv{\Utilt{i, \infty}}} \leq \frac{1 + 2\alp}{\alp} $.
    Then, by Fact~\ref{fact:sequenceproductelemta}, Lemma~\ref{lem:Mpt}, the fact that $\ni{\Dap_{F_i F_i}\inv} = O\pr{\poly{n}} $ and $\pr{\frac{1 + 2\alp}{\alp}}^{O\pr{\log n }} = O\pr{\poly{n}} $, $\la_2\pr{\Bap} = \Omega\pr{\frac{1}{\poly{n}}}$, we have
    \eq{
        \nt{\Bap^{1/2}\pr{\Lap\dg - \Con\Zhat\Con}\LL\Bap^{\dagger/2}} \leq \pr{\frac{2 + \alp}{2\pr{1 + \alp}}}^N \cdot O\pr{\poly{n}}.
    }
    Then, by setting $N = O\pr{\log n}$ in Algorithm~\ref{alg:precondition}, we can let
    \eql{\label{eq:NexpZLinv1}}{
        \nt{\Bap^{1/2}\pr{\Lap\dg - \Con\Zhat\Con}\LL\Bap^{\dagger/2}}^2 \leq \frac{1}{16}.
    }

    By Fact~\ref{fact:FLFgood}, we have
    \eql{\label{eq:LBL}}{
        \tpp{\Lap\dg}\Bap \Lap\dg \pleq \pr{\sum_{i=1}^{d}\dlt_i}^2 \Bap\dg \pleq \frac{1}{16} \Bap\dg.
    }
    Continuing with~\eqref{eq:percondiexpand1},
    \eq{
        &\narr{\Bap}{\Con - \Con\Zhat\LL}^2 \\ 
        \leq& 2\nt{\Bap^{1/2}\Lap\dg\pr{\Lap - \LL}\Bap^{\dagger/2}}^2 + 2\nt{\Bap^{1/2}\pr{\Lap\dg - \Con\Zhat\Con}\LL\Bap^{\dagger/2}}^2  \\
        =& 2 \nt{\Bap^{\dagger/2} \pr{\Lap - \LL}\tp \tpp{\Lap\dg} \Bap \Lap\dg\pr{\Lap - \LL} \Bap^{\dagger/2}} + 2\nt{\Bap^{1/2}\pr{\Lap\dg - \Con\Zhat\Con}\LL\Bap^{\dagger/2}}^2  \\
        \leq& \frac{1}{8 } \nt{\Bap^{\dagger/2} \pr{\Lap - \LL}\tp \Bap\dg \pr{\Lap - \LL} \Bap^{\dagger/2}} + 2\nt{\Bap^{1/2}\pr{\Lap\dg - \Con\Zhat\Con}\LL\Bap^{\dagger/2}}^2  \\
        \leq& \frac{1}{8} + \frac{1 }{8} = \frac{1}{4 },
    }
    where the second inequality is by~\eqref{eq:LBL}; the last inequality is from~\eqref{eq:Lap-LBap1} and~\eqref{eq:NexpZLinv1}.

\end{proof}


\begin{proof}[Proof of Theorem~\ref{thm:TSENSEsolver1}]
  By Fact~\ref{lem:scrobust} and
  the definition of $\dr{\alp, \bet, \dr{\dlt_i}_{i=1}^d}$-Schur complement chain, we have
  \eq{
    \U{\Stt{i+1}} \pleq \pr{1 + \dlt_{i+1}} \pr{3 + \frac{2}{\alp}} \U{\sc{\Stt{i}, F_i}}.
  }
  It follows by Fact~\ref{fact:scprvpleq} and induction that $\U{\Stt{i}} \pleq \pr{3 + \frac{2}{\alp}}^{i-1}\prod_{j=2}^{i}\pr{1 + \dlt_{j}} \sc{\U{\Stt{1}}, \cup_{j=1}^{i-1}F_j} \pleq \pr{3 + \frac{2}{\alp}}^{i-1} \prod_{j=1}^{i}\pr{1 + \dlt_{j}} \sc{\U{\LL}, \cup_{j=1}^{i-1}F_j} $.

  Since $\sum_{i=1}^{d}\dlt_i = O(1)$, $d = O(\log n)$, we have $\pr{3 + \frac{2}{\alp}}^{d}\prod_{j=1}^{d}\pr{1 + \dlt_{j}} = O\pr{\poly{n}},  $
  i.e., $\U{\Stt{i}} \pleq O\pr{\poly{n}}\cdot \sc{\U{\LL}, \cup_{j=1}^{i-1} F_j}$ for any $i\in [d]$.
  Thus, by Fact~\ref{fact:scz}, we have
  $\putmat{\U{\Stt{i}}, C_{i-1}, C_{i-1}, n} \pleq O\pr{\poly{n}}\cdot\U{\LL}.  $
  By combining with Lemma~\ref{lem:BBpleqo1Bap}, we have \eq{\Bap \pleq 2\BB \pleq 2\pr{\U{\LL} + \sum_{i=2}^{d}\putmat{\U{\Stt{i}}, C_{i-1}, C_{i-1}, n}} = O\pr{\poly{n}}\cdot \U{\LL}.        }
  By Lemma~\ref{lem:precondiquality1}, after running $N'$ iterations of the preconditioned Richardson iteration, $ \nA{\U{\LL}}{\xt{N'} - \LL\dg\bb} \leq  \nA{\Bap}{\xt{N'} - \LL\dg\bb} \leq \narr{\Bap}{\II - \Con\Zhat\LL}^N\nA{\Bap}{\LL\dg\bb} \leq \pr{\frac{1}{2}}^{N'}\nA{\Bap}{\LL\dg\bb} \leq \pr{\frac{1}{2}}^{N'}\cdot O\pr{\poly{n}} \cdot \nA{\U{\LL}}{\LL\dg\bb}.  $
  By setting $N' = O\pr{\log\pr{n/\eps}}$, we can let
  $
    \nA{\U{\LL}}{\xt{N'} - \LL\dg\bb} \leq \eps \nA{\U{\LL}}{\LL\dg\bb}.
  $

  The processing time follows by Theorem~\ref{thm:SCC} directly.
  By Theorem~\ref{thm:SCC}, $\sum_{i=1}^{d}\nnz{\Stt{i}} = \sum_{i=1}^d O\pr{\NSE\pr{\pr{1 - \bet}^{i-1}n, \frac{\dlt}{i^2}}} $.
  As $\dlt = O(1)$, $\bet = \frac{1}{16(1+\alp)} = O(1) $, we have $\sum_{i=1}^{d}\nnz{\Stt{i}} = O\pr{\NSE\pr{n, 1}}$.
  Thus, as we set $N = O\pr{\log n}$ in $\PreC$, running $\PreC $ for one time takes $O\pr{\NSE\pr{n, 1}\log n }$ time.
  Then, after obtaining a desirable Schur complement chain, an $\eps$-accurate vector $\xx$ can be computed in $O\pr{\NSE\pr{n, 1}\log n \log\pr{n/\eps} } $ time.

\end{proof}

Using the smaller Eulerian Laplacian sparsifiers based on short cycle decompositions to sparsify the approximate Schur complements returned by Algorithm~\ref{alg:SparSchur},
we get the following solver which has quadratic processing time, but faster solve time.
Its proof is deferred to Appendix~\ref{sec:sparsify}.
\begin{corollary}\label{coro:shortcyclep1}
    Given a strongly connected Eulerian Laplacian $\LL\in \MS{n}{n}$, we can process it time $O(n^2\log^{O(1)} n)$.
    Then, for each query vector $\bb\in \Real^n$ with $\bb \perp \one$, we can compute a vector $\xx\in \Real^n$ with $\nA{\U{\LL}}{\xx - \LL\dg\bb} \leq \eps \nA{\U{\LL}}{\LL\dg\bb}$ in time $O(n\log^5 n\log(n/\eps))$.

\end{corollary}

\begin{remark}
Combining Theorem~\ref{thm:TSENSEsolver1} or Corollary~\ref{coro:shortcyclep1} with
Appendix~D of~\cite{cohen2017almost}
yields full solvers for strongly connected directed Laplacians.
\end{remark}

%% file: Appendix.tex
\begin{appendices}

\section{Matrix Facts}

Some elementary lemmas used frequently in this paper are listed in this section.
Unless otherwise specified, we assume $F = \dr{1, 2, \cdots, \abs{F}}$ and $C = [n]\backslash F$ by default in this section.

\newtheorem*{fact:ne}{Fact~\ref{lem:ne}}
\begin{fact:ne}
    \lemne
\end{fact:ne}

\begin{proof}
  Compared with Lemma~B.2 of~\cite{cohen2017almost},
  we need to show that under the condition
  \eq{2\xx\tp\AA\yy
  \leq \xx\tp\UU\xx + \yy\tp\UU\yy,
  \qquad
  \forall \xx, \yy \in \Real^n,}
  we have
  \eq{\ker\pr{\UU} \sleq \ker\pr{\AA} \cap \ker\pr{\AA\tp}.}

By the assumption,
we have for any $\vv\in \ker\pr{\UU} \dele \dr{\zero }$,
  \eq{
    2\xx\tp\AA\vv \leq \xx\tp\UU\xx,\quad 2\vv\tp\AA\yy \leq \yy\tp\UU\yy,\qquad \forall \xx, \yy\in \Real^n.
  }
  By choosing $\xx = c\AA\vv$, $\yy = c\AA\tp\vv$, we have
  \eq{
    2c \nt{\AA\vv}^2 \leq c^2 \nt{\UU}\nt{\AA}^2\nt{\vv}^2,\quad 2c \nt{\AA\tp\vv}^2 \leq c^2\nt{\UU}\nt{\AA}^2\nt{\vv}^2,\qquad \forall c > 0.
  }
  By letting $c \arr 0^+$, we have $\nt{\AA\vv} = 0$ and $\nt{\AA\tp\vv} = 0$, i.e., $\vv\in \ker\pr{\AA}\cap \ker\pr{\AA\tp}$.
  Thus, $\ker\pr{\UU} \sleq \ker\pr{\AA} \cap \ker\pr{\AA\tp}. $

\end{proof}

\begin{fact}\label{fact:aleqsum}
    If $\AA \aleq a\CC$, $\BB \aleq b\CC$, then $\AA + \BB \aleq \pr{a + b}\CC$.
\end{fact}
\begin{fact}\label{fact:aleqpleq1}
    If $\AA \aleq \BB$, $\BB \pleq \CC$, then $\AA \aleq \CC$.
\end{fact}
\begin{fact}\label{fact:aleqU}
    If $\AA\aleq \BB$, then $\U{\AA} \pleq \BB$.
\end{fact}

\begin{fact}\label{fact:ESchurE}
    Schur complements of Eulerian Laplacians are Eulerian Laplacians;
    Schur complements of strongly connected Eulerian Laplacians are strongly connected Eulerian Laplacians.
\end{fact}

\begin{fact}
    For any Eulerian Laplacian $\LL$, $\U{\LL}$ is PSD.
\end{fact}

\begin{fact}\label{fact:alpRCDDPSDpPD1}
    If matrix $\AA\in \MS{n}{n}$ is $\alp$-RCDD $(\alp \geq 0)$, then $\U{\AA}$ is PSD.
    If $\alp > 0$, then $\U{\AA}$ is PD.
\end{fact}

\begin{fact}\label{fact:strcrank1}
    Any \strc\ Laplacian $\LL\in\MS{n}{n}$ has rank $n-1$.
\end{fact}

\begin{fact}\label{fact:scprvpleq}
    (Lemma~B.1 of~\cite{miller2013approximate})
    Suppose that $\AA, \BB\in\MS{n}{n} $ are PSD, $F, C$ is a partition of $[n]$, where $\AA_{FF}, \BB_{FF}$ are nonsingular and $\AA\pleq \BB$.
    Then,
     $\sc{\AA, F} \pleq \sc{\BB, F}. $

\end{fact}

\begin{fact}\label{fact:sctran}
    (Lemma~C.2 of~\cite{cohen2018solving})
    Let $\MM$ be an $n$-by-$n$ matrix and $F, C$ a partition of $[n]$ such that $\MM_{FF}$ is nonsingular.
    Let $F_1, F_2$ be a partition of $F$ such that $\MM_{F_1 F_1}$ is nonsingular.
    Then,
    \eq{
        \sc{\sc{\MM, F_1}, F_2} = \sc{\MM, F}.
    }

\end{fact}

\begin{fact}\label{fact:Schurxusmall}
    For any symmetric matrix $\UU\in \MatSize{n}{n}$ and $F, C$ a partition of $[n]$, where $\UU_{FF}$ is positive definite,
    for any $\xx\in \Real^{\abs{C}}, \xtil\in \Real^n, \text{with } \xtil_C = \xx$,
    we have
    \eq{
        \nA{\UU}{\xtil}^2 =  \nA{\sc{\UU, F}}{\xx}^2 + \nA{\UU_{FF}}{\xtil_F + \UU_{FF}\inv\UU_{FC}\xx}^2 \geq \nA{\sc{\UU, F}}{\xx}^2.
    }

\end{fact}
\begin{proof}
  Define
  \eq{
    \xu =
    \vc{
        - \UU_{FF}\inv\UU_{FC}\xx  \\
        \xx
    }.
  }
  Then,
  \eq{
    \UU\xu =
    \vc{
        \zerov{F}  \\
        \sc{\UU, F}\xx
    }.
  }
  Since $\pr{\xtil - \xu}_C = \zero$ and $\pr{\UU\xu}_F = \zero$,
  \eq{
    \pr{\xtil - \xu}\tp\UU\xu = 0.
  }
  Thus,
  \eq{
    &\xtil\tp\UU\xtil = \xu\tp\UU\xu + \pr{\xtil - \xu}\tp\UU\pr{\xtil - \xu} + 2\pr{\xtil - \xu}\tp\UU\xu  \\
    =& \xu\tp\UU\xu + \pr{\xtil - \xu}\tp\UU\pr{\xtil - \xu}
    = \nA{\sc{\UU, F}}{\xx}^2 + \nA{\UU_{FF}}{\xtil_F + \UU_{FF}\inv\UU_{FC}\xx}^2.
  }

\end{proof}

\begin{fact}\label{fact:kerAsmltkUAaPSD}
    For matrix $\AA\in \MS{n}{n}$, if $\U{\AA}$ is PSD, then we have  $\ker\pr{\AA} \sleq \ker\pr{\U{\AA}}$.
\end{fact}
\begin{proof}
  Since $\U{\AA}$ is PSD, we have
  \eq{  \AA\xx = \zero \ \Rightarrow \ \xx\tp\U{\AA}\xx = 0 \ \Rightarrow \ \U{\AA}^{\dagger/2}\xx = \zero \ \Rightarrow \ \U{\AA}\xx = \zero,             }
  i.e., $\ker\pr{\AA} \sleq \ker\pr{\U{\AA}}.  $

\end{proof}

\begin{fact}\label{fact:LUL}
    (Lemma~B.9 of~\cite{cohen2017almost})
    For any matrix $\LL\in \MatSize{n}{n}$ with $\U{\LL} \pgeq \zero $ and $\ker\pr{\LL} = \ker\pr{\LL\tp} = \ker\pr{\U{\LL}}$, we have
    \eq{
        \U{\LL} \pleq \LL\U{\LL}^{\dagger}\LL\tp.
    }

\end{fact}

\begin{fact}\label{fact:LDL}
    (Lemma~4.5 of~\cite{cohen2018solving})
    For Eulerian Laplacian $\LL $ and $\DD = \Diag{\LL}$, we have
    \eq{
        \LL\tp\DD\inv\LL \pleq 2\U{\LL}.
    }

\end{fact}

\begin{fact}\label{fact:scz}
    Consider symmetric matrices $\AA\in \MS{n}{n}$, $F, C$ a partition of $[n]$ and $\BB\in \MS{\abs{F}}{\abs{F}}$, where $\AA_{FF}$ is PD.
    Then, $\BB \pleq \sc{\AA, F}$ is equivalent to
    \eq{
        \mx{
            \zerom{F}{F} & \zerom{F}{C}  \\
            \zerom{C}{F} & \BB
        } \pleq
        \AA,
    }
    or equivalently,
    \eq{
        \xx\tp
        \mx{
            \zerom{F}{F} & \zerom{F}{C}  \\
            \zerom{C}{F} & \BB
        }
        \xx
        \leq \xx\tp \AA \xx,\ \forall \xx\in \Real^n.
    }

\end{fact}
\begin{proof}
  If $\BB \pleq \sc{\AA, F} $, then
  by Fact~\ref{fact:Schurxusmall}, for any $\xx\in \Real^n$,
  \eq{
    \xx\tp
        \mx{
            \zerom{F}{F} & \zerom{F}{C}  \\
            \zerom{C}{F} & \BB
        }
        \xx
        = \xx_C\tp \BB \xx_C \leq \xx_C\tp \sc{\AA, F} \xx_C \leq \xx\tp \AA \xx.
  }
  This gives $\mx{
            \zerom{F}{F} & \zerom{F}{C}  \\
            \zerom{C}{F} & \BB
        } \pleq
        \AA$.

  If $\yy\tp
        \mx{
            \zerom{F}{F} & \zerom{F}{C}  \\
            \zerom{C}{F} & \BB
        }
        \yy
        \leq \yy\tp \AA \yy,\ \forall \yy\in \Real^n$,
  then
  for any $\xx\in \Real^{\abs{C}}$, define
  $
    \xtil =
    \vc{
        -\AA_{FF}\inv\AA_{FC}\xx  \\
        \xx
    }.
  $
  Then, $\xx\tp\BB\xx = \xtil\tp\mx{
            \zerom{F}{F} & \zerom{F}{C}  \\
            \zerom{C}{F} & \BB
        }\xtil \leq \xtil\tp\AA\xtil = \xx\tp\sc{\AA, F}\xx$.

\end{proof}

The following fact is  from Lemma~4.4 of~\cite{cohen2018solving} and Fact~\ref{fact:scz}.
\begin{fact}\label{lem:scrobust} 
    For \strc\ Eulerian Laplacian $\LL\in \MatSize{n}{n} $ and $F, C$ a partition  of $[n]$ such that $\LL_{FF}$ is $\alp$-RCDD, we have
    \eq{
        \U{\sc{\LL,  F}} \pleq \pr{3 + \frac{2}{\alp}}\sc{\U{\LL}, F}
    }

\end{fact}

\begin{fact}\label{fact:ninobnt}
    (Lemma~B.4 of~\cite{cohen2017almost})
    For positive diagonal matrices $\DD_1\in\MS{m}{m}$, $\DD_2\in \MS{n}{n}$ and arbitrary $\MM\in \MS{m}{n}$, we have
    \eq{
        \nt{\DD_1\MM\DD_2} \leq \max\dr{\ni{\DD_1^2\MM}, \no{\MM\DD_2^2}}.
    }

\end{fact}

\begin{fact}\label{fact:DLD}
    For any Eulerian Laplacian $\LL\in\MS{n}{n}$, let $\DD = \Diag{\LL}$, then,
    $
        \nt{\DD^{-1/2}\LL\DD^{-1/2}} \leq 2.
    $

\end{fact}
\begin{proof}
  Since $\LL$ is an Eulerian Laplacian, $\ni{\DD\inv\LL} \leq 2$, $\no{\LL\DD\inv} \leq 2$.
  Then, the result follows by Fact~\ref{fact:ninobnt}.

\end{proof}

\begin{fact}\label{fact:scUpleqUsc}
    For any matrix $\LL\in\MS{n}{n}$, denote $\UU = \U{\LL}$.  If $F, C$ is a partition of $[n]$ such that $\UU_{FF}$ is PD, then
    \eq{
        \sc{\UU, F} \pleq \U{\sc{\LL, F}}.
    }

\end{fact}
\begin{proof}
  By Fact~\ref{fact:kerAsmltkUAaPSD}, $\LL_{FF}$ is nonsingular.
  For any $\xx\in \Real^{\abs{C}}$, define
  $
    \xtil =
    \vc{
        - \LL_{FF}\inv \LL_{FC} \xx  \\
        \xx
    }.
  $

  Then, by Fact~\ref{fact:Schurxusmall},
  \eq{
    \xx\tp\sc{\UU, F}\xx \leq \xtil\tp\UU\xtil = \xtil\tp\LL\xtil = \xx\tp \sc{\LL, F} \xx = \xx\tp \U{\sc{\LL, F}} \xx,
  }
  i.e.,
  $
    \sc{\UU, F} \pleq \U{\sc{\LL, F}}.
  $

\end{proof}

\begin{fact}\label{fact:MpinvABC}
    (Lemma~C.3 of~\cite{cohen2018solving})
    Consider matrices $\AA \in \MS{m}{m}$, $\BB\in \MS{m}{n}$ and $\CC\in \MS{n}{n}$. Let $\MM = \AA\BB\CC$.
    Let $\PP_{\MM}$, $\PP_{\MM\tp}$ denote the orthogonal projection matrix onto the column space of $\MM$, $\MM\tp$, respectively.
    If  $\AA, \CC$ are nonsingular, then
    \eq{
        \MM\dg = \PP_{\MM}\CC\inv\BB\dg\AA\inv\PP_{\MM\tp}.
    }

\end{fact}

\begin{fact}\label{fact:la2scup1}
    Let $\UU \in \MS{n}{n} \ (n \geq 2)$ be a \strc\ symmetric Laplacian, then,
    for any partition $F, C$ of $[n]$, we have
    \eq{
        \la_2\pr{\UU} \leq \la_2\pr{\sc{\UU, F}}.
    }
    And
    \eq{
        \min_{i\in [n]} \UU_{ii} \geq \frac{1}{2}\la_2\pr{\UU}.
    }

\end{fact}
\begin{proof}
    By Lemma~C.1 of~\cite{cohen2018solving}, $\sc{\UU, F} = \pr{\PP_S \UU\dg \PP_S }\dg$, where $\PP_{S}$ is the projection matrix onto the image space of $\sc{\UU, F}$.
    Then, $\la_2\pr{\sc{\UU, F}} = \frac{1}{\la_{\max}\pr{\PP_S \UU\dg \PP_S}} \geq \frac{1}{\la_{\max}\pr{\UU\dg}} = \la_2\pr{\UU}, $
    where the inequality uses the fact that $\PP_S$ is a projection matrix.

    Let $\dr{\ee_i}_{i=1}^n$ be standard basis of $\Real^n$, then
    $
        \la_2\pr{\UU} = \inf_{\nt{\xx} = 1, \xx\perp \one} \xx\tp \UU \xx
        = \inf_{\xx \neq \frac{1}{n}\pr{\xx\tp\one}\one} \frac{\xx\tp \UU \xx}{\nt{\xx - \frac{\one\tp\xx}{n}\one}^2 }
        \leq \min_{i\in[n]} \frac{\ee_i\tp \UU \ee_i}{\nt{\ee_i - \frac{\one\tp\ee_i}{n}\one}^2 }
        \leq \min_{i\in[n]} \frac{\UU_{ii}}{1 - \frac{2}{n} + \frac{1}{n}}
        \leq 2\min_{i\in [n]} \UU_{ii},
    $
    where the last inequality is from $n \geq 2$.

\end{proof}

\begin{fact}\label{fact:sequenceproductelemta}
    For 2 sequences of matrices $\AA_1, \AA_2, \cdots, \AA_N \in \MS{n}{n}$ and $\BB_1, \BB_2, \cdots, \BB_N\in \MS{n}{n}$.
    For any $1\leq k\leq N$,
    \eq{
        &\ni{\AA_1 \cdots \AA_N - \BB_1 \cdots \BB_N}  \\
        \leq& \pr{n\sum_{i=1}^{k}\ni{\AA_i - \BB_i} + \sum_{i=k+1}^{N}\no{\AA_i - \BB_i}}\prod_{i=1}^{k}\max\dr{1, \ni{\AA_i} + \ni{\AA_i - \BB_i}} \\
        & \cdot \prod_{i=k+1}^{N}\max\dr{1, \no{\AA_i} + \no{\AA_i - \BB_i} }.
    }

\end{fact}
\begin{proof}
  It follows directly from the expansion
  \eq{
    &\AA_1 \cdots \AA_N - \BB_1 \cdots \BB_N
    = \sum_{i=1}^{N}\prod_{j=1}^{i-1}\AA_j \pr{\AA_i - \BB_i} \prod_{j=i+1}^{N} \BB_j
  }
  and the fact that $\no{\CC} \leq n\ni{\CC}$ for any matrix $\CC \in \MS{n}{n} $.

\end{proof}

\begin{fact}\label{fact:FLFgood}
    (Lemma~6.3 of~\cite{cohen2018solving}, paraphrased)
    Consider matrix $\AA \in \MS{n}{n}$, where $\ker\pr{\AA} = \ker\pr{\AA\tp}$.
    Let $F_1, F_2, \cdots, F_d$ be a partition of $[n]$.
    Denote $E_1 = \emptyset$, $E_i = \cup_{j=1}^{i-1} F_i \ (2\leq i\leq d)$ and $C_0 = [n]$, $C_i = [n] \dele \pr{E_i \cup F_i} \ (1\leq i\leq d)$.
    Suppose that $\U{\sc{\AA, E_i}}_{F_i F_i}$ is PD for any $1\leq i\leq d - 1$ and $\U{\sc{\AA, E_i}}$ is PSD for any $1\leq i\leq d$.
    Let $\dr{\theta_i}_{i=1}^d$ be nonnegative numbers such that $\sum_{i=1}^{d}\theta_i = 1$.
    Define the matrix $\BB = \sum_{i=1}^{d} \theta_i\putmat{\sc{\AA, E_i}, C_{i-1}, C_{i-1}, n} $, then,
    \eq{
        \tpp{\AA\dg} \BB \AA\dg \pleq \BB\dg.
    }

\end{fact}

\begin{fact}\label{fact:repprvpleq}
    If $\AA\pleq \BB$, then
    $
        \rep{k, C, \AA} \pleq \rep{k, C, \BB},
    $
    $
        \repp{k, C, \AA, N} \pleq \repp{k, C, \BB, N}.
    $

\end{fact}
\begin{fact}\label{fact:repatimesb1}
    $\rep{b, C, \rep{a, C, \AA}} = \rep{a \cdot b, C, \AA}.  $
\end{fact}

\def\putb#1{#1\Fn + \Cn}

\section{Exact Partial Block Elimination }\label{sec:exactPBE}
Let $\PP$ be a permutation matrix such that $\PP\LL\PP\tp = \mx{\LL_{FF} & \LL_{FC}  \\ \LL_{CF} & \LL_{CC} } $ in this Section.
We initiate $\Lt{0} = \LL $,  
$\DD = \Diag{\Lt{0}}$ and $\At{0} = \DD - \Lt{0}$, and then update for $k = 1, 2, \cdots, $
\begin{subequations}\label{eq:randomwalk1}
    \begin{align}
     \Lt{k} &=
            \PP\tp\mx{
                \DD_{FF} & - \At{k-1}_{FC}  \\
                - \At{k-1}_{CF} & 2\Lt{k-1}_{CC}
            }\PP
            - \At{k-1}_{:, F}\DD_{FF}\inv\At{k-1}_{F,:} \notag  \\
            &=
            \PP\tp\mx{
                \DD_{FF} - \At{k-1}_{FF} \DD_{FF}\inv\At{k-1}_{FF } & - \pr{\II + \At{k-1}_{FF}\DD_{FF}\inv}\At{k-1}_{FC}  \\
                - \At{k-1}_{CF}\pr{\II + \DD_{FF}\inv\At{k-1}_{FF}} & 2\Lt{k-1}_{CC} - \At{k-1}_{CF}\DD_{FF}\inv\At{k-1}_{FC}
            }\PP \label{line:densebiclique}  \\
   \At{k} &= \PP\tp\mx{\DD_{FF} & \\ & \Diag{\Lt{k}}}\PP - \Lt{k}.  \label{line:denseA}
   \end{align}
\end{subequations}

The following lemmas characterize the performance of the ideal but inefficient scheme~\eqref{eq:randomwalk1}.
\begin{lemma}\label{lem:scLtkequal1}
    Let $\Lt{k}$ be the output of running $k$ steps of
    $\textsc{IdealSchur}(\LL, F)$.
    At any step we have:
    \eql{\label{scLtk}}{
        \sc{\Lt{k}, F} = 2^k\sc{\LL, F}.
    }
\end{lemma}
\begin{proof}

  By~\eqref{eq:D-A},
  \eq{
    &2\sc{\Lt{k-1}, F} = 2\Lt{k-1}_{CC} - 2\At{k-1}_{CF} \pr{\Lt{k-1}_{FF}}\inv \At{k-1}_{FC}  \\
    =& 2\Lt{k-1}_{CC} - 2\At{k-1}_{CF} \pr{\DD_{FF} - \At{k-1}_{FF}}\inv \At{k-1}_{FC}  \\
    =& 2\Lt{k-1}_{CC}  - \At{k-1}_{CF} \DD_{FF}\inv \At{k-1}_{FC} \\
     & - \At{k-1}_{CF}\pr{\II + \DD_{FF}\inv\At{k-1}_{FF}} \pr{\DD_{FF} - \At{k-1}_{FF}\DD_{FF}\inv\At{k-1}_{FF}}\inv \pr{\II + \At{k-1}_{FF}\DD_{FF}\inv}\At{k-1}_{FC}  \\
    =& \Lt{k}_{CC} - \At{k}_{CF} \pr{\Lt{k}_{FF}}\inv \At{k}_{FC}
    =  \sc{\Lt{k}, F}.
  }
  Then, the relation~\eqref{scLtk} follows by  induction.

\end{proof}

\begin{lemma}\label{lem:LtkE}
    For any $k\geq 0$, $\Lt{k}$ is an Eulerian Laplacian.
\end{lemma}
\begin{proof}
  We prove it by induction.
  Firstly, $\Lt{0} = \LL$ is an Eulerian Laplacian.

  Assuming that $\Lt{k-1}$ is an Eulerian Laplacian, we prove that $\Lt{k}$ is also an Eulerian Laplacian.
  By induction hypothesis, $\Lt{k-1 }\one = \zero$, i.e.,
  \eq{
    &\At{k-1}_{:, F}\one = \At{k-1}_{FF}\one + \At{k-1}_{FC}\one = \DD_{FF}\one  \\
    &\At{k-1}_{CF}\one = \Lt{k-1}_{CC}\one.
  }
  Thus,
  \eq{
    \PP\Lt{k}\PP\tp\one =&
    \mx{
        \DD_{FF} & - \At{k-1}_{FC} \\
        - \At{k-1}_{CF} & 2\Lt{k-1}_{CC}
    }\one
    - \At{k-1}_{:,F}\DD_{FF}\inv\At{k-1}_{:,F}\one  \\
    =& \vc{
           \DD_{FF}\one - \At{k-1}_{FC}\one  \\
           - \At{k-1}_{CF}\one + 2\Lt{k-1}_{CC}\one
        }
        - \vc{
                \At{k-1}_{FF}\one  \\
                \At{k-1}_{CF}\one
            }  \\
    =& \zero.
  }
  Since $\PP$ is a permutation matrix, we have $\Lt{k}\one = \zero$.
  Similarly,  $\one\tp\Lt{k} = \zero\tp$.
  Therefore, $\Lt{k}$ is an Eulerian Laplacian.
  The result then follows by induction.

\end{proof}

\begin{lemma}\label{lem:LtkCClimequal1}
    The result of the update formula~\eqref{eq:randomwalk1} satisfies
    $
        \lim_{k \arr +\infty}\frac{1}{2^k}\Lt{k}_{CC} = \sc{\LL, F}.
    $

\end{lemma}
\begin{proof}
  Since $\LL_{FF} $ is $\alp$-RCDD, $\ni{\DD_{FF}\inv\At{0}_{FF}} \leq \frac{1}{1 + \alp} $.
  By the equality $\Lt{k}_{FF} = \DD_{FF} - \At{k-1}_{FF}\DD_{FF}\inv\At{k-1}_{FF}$ from~\eqref{line:densebiclique},
  we have
  $
    \DD_{FF}\inv\At{k} = \DD_{FF}\inv\At{k-1}_{FF}\DD_{FF}\inv\At{k-1}_{FF}.
  $
  Thus, by induction,
  \eql{\label{eq:DinvAsuperl}}{
    \ni{\DD_{FF}\inv\At{k}} \leq \pr{\frac{1}{1 + \alp}}^{2^k}. 
  }
  By Lemma~\ref{lem:LtkE}, $\ni{\At{k}_{CF}} \leq n\ni{\At{k}_{CF}} \leq n\nt{\DD_{FF}}$, $\ni{\DD_{FF}\inv\At{k}_{FC}} \leq 1$.

  Combining the above arguments,
  \eql{\label{eq:exactAkzero}}{
    &\lim_{k\arr +\infty}\frac{1}{2^k}\ni{\At{k}_{CF}\pr{\Lt{k}_{FF}}\inv\At{k}_{FC}}  \\
    =& \lim_{k\arr +\infty}\frac{1}{2^k}\ni{\At{k}_{CF}\pr{\II - \DD_{FF}\inv\At{k}_{FF}}\inv\DD_{FF}\inv\At{k}_{FC}}  \\
    \leq& \limsup_{k\arr +\infty}\frac{1}{2^k}\ni{\At{k}_{CF}}\ni{\pr{\II - \DD_{FF}\inv\At{k}_{FF}}\inv}\ni{\DD_{FF}\inv\At{k}_{FC}}  \\
    \leq& n \nt{\DD_{FF}} \cdot \limsup_{k\arr +\infty}\frac{1}{2^k} \cdot  \limsup_{k\arr +\infty}\ni{\pr{\II - \DD_{FF}\inv\At{k}_{FF}}\inv} = 0,
  }
  i.e.,
  \eq{
    \lim_{k \arr +\infty} \frac{1}{2^k}\At{k}_{CF}\pr{\Lt{k}_{FF}}\inv\At{k}_{FC} = \zero.
  }

  Then, using Lemma~\ref{lem:scLtkequal1} and the above equation, we have
  \eq{
      &\lim_{k \arr +\infty}\frac{1}{2^k}\Lt{k}_{CC} = \lim_{k \arr +\infty}\frac{1}{2^k}\sc{\Lt{k}, F} + \lim_{k \arr +\infty}\frac{1}{2^k}\At{k}_{CF}\iv{\Lt{k}_{FF}}\At{k}_{FC} 
    = \sc{\LL, F}.
  }

\end{proof}

\def\Lm#1{\widetilde{\mathcal{L}}^{\pr{#1}}}
\section{Sparsifying Directed Laplacians}\label{sec:sparsify}

First, we check that the Eulerian Laplacian sparsifier in
Section~3 of~\cite{cohen2017almost} meets the requirements
of Theorem~\ref{thm:SparEoracle1}.
This procedure can be briefly summarized as:
\begin{enumerate}
    \item decompose $\LL = \sum_{i=1}^{K} \calLt{i} $ such that each $\U{\calLt{i}} $ is an expander;
    \item sample the entries in the adjacency matrix of each
    $\calLt{i}$ and use a patch matrix to keep the row sums and the column sums invariant.
\end{enumerate}
This procedure was analyzed in~\cite{cohen2017almost}
by
(1) using matrix concentration inequalities to bound the errors in each adjacency matrix with respect to the in-degree and out-degree diagonal matrix; (2) using the property of the expander to bound the errors with respect to $\U{\calLt{i}}$ and in turn $\U{\LL}$.

Next, we give a precise bound of the running time of
directed Laplacian sparsification by combining the
expander decomposition in~\cite{saranurak2019expander}
with the degree-fixing on expanders routine from
Section~3 of~\cite{cohen2017almost}.

By setting $\phi = O\pr{1/\log^3 n}$ in Theorem~4.1 of~\cite{saranurak2019expander} and deleting the edges recursively, we can have a $\pr{s, \phi, 1} $-decomposition (Definition~3.14 of~\cite{cohen2017almost}) of the original directed graph $\calG[\LL]$, denoted by $\dr{\calLt{i}}_{i=1}^K$, and each $\U{\calLt{i}}$ is a $\phi$-expander.
Here $s = n\log n$ is the sum of the sizes of the subgraphs $\calLt{i}$.
The running time of this step is $O\pr{m\log^8 n}$.

By Cheeger's inequality, the spectral gap $\delta$ of each $\phi$-expander is $O\pr{\phi^2}$.
Then, by Lemma~3.13 of~\cite{cohen2017almost}, we can have a $\Lm{i}$ by sampling edges in  $\calLt{i}$, such that $\sum_{i=1}^{K}\nnz{\Lm{i}} \leq s\log n / \delta^2 = O\pr{n\log^{14} n}.   $
And $\Lm{i}$ is an $O(1)$-asymmetric approximation of $\calLt{i}$.
By summing up $\widetilde{\calL} = \sum_{i=1}^{K}\calLt{i}$, we have a $O(1)$-asymmetric approximation of $\LL$, with $\nnz{\widetilde{\calL}} = O(n\log^{14} n)$.

Putting these costs into Theorem~\ref{thm:TSENSEsolver1},
specifically setting $\TSE\pr{m, n, 1} = O\pr{m\log^8 n}$,
$\NSE\pr{n, 1} = O\pr{n\log^{14} n}$,
gives that the overall (construction + solve)
running times of our algorithm is $O\pr{m\log^8 n + n\log^{15}n\log\frac{n}{\eps}} + \Otil{n\log^{23} n}$.

With a slower processing time,
we can use short cycle based Eulerian sparsifiers
to get a smaller Schur complement chain,
and solve for each query vector faster.

Given an Eulerian Laplacian $\LL$,
the current best bound for the nonzero entries of its sparsifier is
$\Otil{n \log^4{n} \epsilon^{-2}}$ edges (Lemma~\ref{lem:shortcycleexistsparseEL}).
\begin{lemma}\label{lem:shortcycleexistsparseEL} (Existence of Eulerian Laplacian sparsifier)
\cite{CGPSSW18}
For any Eulerian Laplacian $\LL\in \MS{n}{n}$ with $\nnz{\LL} = m$ and error parameter $\eps \in (0, 1)$, there is a Eulerian Laplacian $\Lap$ such that $\nnz{\Lap} \leq O(n \log^{4}n \eps^{-2})$  and $\Lap - \LL \aleq \eps \cdot \U{\LL}. $
Such an $\Lap$ can be constructed in time $O(mn\log^{O(1)} n)$.

\end{lemma}

Invoking the Eulerian sparsifier routine of Lemma~\ref{lem:shortcycleexistsparseEL} above within each iteration of Algorithm~\ref{alg:SCC},
with error set to $\eps = O(\frac{1}{i^2})$
at the $i$-th iteration, we obtain,
after $O(n^2 \log^{O(1)} n )$ preprocessing time,
an $\{\alp, \frac{1}{16(1+\alp)},
\{\frac{O(1)}{i^2}\}_{i=1}^{d}\}$-Schur complement chain
\[
\dr{\dr{\Stt{i}}_{i=1}^d, \dr{F_i}_{i=1}^d }
\]
with $\alp = O(1)$, $d = O(\log n)$ and $\sum_{i=1}^{n}\nnz{\Stt{i}} = O(n\log^4 n)$.
Then, Corollary~\ref{coro:shortcyclep1} follows
similarly to the proof of Theorem~\ref{thm:TSENSEsolver1}.

\section{Supporting Lemmas and Omitted Proofs }\label{sec:someprfs}

\begin{lemma}\label{lem:SE}
    There is a routine $\SE$ which takes in an Eulerian Laplacian $\LL\in\MS{n}{n}$,
    error parameter $\eps\in (0, 1)$ and a subset $F\sleq [n]$, where $\nnz{\LL} = m$.
    And then, $\SE$ runs in $O\pr{\TSE\pr{m, n, \dlt}}$ time to return an Eulerian Laplacian $\Ltil\in \MS{n}{n}$ such that $\Diag{\Ltil} = \Diag{\LL}$, $\Ltil_{FF}\one = \LL_{FF}\one$, $\Ltil_{FF}\tp\one = \LL_{FF}\tp\one$, $\nnz{\Ltil} = O\pr{\NSE\pr{n, \dlt} }$ and
    $
        \Ltil - \LL \aleq \eps \cdot \U{\LL}
    $
    with high probability.

    Analogously, under the same conditions as in Lemma~\ref{lem:SparP}, there is a routine $\SP$ which takes in vectors $\xx, \yy$, error parameter $\eps$, probability $p$ and a subset $F\sleq [n]$, runs in $O\pr{m\eps^{-2}\log\frac{m}{p}}$ to return with high probability a nonnegative matrix $\BB$ which possesses all the properties of $\AA$ in Lemma~\ref{lem:SparP}.
    In addition, $\BB_{FF}\one = \pr{\yy_F\tp\one}\xx_F$ and $\one\tp\BB_{FF} = \pr{\xx_F\tp\one}\yy_{F}\tp$.

\end{lemma}
\begin{proof}
  We apply $\SparE$ to the directed Laplacians
  \eq{
    &\mx{\LL_{FF} - \Diag{\one\tp\LL_{FF}} & \zerom{F}{C} \\ \zerom{C}{F} & \zerom{C}{C}},
    \mx{\zerom{F}{F} & \zerom{F}{C} \\ \zerom{C}{F} & \LL_{CC} - \Diag{\one\tp\LL_{CC}}},  \\
    &\mx{-\Diag{\one\tp\LL_{CF}} & \zerom{F}{C} \\ \LL_{CF} & \zerom{C}{C} },
    \mx{\zerom{F}{F} & \LL_{FC} \\ \zerom{C}{F} & -\Diag{\one\tp\LL_{FC}}}
  }
respectively
and summing up the resulting sparsified matrices, we have the $\SparE$ in Lemma~\ref{lem:SE}.
Analogously, by setting $\pr{\xx, \yy}$ as
\eq{
    \pr{\vc{\zerov{F} \\ \xx_{C}}, \vc{ \yy_{F} \\ \zerov{C}}},
    \pr{\vc{\xx_{F} \\ \zerov{C} }, \vc{\yy_{F} \\ \zerov{C}}},
    \pr{\vc{\xx_{F} \\ \zerov{C} }, \vc{\zerov{F} \\  \yy_{C} }},
    \pr{\vc{\zerov{F} \\ \xx_{C}}, \vc{\zerov{C} \\ \yy_{C} }}
}
in $\SparP$
respectively, we get $\SP$ in Lemma~\ref{lem:SE}.

The performance of  $\SE$, $\SP$ follows directly by  the fact that $\SparE$ preserves the diagonal entries and the row sums,  $\SP $ preserves the row and column sums.

\end{proof}
\begin{remark}
    By carefully designing a sampling rule, $\SE$ and $\SP$ can be replaced by a single sparsification procedure. Here, we use $\SparE$ and $\SparP$ to construct $\SE$ and $\SP$ merely for simplicity.
    When invoking $\SparP$, $\SP$ in this paper, we set $p = O\pr{\frac{1}{\poly{n}}}$ and omit the probability parameter $p$ for notational simplicity.
    For $\xx = \zero$ or $\yy = \zero$, both $\SparP$ and $\SP$ return $\zero\in \MatSize{n}{n}$ naturally.
\end{remark}

We have the following properties of Algorithm~\ref{alg:SparSchur}.
\begin{lemma}\label{lem:LapetcpropSparSchurCpmt1}
    With high probability, the following statements hold:
    \begin{enumerate}[(i)]
      \item $\dr{\Ltt{k}}_{k=0}^K$ are Eulerian Laplacians; \label{item:LttkE}

      \item $\dr{\Att{k}}_{k=0}^{K} $ are nonnegative matrices satisfying
        \eql{\label{eq:Attsuperl}}{
            \ni{\DD_{FF}\inv\Att{k}_{FF}} \leq \pr{\frac{1}{1 + \alp}}^{2^k},\ \forall 0\leq k\leq K;
        }
        \label{item:Attsuperl}

      \item $\SS, \Sap$ are Eulerian Laplacians; \label{item:SapE}
      \item The matrix $\Rap$ satisfies
        $\Rap\one = \Rap\tp\one = \zero $ and
        \eql{\label{eq:R}}{
                \nt{\Rap} \leq \frac{n^2\nt{\DD_{FF}}}{2^{K-1} \alp}\pr{\frac{1}{1 + \alp}}^{2^K}.
            }

    \end{enumerate}

\end{lemma}
\begin{proof}

    We prove~(\ref{item:LttkE}) by  induction.
    Firstly, $\Ltt{0} = \LL$ is an Eulerian Laplacian.
    Suppose that $\Ltt{k-1}$ is an Eulerian Laplacian, we will prove the Eulerianness for $k$.

    Since $\SP$ preserves the row sum, \eq{\Ytt{k}\one = \sum_{i\in F }\frac{1}{\DD_{ii}}\Att{k-1}_{:,i}\Att{k-1}_{i,:}\one = \Att{k-1}_{:,F}\DD_{FF}\inv\Att{k-1}_{F,:}. }

    In this proof, we denote
    \eql{\label{eq:Wtk}}{
        \Wt{k} \defeq \PP\tp\mx{
            \DD_{FF} & - \Att{k-1}_{FC}  \\
            - \Att{k-1}_{CF} & 2\Ltt{k-1}_{CC}
        }\PP
        - \Att{k-1}_{:, F}\DD_{FF}\inv\Att{k-1}_{F,:},
    }
    where $\PP$ is the permutation matrix defined in Section~\ref{sec:exactPBE}.
    By similar arguments with Lemma~\ref{lem:LtkE}, $\Wt{k}$ is an Eulerian Laplacian.

    Then, we have
    \eq{
        \Ltt{k, 0}\one = \Wt{k}\one + \pr{\Att{k-1}_{:,F}\DD_{FF}\inv\Att{k-1}_{F, :} - \Ytt{k}}\one = \zero.
    }
    By combining with the fact that $\Ytt{k}$ is a nonnegative matrix (from Lemma~\ref{lem:SparP}), $\Ltt{k, 0}$ is an Eulerian Laplacian.
    Then, (\ref{item:LttkE}) follows by Lemma~\ref{lem:SE} and induction.

    The nonnegativity of $\Att{k}$ follows directly by Lemma~\ref{lem:SE}.
    Since $\SP$ preserves the row sum  and $\SE$ preserves the diagonal and row sum on the submatrix $\Ltt{k,0}_{FF}$, we have  
    \eq{
       &\ni{\DD_{FF}\inv\Att{k}_{FF}} = \DD_{FF}\inv\Att{k}_{FF}\one = \DD_{FF}\inv\sum_{i\in F }\frac{1}{\DD_{ii}}\Ytt{k,i}_{FF}\one
       = \DD_{FF}\inv\sum_{i\in F } \frac{1}{\DD_{ii}} \Att{k-1}_{F,i} \Att{k-1}_{i,F} \one  \\
       =& \DD_{FF}\inv\Att{k-1}_{FF}\DD_{FF}\inv  \Att{k-1}\one = \ni{\DD_{FF}\inv\Att{k-1}_{FF}\DD_{FF}\inv\Att{k-1}_{FF}} \leq \ni{\DD_{FF}\inv\Att{k-1}_{FF}}^2.
    }
    Then,~(\ref{item:Attsuperl}) can be shown by induction.

    For~(\ref{item:SapE}), it can be shown directly by the nonnegativity of $\Xap$ and~\eqref{item:LttkE}  that all off-diagonal entries of $\Stt{0}$ is non-positive.
    By Fact~\ref{fact:ESchurE}, $\sc{\Ltt{K}, F}$ is an Eulerian Laplacian.
    As $\SparP$ preserves the row sum,
    \eql{\label{eq:Sttone}}{
        &2^K \Stt{0}\one = \sc{\Ltt{K}, F}\one + \Att{K}_{CF}\pr{\DD_{FF} - \Att{K}_{FF}}\inv\Att{K}_{FC}\one - \Xap\one  \\
        =&  \Att{K}_{CF}\pr{\DD_{FF} - \Att{K}_{FF}}\inv\Att{K}_{FC}\one - \sum_{i\in F }\frac{1}{\DD_{ii}}\Att{K}_{C,i}\Att{K}_{i,C}  \one  \\
        =& \Att{K}_{CF}\DD_{FF}\inv\Att{K}_{FF}\pr{\pr{\DD_{FF} - \Att{K}_{FF}}\inv - \DD_{FF}\inv}\Att{K}_{FC}\one  \\
        =& \Att{K}_{CF}\DD_{FF}\inv\Att{K}_{FF}\sum_{i=0}^{+\infty}\pr{\DD_{FF}\inv\Att{k}_{FF}}^i\DD_{FF}\inv\Att{K}_{FC}\one.
    }
    Then, $\Stt{0}\one$ is a nonnegative vector.
    Analogously, $\one\tp\Stt{0}$ is also nonnegative.
    So, $\Stt{0}$ is RCDD.
    Since $\Stt{0}\one$, $\one\tp\Stt{0}$ are nonnegative,
    from the way we compute the patching matrix $\RR$, off-diagonal entries of $\RR$ are non-positive.
    Thus, all off-diagonal entries of $\Sap$ are non-positive.
    It follows by the definition of $\RR_{1,1}$ and direct calculations that $\Sap\one = \Sap\tp\one = \zero$.
    Then, we have shown  $\Sap$ is an Eulerian Laplacian.
    Thus, $\SS$ is also an Eulerian Laplacian by the definition of the oracle $\SparE$.

    By~\eqref{eq:Sttone} and~\eqref{eq:defRap}, $\Stt{0}\one = \pr{\Rap - \RR}\one$.
    Then, as we have just shown $\Sap$ is an Eulerian Laplacian, $\Rap\one = \pr{\Stt{0} + \RR}\one = \Sap\one = \zero$.
    Analogously, $\Rap\tp\one = \zero$.
    As we have shown $\RR$ is non-positive, we have $\ni{\RR}  \leq \one\tp\RR\one = \one\tp\Stt{0}\one$.
    As $ \Rap - \RR = \Att{K}_{CF}\DD_{FF}\inv\Att{K}_{FF}\pr{\pr{\DD_{FF} - \Att{K}_{FF}}\inv - \DD_{FF}\inv}\Att{K}_{FC}  $ is nonnegative,
    we have $\ni{\Rap - \RR} \leq \one\tp\pr{\Rap - \RR}\one \comeq{\eqref{eq:Sttone}}  \one\tp\Stt{0}\one $.
    Thus, $\ni{\Rap} \leq \ni{\RR} + \ni{\Rap - \RR} \leq 2 \cdot \one\tp\Stt{0}\one. $
    As $\SP$ preserves the row sum, by Lemma~\ref{lem:LtkE}, 
    $
        \ni{\Att{K}_{CF}} \leq n\no{\Att{K}_{CF}} \leq n\nt{\DD_{FF}}.
    $
    Since $\Ltt{K}$ is an Eulerian Laplacian, $\ni{\DD_{FF}\inv\Att{K}_{FC}} \leq 1$.
    Then, by~\eqref{eq:Sttone},
    \eq{
        \one\tp\Stt{0}\one \leq \frac{n\ni{\Stt{0}\one} }{2^K}
        \leq \frac{n^2\nt{\DD_{FF}}}{2^K }\pr{\frac{1}{1 + \alp}}^{2^K}\sum_{i=0}^{+\infty}\pr{\frac{1}{1 + \alp}}^{2^i}
        \leq \frac{n^2\nt{\DD_{FF}}}{\alp}\pr{\frac{1}{1 + \alp}}^{2^K}.
    }
    So, $\ni{\Rap} \leq \frac{n^2\nt{\DD_{FF}}}{2^{K-1}\alp}\pr{\frac{1}{1 + \alp}}^{2^K}.  $
    Analogously, $\no{\Rap} \leq \frac{n^2\nt{\DD_{FF}}}{2^{K-1} \alp}\pr{\frac{1}{1 + \alp}}^{2^K}.  $
    Then,~\eqref{eq:R} follows by Fact~\ref{fact:ninobnt}.

\end{proof}

\begin{proof}[Proof of Lemma~\ref{lem:EYEXEtt}]

  For any nonnegative vectors $\aa, \bb \in \Real^n$, we define $\Uvc{\aa, \bb}$ as the undirectification of a biclique as follows:
  \eq{
    \Uvc{\aa, \bb} = \UG{\pr{\one\tp\aa}\Diag{\bb} - \aa\bb\tp } =
    \frac{1}{2}\pr{\pr{\bb\tp\one}\Diag{\aa} + \pr{\aa\tp\one}\Diag{\bb}} - \aa\bb\tp - \bb \aa\tp.
  }

    Then, by Lemma~\ref{lem:SparP}, Lemma~\ref{lem:SE} and Lemma~\ref{lem:ne}, we have
    \eq{
        2\ex\tp \EY{k,i} \yy \leq \eps\pr{\ex\tp\Uvc{\Att{k-1}_{:,i}, \tpp{\Att{k-1}_{i,:}}}\ex + \ey\tp\Uvc{\Att{k-1}_{:,i}, \tpp{\Att{k-1}_{i,:}}}\ey}.
    }
    Then, summing over $i\in F $  yields that
    \eql{\label{eq:exEYeyUt}}{
        2 \ex\tp \EY{k} \ey \leq \eps \pr{\ex\tp \Ut{k} \ex + \ey\tp \Ut{k} \ey},
    }
    where $\Ut{k} = \sum_{i \in F  }\frac{1}{\DD_{ii}}\Uvc{\Att{k-1}_{:,i}, \tpp{\Att{k-1}_{i,:}}}$.

    Since $\Ut{k}$ is a weighted summation of symmetric Laplacians, $\Ut{k}$ is also a symmetric Laplacian.

    We also define $\Wt{k}$ as in~\eqref{eq:Wtk} in this proof.
    And we have shown $\Wt{k}$ is an Eulerian Laplacian.

    By setting $\LL = \Ltt{k-1}$ in Lemma~\ref{lem:Mti} and Lemma~\ref{lem:ULtm}, we have
    \eql{\label{cor:Wtk}}{\U{\Wt{k}} \pleq 2\pr{3 + \frac{2}{\alp}} \U{\Ltt{k-1}}.  } .

    By the definition of $\Wt{k}$,
    \eq{ \U{\Wt{k}} = \Zt{k} - \frac{1}{2}\pr{\Att{k-1}_{:, F}\DD_{FF}\inv\Att{k-1}_{F, :} + \tpp{\Att{k-1}_{F,:}}\DD_{FF}\inv\tpp{\Att{k-1}_{:, F}}}, }
    where $\Zt{k}$ is a matrix whose off-diagonal entries are all non-positive.

    And by the definition of $\Ut{k}$, we have
    \eq{ \Ut{k} = \Dt{k} - \frac{1}{2}\pr{\Att{k-1}_{:, F}\DD_{FF}\inv\Att{k-1}_{F, :} + \tpp{\Att{k-1}_{F,:}}\DD_{FF}\inv\tpp{\Att{k-1}_{:, F}}}, }
    where $\Dt{k}$ is a diagonal matrix.
    Thus, the off-diagonal entries of $\Gt{k} \defeq \U{\Wt{k}} - \Ut{k} = \Zt{k} - \Dt{k}$ is all non-positive.
    And since $\U{\Wt{k}}$ and $\Ut{k}$ are both symmetric Laplacians, we have
    $\Gt{k}\one = \tpp{\Gt{k}}\one = \zero$.
    So, $\Gt{k}$ is a symmetric Laplacian. Then, $\Gt{k} \pgeq \zero$.
    Thus, we have
    \eq{
        \U{\Wt{k}} = \Ut{k} + \Gt{k} \pgeq \Ut{k}.
    }
    By combining with~\eqref{eq:exEYeyUt}, we have
    \eql{\label{eq:EYUWt}}{
        2\ex\tp \EY{k} \ey \leq \eps \pr{\ex\tp \U{\Wt{k}} \ex + \ey\tp \U{\Wt{k}} \ey}.
    }
    By Fact~\ref{lem:ne} and Fact~\ref{fact:aleqU}, $\U{\EY{k}} \pleq \eps \U{\Wt{k}}$.

    It follows by the definitions of $\Wt{k}$ and $\EY{k}$ that
    $
        \Ltt{k, 0} = \Wt{k} + \EY{k}.
    $
    Thus,
    \eq{
        \U{\Ltt{k, 0}} = \U{\Wt{k}} + \U{\EY{k}}  \pleq \pr{1 + \eps}\U{\Wt{k}}.
    }
    By Lemma~\ref{lem:SE},
    \eql{\label{eq:Ettk011112}}{
        2\ex\tp \Ett{k, 0} \ey \leq \eps \pr{\ex\tp \U{\Ltt{k, 0}} \ex + \ey\tp \U{\Ltt{k, 0}} \ey} \leq \eps\pr{1 + \eps}\pr{\ex\tp \U{\Wt{k}} \ex + \ey\tp \U{\Wt{k}} \ey}.
    }
    By~\eqref{cor:Wtk},~\eqref{eq:EYUWt},~\eqref{eq:Ettk011112}, Fact~\ref{lem:ne} and the relation $\Ett{k} = \EY{k} + \Ett{k,0} $, we have
    \eq{
        \Ett{k} \aleq \pr{\eps + \pr{1 + \eps}\eps} \U{\Wt{k}} \pleq 2\pr{3 + \frac{2}{\alp}}\pr{2\eps + \eps^2 }\U{\Ltt{k-1}}.
    }
    By the definition of  $\epsz $,~\eqref{eq:EYLtt} follows.
    The inequality~\eqref{eq:EXscLtt} follows analogously.

\end{proof}

\begin{proof}[Proof of Lemma~\ref{enum:Q6}]
  Denote \eq{
                \xhatt{i} = \big(\xx_C\tp \underbrace{\xx_{F}\tp \ \cdots \ \xx_F\tp}_{\text{$2^i$ repetitions of $\xx_F\tp$}}\big)\tp.
              }
    By Lemma~\ref{lem:Mti},
    $
        \tpp{\xhatt{i}} \U{\Mt{0, i}} \xhatt{i} = 2^i \xx\tp \U{\LL} \xx.
    $
    Thus,
    \eq{
        \tpp{\xhatt{k}} \repFC{2^{k-i}, F, C, \U{\Mt{0, i}}} \xhatt{k} = 2^{k-i} \cdot 2^i \xx\tp\U{\LL} \xx = 2^k\xx\tp\U{\LL}\xx.
    }
    Since $\XL{0} = \U{\LL}$ and the sum of coefficients of the terms on the RHS of~\eqref{eq:newday1} is $1$,
    it follows by induction that
    \eql{\label{eq:2pwerk1}}{
        \tpp{\xhatt{k}} \XL{k} \xhatt{k} = 2^k \xx\tp \U{\LL} \xx.
    }

  By Fact~\ref{fact:alpRCDDPSDpPD1}, $\XL{k}_{-[n], -[n]}$ is PD.
  Then, by Fact~\ref{fact:Schurxusmall} and~\eqref{eq:2pwerk1},  we have
  \eq{
    \xx\tp\sc{\XL{k}, -[n]}\xx \leq \tpp{\xhatt{k}}\XL{k}\xhatt{k} = 2^k \xx\tp\U{\LL}\xx,  
  }
  i.e., $\sc{\XL{k}, -[n]}\pleq 2^k\U{\LL}$.
  Then, using Fact~\ref{fact:scprvpleq} and Fact~\ref{fact:scUpleqUsc}, we have
    $\sc{\XL{k}, -C} \pleq 2^k \sc{\U{\LL}, F} \pleq 2^k \U{\sc{\LL, F}}. $

\end{proof}

\begin{proof}[Proof of Theorem~\ref{thm:SCC}]
  By Lemma~\ref{thm:SparSchur}, we have
  \eq{
    \St{i+1} - \sc{\Stt{i}, F_i} \aleq  \dlt_{i+1}' \U{\sc{\Stt{i}, F_i}} = \frac{\dlt }{3 i^2 }\U{\sc{\Stt{i}, F_i}}.
  }
  Then, $\pr{1 - \dlt_{i+1}'}\U{\sc{\Stt{i}  , F_i}} \pleq \U{\St{i+1}} \pleq \pr{1 + \dlt_{i+1}'} \U{\sc{\Stt{i}, F_i}}. $
  Thus, we have
  \eq{
    \U{\Stt{i+1}} = \frac{1}{1 - \dlt_{i+1}' }\U{\St{i+1}  } \pgeq \U{\sc{\Stt{i}, F_i}}
  }
  and
  \eq{
    &\Stt{i+1} - \sc{\Stt{i}, F_i} = \St{i+1} - \sc{\Stt{i}, F_i} + \frac{\dlt_{i+1}'}{1 - \dlt_{i+1}'}\U{\St{i+1}} \\  \aleq&  \pr{\dlt_{i+1}' + \frac{\dlt_{i+1}'}{1 - \dlt_{i+1}'}\cdot \pr{1 + \dlt_{i+1}'}} \U{\sc{\Stt{i}, F_i}} \pleq \frac{\dlt}{ i^2} \U{\sc{\Stt{i}, F_i}}.
  }
  Analogously, $\U{\LL} \pleq \U{\Stt{1}} $, $\Stt{1} - \LL \aleq \dlt \cdot \U{\LL}.  $

  By Lemma~\ref{lem:FindRCDD}, $\bet = \frac{1}{16\pr{1 + \alp}}$.
  Thus, the loop will terminate in $d = O\pr{\log n}$ iterations.
  Since $\TSE\pr{m, n, \dlt}$ and $\NSE\pr{n, \dlt}$ depends on $m, n $ nearly linearly and $\dlt\inv $ polynomially,
  the result follows by combining Theorem~\ref{thm:SparSchur} and Lemma~\ref{lem:FindRCDD} with the fact $\sum_{i=1}^{+\infty} \pr{1 - \bet}^{i-1} \poly{i^2 } = O(1) $.

\end{proof}

\begin{proof}[Proof of Lemma~\ref{lem:BBpleqo1Bap}]
  By~\eqref{eq:defLap},
  \eq{
    &\putmat{\sc{\Lap, F_1}, C_1, C_1, n}  \\
    =& \putmat{\sc{\Stt{1}, F_1}, C_1, C_1, n} + \putmat{\Stt{2} - \sc{\Stt{1}, F_1}, C_1, C_1, n} \\
     & + \sum_{i=2}^{d-1} \putmat{\Stt{i+1} - \sc{\Stt{i}, F_{i}}, C_{i}, C_{i}, n}  \\
    =& \Stt{2} + \sum_{i=2}^{d-1} \putmat{\Stt{i+1} - \sc{\Stt{i}, F_{i}}, C_{i}, C_{i}, n} .
  }
  Repeating this process gives us that
  \eq{
    &\putmat{\sc{\Lap, \cup_{j=1}^{i-1} F_j}, C_{i-1}, C_{i-1}, n}  \\
     =& \putmat{\Stt{i}, C_{i-1}, C_{i-1}, n } + \sum_{j=i}^{d-1} \putmat{\Stt{j+1} - \sc{\Stt{j}, F_{j}}, C_{j}, C_{j}, n},\ \forall i\in [d].
  }
  Thus,
  \eq{
    \Bap = \BB + \dlt_1\pr{\U{\Stt{1}} - \U{\LL}} + \sum_{i=1}^{d-1} \pr{\sum_{j=1}^{i+1}\dlt_j} \putmat{\U{\Stt{i+1}} - \U{\sc{\Stt{i}, F_i}}, C_i, C_i, n}.
  }
  By the definition of Schur complement chain, $\U{\Stt{1}} \pgeq \U{\LL} $, $\U{\Stt{i+1}} \pgeq \U{\sc{\Stt{i}, F_i}} $ and $\Stt{1} - \LL \aleq \dlt_1 \cdot \U{\LL} $, $\Stt{i+1} - \sc{\Stt{i}, F_i} \aleq \dlt_{i+1}\cdot \U{\sc{\Stt{i}, F_i}}.   $

  Combining with the condition $\sum_{i=1}^{d}\dlt_i\leq 1$, we have
  \eq{
    \BB \pleq \Bap \pleq \BB + \dlt_1^2 \U{\LL} + \sum_{i=1}^{d-1} \dlt_{i+1} \pr{\sum_{j=1}^{i+1} \dlt_{j}} \U{\sc{\Stt{i}, F_i}} \pleq 2\BB.
  }

\end{proof}

\begin{proof}[Proof of Lemma~\ref{lem:Mpt}]
    By the definition of $\Mpt{i, N}$, we have
    \eq{
        \iv{\Stt{i}_{F_i F_i}} - \Mpt{i, N} = \frac{1}{2}\sum_{k=N}^{+\infty}\pr{\II - \frac{1}{2}\Dap_{F_i F_i}\inv\Stt{i}_{F_i F_i}}^k\Dap_{F_i F_i}\inv = \pr{\II - \frac{1}{2}\Dap_{F_i F_i}\inv\Stt{i}_{F_i F_i}}^N \iv{\Stt{i}_{F_i F_i}}.
    }
    By the Gershgorin circle theorem, the modulus of the eigenvalues of $\II - \frac{1}{2}\Dap_{F_i F_i}\inv\Stt{i}_{F_i F_i}$ are no greater than $\frac{2 + \alp}{2\pr{1 + \alp}}$.
    Then, the modulus of the eigenvalues of $\pr{\II - \frac{1}{2}\Dap_{F_i F_i}\inv\Stt{i}_{F_i F_i}}^N$ are no greater than $\pr{\frac{2 + \alp}{2\pr{1 + \alp}}}^N  $.
    Thus, we have $\II - \pr{\II - \frac{1}{2}\Dap_{F_i F_i}\inv\Stt{i}_{F_i F_i}}^N$ is nonsingular.
    It follows that \eq{\Mpt{i, N} =  \pr{\II - \pr{\II - \frac{1}{2}\Dap_{F_i F_i}\inv\Stt{i}_{F_i F_i}}^N} \iv{\Stt{i}_{F_i F_i}} } is nonsingular.

    Then, we have
    \eq{
        &\ni{\iv{\Utilt{i, N}} - \iv{\Utilt{i, \infty}}} \\
        =& \ni{\pr{\iv{\Stt{i}_{F_i F_i}} - \Mpt{i, N}}\Stt{i}_{F_i C_i}}
        = \frac{1}{2}\ni{\sum_{k=N}^{+\infty}\pr{\II - \frac{1}{2}\Dap_{F_i F_i}\inv\Stt{i}_{F_i F_i}}^k\Dap_{F_i F_i}\inv\Stt{i}_{F_i C_i} }  \\
        \leq& \frac{1}{2}\sum_{k=N}^{+\infty}\ni{\II - \frac{1}{2}\Dap_{F_i F_i}\inv\Stt{i}_{F_i F_i}}^k\ni{\Dap_{F_i F_i}\inv\Stt{i}_{F_i C_i}  }
        \leq \frac{1}{2}\sum_{k=N}^{+\infty} \pr{\frac{2 + \alp}{2\pr{1 + \alp}}}^k = \frac{\pr{1 + \alp}}{\alp}  \pr{\frac{2 + \alp}{2\pr{1 + \alp}}}^N.
    }
    Analogously, we have $\no{\iv{\Ltilt{i, N}} - \iv{\Ltilt{i, \infty}}} \leq \frac{\pr{1 + \alp}}{\alp}\pr{\frac{2 + \alp}{2\pr{1 + \alp}}}^N  $ and  \\
    $
        \ni{\iv{\Stt{i}_{F_i F_i}} - \Mpt{i, N}} \leq \frac{\pr{1 + \alp}}{\alp}\pr{\frac{2 + \alp}{2\pr{1 + \alp}}}^N \ni{\Dap_{F_i F_i}\inv}.
    $

\end{proof}

\end{appendices}